\documentclass[a4paper]{article}
\usepackage{cite}
\usepackage{amsmath,amssymb,amsfonts,amsthm}
\usepackage{hyperref}
\usepackage{stmaryrd}
\usepackage{xspace}
\usepackage{proof}
\usepackage{mathtools}
\usepackage{textcomp}
\usepackage{xcolor}
\RequirePackage{xparse}
\usepackage{cleveref}
\usepackage{pdflscape}
\usepackage{authblk}
\usepackage[margin=1.2in]{geometry}

\newtheorem{rem}{Remark}

\newtheorem{definition}{Definition}
\newtheorem{example}{Example}
\newtheorem{lemma}{Lemma}
\newtheorem{proposition}{Proposition}
\newtheorem{cor}{Corollary}
\newtheorem{theorem}{Theorem}
\newcommand{\qedeqn}{\tag*{$\square$}}

\newcommand{\BI}{BI\xspace}
\newcommand{\CBI}{\CPSL}
\newcommand{\PSL}{PSL\xspace}
\newcommand{\CPSL}{CSL\xspace}
\newcommand{\defsym}{\mathrel{:=}}

\newcommand{\poly}{\mathrm{poly}}
\newcommand{\eX}{\mathbf{X}}
\newcommand{\eY}{\mathbf{Y}}
\newcommand{\eZ}{\mathbf{Z}}
\newcommand{\eT}{\mathbf{T}}
\newcommand{\eC}{\mathbf{C}}

\newcommand{\U}[1]{\mathbf{U}(#1)}
\newcommand{\indep}{\perp \!\!\! \perp}
\newcommand{\OMIT}[1]{}

\newcommand{\NN}{\mathbb{N}}
\newcommand{\BB}{\mathbb{B}}

\newcommand{\RRunit}{\mathbb{R}_{[0,1]}}
\newcommand{\dom}{\mathsf{dom}}

\newcommand{\bnf}{::=}
\newcommand{\midd}{\; \mbox{\Large{$\mid$}}\;}
\newcommand{\restr}[1]{_{#1}}

\newcommand{\distr}[1]{\mathbb{D}(#1)}
\newcommand{\distrone}{\mathcal{D}}
\newcommand{\distrtwo}{\mathcal{E}}

\newcommand{\polyone}{p}

\newcommand{\polytwo}{q}

\newcommand{\nat}{n}

\newcommand{\tyset}{\mathsf{T}}
\newcommand{\tyone}{\tau}
\newcommand{\tytwo}{\sigma}
\newcommand{\mem}{m}
\newcommand{\memone}{\mem}
\newcommand{\memtwo}{l}
\newcommand{\memthree}{w}
\newcommand{\memfour}{u}
\newcommand{\stng}[1]{\mathsf{Str}[#1]}
\newcommand{\bool}{\mathsf{Bool}}
\newcommand{\sem}[2]{\llbracket #1\rrbracket_{#2}}
\newcommand{\ksem}[2]{\llbracket #1\rrbracket^{\mathtt K}_{#2}}
\newcommand{\parsem}[2]{\llparenthesis #1\rrparenthesis_{#2}}

\newcommand{\binjudg}[2]{#1\vdash#2}
\newcommand{\ternjudg}[3]{#1\vdash#2:#3}

\newcommand{\rvarset}{\mathcal{V}}
\newcommand{\rvarsetone}{\mathsf{S}}

\newcommand{\renvone}{\Delta}

\newcommand{\renvtwo}{\Xi}
\newcommand{\renvthree}{\Theta}
\newcommand{\renvfour}{\Omega}
\newcommand{\join}[2]{#1\Join #2}

\newcommand{\rvarone}{r}
\newcommand{\rvartwo}{s}
\newcommand{\prgone}{\mathtt{P}}
\newcommand{\prgtwo}{\mathtt{R}}

\newcommand{\rand}{\mathit{rnd}}
\newcommand{\setzero}{\mathit{setzero}}
\newcommand{\head}{\mathit{head}}
\newcommand{\tail}{\mathit{tail}}

\newcommand{\exor}{\mathit{xor}}
\newcommand{\concat}{\mathit{concat}}

\newcommand{\rsone}{\mathcal{R}}
\newcommand{\rstwo}{\mathcal{Q}}
\newcommand{\rsthree}{\mathcal{T}}
\newcommand{\comp}{\otimes}
\newcommand{\stset}{\mathcal{S}}
\newcommand{\emst}[1]{\mathbf{0}_{#1}}
\newcommand{\stone}{s}
\newcommand{\sttwo}{r}
\newcommand{\stthree}{t}
\newcommand{\stfour}{u}
\newcommand{\stfive}{v}
\newcommand{\es}[1]{\mathsf{ES}({#1})}

\newcommand{\rvone}{\mathbf{R}}
\newcommand{\rvtwo}{\mathbf{Q}}

\newcommand{\ind}{\approx} 
\newcommand{\ext}{\sqsubseteq} 

\newcommand{\tyf}[2]{{(#1)^{#2}}}
\newcommand{\fset}{\Phi}
\newcommand{\exfset}{{\Phi_{\mathsf{ex}}}}
\newcommand{\apfset}{{\Phi_{\mathsf{ap}}}}
\newcommand{\fone}{\phi}
\newcommand{\ftwo}{\psi}
\newcommand{\fthree}{\xi}
\newcommand{\ffour}{\eta}
\newcommand{\apone}{A}
\newcommand{\apset}{\mathcal{AP}}
\newcommand{\ftrue}{\top}
\newcommand{\ffalse}{\bot}
\newcommand{\mw}{-\kern-.6em\raisebox{-.659ex}{*}\ }
\newcommand{\sep}{*}
\DeclareMathOperator*{\Sep}{\text{\Large $\ast$}}
\newcommand{\arr}{\rightarrow}

\newcommand{\sr}[2]{#1\models #2}
\newcommand{\prenv}[3]{#2\vdash^{#1} #3}

\newcounter{numberone}
\newenvironment{varenumerate}
{
	\begin{list}{\arabic{numberone}.}
		{
			\usecounter{numberone}
			\setlength{\itemsep}{0pt}
			\setlength{\topsep}{0pt}
			\setlength{\parsep}{0pt}
			\setlength{\partopsep}{0pt}
			\setlength{\leftmargin}{15pt}
			\setlength{\rightmargin}{0pt}
			\setlength{\itemindent}{0pt}
			\setlength{\labelsep}{5pt}
			\setlength{\labelwidth}{15pt}
	}}
	{
	\end{list}
}

\newcommand{\deone}{d}
\newcommand{\detwo}{c}
\newcommand{\reone}{e}
\newcommand{\retwo}{g}
\newcommand{\rethree}{h}
\newcommand{\funset}{\mathcal{F}}
\newcommand{\dfset}{\mathcal{DF}}
\newcommand{\rfset}{\mathcal{RF}}

\newcommand{\rfunone}{f}
\newcommand{\rfuntwo}{g}

\newcommand{\tof}[1]{\mathit{TypeOf}(#1)}
\newcommand{\raof}[1]{\mathit{AlgoOf}(#1)}
\newcommand{\daof}[1]{\mathit{DetAlgoOf}(#1)}
\newcommand{\juone}{\pi}
\newcommand{\jutwo}{\theta}

\newcommand{\deriv}{\;\triangleright\;}
\newcommand{\ud}[1]{\mathbf{U}(#1)}
	
\newcommand{\eq}[2]{\mathbf{EQ}(#1,#2)}
\newcommand{\espl}[2]{\mathbf{IS}(#1,#2)}
\newcommand{\indp}[2]{\mathbf{CI}(#1,#2)}

\newcommand{\df}[1]{\top}

\newcommand{\fv}[1]{\mathsf{FV}(#1)}
\newcommand{\mv}[1]{\mathsf{MV}(#1)}

\newcommand{\pskip}{\mathtt{skip}}
\newcommand{\ass}[2]{#1\leftarrow #2}

\newcommand{\seq}[2]{#1;#2}

\newcommand{\ifr}[3]{\mathtt{if}\;#1\;\mathtt{then}\;#2\;\mathtt{else}\;#3}

\newcommand{\hot}[3]{\left\{#1\right\}\;#2\;\left\{#3\right\}}
\newcommand{\vhot}[3]{\left\{#1\right\}\\#2\\\left\{#3\right\}}
\newcommand{\hotj}[3]{\vdash\hot{#1}{#2}{#3}}

\newcommand{\semic}{;}

\newcommand{\detsem}[2]{\sem{#1}{#2}^{\mathtt d}}
\newcommand{\bind}[3][]{\text{bind}_{#1}(#2, #3)}
\newcommand{\gbind}[3][]{{\text{bind}_{#1}\left(#2, #3\right)}}
\newcommand{\unit}[2][]{{\text{unit}_{#1}(#2)}}
\newcommand{\condee}[4][]{{#2_{{#1}|#3 = #4}}}
\newcommand{\conv}[1]{\mathrel{\oplus_{#1}}}

\newcommand{\Prob}[1][]{{\text{Pr}_{#1}}}

\newcommand{\supp}{{\mathsf{supp}}}
\newcommand{\kleisli}[1]{#1^\dag}
\newcommand{\tensprod}{\mathrel{\otimes}}

\newcommand{\distone}{{\mathcal{A}}}
\newcommand{\disttwo}{{\mathcal{B}}}

\newcommand{\negl}{{\mathit{negl}}}

\newcommand{\monone}{M}
\newcommand{\mulone}{\bullet}
\newcommand{\id}{e}
\newcommand{\poone}{\leq}
\newcommand{\elone}{x}
\newcommand{\eltwo}{y}
\newcommand{\elthree}{z}

\newcommand{\unif}[1]{\mathsf{unif}(#1)}
\newcommand{\defined}{\downarrow}

\newcommand{\RSKIP}{\mathsf{Skip}}
\newcommand{\RSEQ}{\mathsf{Seq}}
\newcommand{\RASS}{\mathsf{Assn}}
\newcommand{\DASS}{\mathsf{DAssn}}

\newcommand{\SRASS}{\mathsf{SRAssn}}
\newcommand{\SDASS}{\mathsf{SDAssn}}

\newcommand{\RCONDCM}{\mathsf{RCond}}

\newcommand{\RWEAK}{\mathsf{Weak}}
\newcommand{\RCONST}{\mathsf{Const}}
\newcommand{\RFRAME}{\mathsf{Frame}}
\newcommand{\RRESTR}{\mathsf{Restr}}

\newcommand{\OTP}{\mathtt{OTP}}
\newcommand{\POTP}{\mathtt{POTP}}
\newcommand{\XOR}{\mathtt{XOR}}

\newenvironment{proofcases}
{\begin{list}{\labelitemi}{
 \setlength{\itemsep}{0pt}
 \setlength{\topsep}{0pt}
 \setlength{\parsep}{0pt}
 \setlength{\partopsep}{0pt}
 \setlength{\leftmargin}{10pt}
 \setlength{\rightmargin}{0pt}
 \setlength{\itemindent}{0pt}
 \setlength{\labelsep}{5pt}
 \setlength{\labelwidth}{10pt}}}
{\end{list}}
\newcommand{\proofcase}[1][]{\item[-] (#1).}

\newcommand{\commentout}[1]{}
\newcommand{\revision}[1]{{#1}}

\author[1,2]{Ugo Dal Lago \thanks{This author was partially supported by the MUR FARE project "Compositional and Effectful Program Distances", CAFFEINE}}
\author[2,3]{Davide Davoli \thanks{This author was supported by the French ``Agence Nationale de la
Recherche'' through the project UCA DS4H ANR-17-EURE-0004.}}
\author[4]{Bruce M. Kapron\thanks{This author was supported in part by NSERC RGPIN 2021-02481 and by a Visiting Fellowship  at the Institute of Advanced Studies of the Alma Mater Studiorum - University of Bologna.}}
\affil[1]{Univerist\`a di Bologna}
\affil[2]{Inria}
\affil[3]{Université Côte d'Azur}
\affil[4]{University of Victoria}

\begin{document}

\title{On Separation Logic, Computational Independence, and 
  Pseudorandomness (Extended Version)}

\maketitle

\begin{abstract}
Separation logic is a substructural logic which has proved to have numerous
and fruitful applications to the verification of programs working on dynamic data structures. Recently, 
Barthe, Hsu and Liao have proposed a new way of giving semantics to separation 
logic formulas in which separating conjunction is interpreted in 
terms of probabilistic independence. The latter is taken in its exact form, i.e., two 
events are independent if and only if the joint probability is the product of the probabilities
of the two events. There is indeed a literature on weaker notions of independence which
are \emph{computational} in nature, i.e. independence holds only against efficient adversaries and modulo a 
negligible probability of success. The aim of this work is to explore the 
nature of computational independence in a cryptographic scenario, in view of the aforementioned
advances in separation logic.
We show on the one hand that the semantics of separation logic can be adapted so as to
account for complexity bounded adversaries, and on the other hand that the obtained logical system
is useful for writing simple and compact proofs of standard cryptographic results in which the
adversary remains hidden. Remarkably, this allows for a fruitful interplay between independence 
and pseudorandomness, itself a crucial notion in cryptography.
\end{abstract}

\section{Introduction}
Cryptographic primitives and protocols can be defined and analyzed using two 
distinct models, i.e. the \emph{computational}
model~\cite{GoldwasserMicali84,GoldwasserMicali88,Yao} and the 
\emph{symbolic} model~\cite{NeedhamSchroeder,DolevYao}.
While in the latter adversaries have unlimited computational power, 
in the former they are assumed to be 
efficient (i.e. polytime) randomized algorithms. In the last forty years, 
research concerning the interactions between logic and cryptography has been intense, 
but while the symbolic model, given its more qualitative 
nature, has been the object of many systematic
investigations~\cite{BorrowsAbadiNeedham,Datta,Abadi,Millen,Blanchet}, the same 
cannot be said about the computational model, in relation to which the 
proposals for logical systems are more scattered and mostly recent (see 
Section~\ref{sec:related} for a more in-depth discussion on related work).

Proposals concerning logical systems that capture \emph{computational} 
cryptographic 
reasoning are bound to face, at least in principle, some 
challenges. The first is certainly the fact 
that adversaries have limited complexity, and can have a non-null, albeit 
negligible, probability of success. All that must be somehow reflected in 
the logical languages used, which are often not standard. Similarly, the 
underlying notion of equivalence between programs cannot be the standard, exact, 
one, but must be relaxed to so-called computational indistinguishability, the 
logical nature of which has been studied from different viewpoints~\cite{IK06,BDKL10}. Finally, 
one should mention pseudorandomness, the natural computational counterpart to
randomness, whose logical status has recently been scrutinized from the perspective
of descriptive complexity~\cite{DT23}.

There is yet another aspect of the computational model which is deeply rooted in 
probability theory, and which is known to be fruitful in cryptography. Since 
the work by Shannon, indeed, it is well known that the perfect security of a 
symmetric cipher can be formulated very elegantly as the \emph{independence} between 
the random variable modeling the message and the one modeling the ciphertext. On the one 
hand, it is certainly worth noting that since Yao's pioneering work on the 
computational model, strong links are known between the notion of 
computational security for symmetric encryption schemes and appropriate variations of the 
concept of independence \cite{Yao}, all this thus generalizing Shannon's definition to the setting of
complexity-bounded adversaries. Such a view has remained outside mainstream 
cryptography, with interesting but sporadic contributions~\cite{MRS88,Fay14,HK15}. On the other 
hand, it must be said that a logical account of the notion of \emph{statistical}
independence has indeed been proposed very recently by Barthe, Hsu and Liao~\cite{PSL},
based on separation logic. In this approach, dubbed \emph{probabilistic separation 
logic} (PSL) in the following, separating conjunction is interpreted in the sense of 
independence and not in the traditional sense of dynamic data structures. 
Probabilistic separation logic is by its nature information-theoretic, 
therefore not able to take into account computational notions of independence. 
Nonetheless, not only Shannon's proof, but also much more complex ones, become very 
simple and straightforward, being based on some basic principles regarding 
probabilistic independence.

It is therefore natural to ask whether and to what extent the approach of 
\cite{PSL}
 can be adapted to \emph{computational} notions of 
independence. This is exactly the aim of this paper.

\subsection{Contributions and Outline}
The first thing we do, after briefly motivating our work through a couple of 
examples, is to introduce in 
Section~\ref{sec:programs} a minimal imperative programming language where all 
programs can be evaluated in polytime \emph{by construction}. The language consists only
of assignments, sequential composition, and conditional constructs.
Notably, randomization is captured through function symbols. We take this approach first of all 
because it allows the treatment of some interesting examples, and secondly 
because the possibility of capturing programs with loops turns out to be by no 
means trivial especially from the point of view of managing the adversary's advantage.

We then proceed, in Section~\ref{sec:logic}, to a study of how Barthe et 
al.'s probabilistic separation logic can be adapted so as to capture 
computational, rather than information-theoretic, independence.  
The logic stays the same syntactically, while the underlying semantics has to 
be modified in a crucial way. More specifically, the partial Kripke resource 
monoid from~\cite{Pym04} has to be modified in its order-theoretic part, allowing 
computationally indistinguishable distributions to be compared, and thus 
implicitly accounting for an approximate notion of independence. These 
apparently innocuous modifications have very deep consequences for the 
interpretation of separating conjunction. We exclude from formulas both forms of implication,
whose presence would complicate considerably the underlying metatheory, and which are not
particularly useful in proofs. Remarkably, our novel semantic interpretation naturally induces atomic formulas capturing computationally independent pairs of expressions and pseudorandom expressions, this way dealing with fundamental concepts in cryptography in a simple and natural way.

The next section, i.e. Section~\ref{sec:inference}, instead, concerns the use of our 
separation logic as a program logic. More specifically, we study how to 
define Hoare triples for the language introduced in Section~\ref{sec:programs}, and
how to define sound inference rules for such triples. Here, what is surprising is that the rules are not too different from those presented in~\cite{PSL}. In other words, computational notions can be treated similarly to the corresponding information-theoretic notions. We apply the introduced formal system to a couple of simple examples of cryptographic proofs, namely the security of the encryption scheme induced by a pseudorandom generator, and the proof that suitably composing pseudorandom generators preserve pseudorandomness, at the same time allowing for an increase in the expansion factor.

Finally, Section~\ref{sec:compindep} is concerned with motivating computational independence, and with studying its role in cryptographic proofs. As mentioned above, such an approach was already proposed by Yao \cite{Yao} in some of the earliest work on encryption security, although his work used the language of computational entropy. An explicit notion of computational independence was introduced and studied in \cite{Fay14}.  We formulate a notion of computational independence which is in some ways more closely related to the computational semantics of the separating conjunction, and prove that it is equivalent to that of \cite{Fay14}. Finally, we solve an open problem posed in that paper, namely giving a characterization of semantic security for private-key encryption against \emph{uniform} adversaries~\cite{Goldreich}. Overall, the results of this section further support the naturalness of computational independence in a cryptographic setting as well as indicating its utility as an alternate approach to characterizing cryptographic security. This also argues for the naturalness and utility of cryptographic separation logic as a formal counterpart to this computational approach.

\renewcommand{\df}[1]{\mathbf{D}(#1)}
\section{A Cryptographic Viewpoint on (Probabilistic) Separation Logic}\label{sec:viewpoint}
In this section, we give, without going into technical details the 
introduction of which is deferred to the next section, a simple 
account of separation logic, the probabilistic interpretation of 
the latter proposed by Barthe et al., and the possibility of employing
it to derive proofs in the so-called computational model.

Consider the following simple probabilistic program, call it $\OTP$:
\begin{align*}
  &\ass k \rand();\\
  &\ass c \exor(m,k)
\end{align*}
This program assigns to the variable $k$, a binary 
string of length $n$, one of the $2^n$ possible values it can take, each with 
the same probability. It then uses this value, seen as a key, to construct a 
ciphertext $c$ starting from a message $m$ through the exclusive or operation, 
thus implementing the well-known Vernam's one-time pad, an encryption scheme that is established as 
secure in the perfect sense~\cite{KatzLindell}. This last property can be expressed in various 
ways, e.g. by prescribing that the random variables $\mathcal{C}$ and 
$\mathcal{M}$ which model the values of $c$ and $m$ after program execution, 
respectively, are \emph{independent}.

Would it be possible to (re)prove the security of such an encryption scheme using 
Hoare's logic? As already mentioned, Barthe et al. gave a positive answer to this question a few years 
ago, based on separation logic~\cite{PSL}. In their framework, one can derive the 
following triple
$$
\hot{\df m}\OTP{\df m \sep \ud c}
$$
This tells us that, whenever messages are distributed according to \emph{any}
distribution (this is captured by the precondition $\df m$), the ciphertexts
will be uniformly distributed, $\ud c$, and moreover (and most importantly!) 
their distribution will be statistically independent from that on 
messages. This latter condition is captured by the binary logical operator 
$\sep$, the so-called \emph{separating conjunction}. 

The separating conjunction is a connective of \emph{separation 
logic}~\cite{Reynolds00}, which is often interpreted in such a 
way as to enable reasoning about pointer-manipulating programs: the elements 
of the domain are \emph{heaps} representing the state of memory and 
the truth of a formula $\phi * \psi$ in a store $s$ indicates that 
 $s$ may be split into two stores, the first satisfying $\phi$ and the 
 second satisfying $\psi$. In Barthe et al.'s work, instead, there is no heap,
 each variable contains a piece of data, and formulas are interpreted 
over \emph{distributions} of states. In that setting, it is natural for 
separating conjunction to model \emph{independence}: a distribution $s$ satisfying a formula 
$\phi * \psi$ captures the fact that $s$ can be written as the \emph{tensor 
product} of two distributions, the first satisfying $\phi$ and the 
second satisfying $\psi$. In the program above, then, $\df m \sep \ud c$ 
correctly captures the desired property: the distributions of states over 
$k,m,c$ one obtains after executing $\OTP$ can be written as the tensor product 
of two distributions, the value of $c$ being uniformly 
distributed in the second distribution. As such, then, separation logic 
correctly captures the perfect security of the OTP.

The question that arises naturally, at this point, is whether this way of 
reasoning is also compatible with the computational model, whereas, as 
is well known, adversaries are complexity bounded 
and their probability of success can be strictly positive, although negligible. 
Is it possible, for example, to prove something similar to the aforementioned 
independence result for the following simple variation of Vernam's cipher?
\begin{align*}
  &\ass \rvarone  \rand();\\
  &\ass c  \exor(m,g(\rvarone))
\end{align*}
Here, we first pass a random seed $r$ to a pseudorandom generator $g$, stretching it 
into a stream of bits that, although not random, looks so to any PPT adversary. 
Then, we proceed as in Vernam's cipher. \revision{In this way, we are able to use a key of length
$n$ to encrypt messages of greater length.}

The first thing to observe is that it is simply \emph{not} true that for any \revision{distribution on messages, the 
distributions on messages and ciphertexts are \revision{statistically} independent.
In particular, far a distribution on messages having full support, since there are fewer keys than messages, we can find a ciphertext $c$ which has nonzero probability and a message $m$ such that the joint probability of $c$ and $m$ is zero. This means independence does not hold.
}
As observed by Yao~\cite{Yao}, and later studied by Fay~\cite{Fay14}, 
however, the aforementioned random variables are independent in a 
\emph{computational} sense, i.e. their joint distribution, even if impossible to be written as the tensor 
product of two distributions, looks so, once again, to the eyes of efficient
adversaries.

This work can be seen as the first attempt to go \emph{in the direction} traced by
Yao and Fay \emph{starting from} Barthe et al.'s PSL. In particular, the latter
will be adapted to the fact that programs and adversaries are 
complexity-bounded and that the notion of randomness
is computational, 
i.e. equivalent to \emph{pseudorandomness}. Similarly, for
separating conjunction and independence: everything holds in an approximate 
sense, the approximation being invisible to the eyes of efficient adversaries
having nonnegligible probability of success. 

While this adaptation may apparently seem straightforward, it brings with it a series 
of complications. First of all, the class of programs to which 
this reasoning can be applied must be drastically restricted, dealing with 
conditionals and sequential compositions, but without loops. The very notion of 
formula must be narrowed, getting rid of logical implication and restricting 
attention to conjunctive formulas. \revision{Furthermore, in order to make the splitting of the program unique, we will require formulas to be labeled with type environments. These changes are not motivated by semantic reasons: in fact, BI \emph{can} be given a complexity-conscious semantics in the form of a partial resource monoid, therefore following a well-trodden path. The point is that in BI and its standard semantics proving the soundness of some crucial rules, and in particular the so-called frame rule, turns out to be very difficult compared to doing so for PSL. It was therefore natural to move towards a narrower definitional framework.} That being said, simple examples like the one 
above can be shown to have the desired independence properties.

A question naturally arises at this point: is computational independence related in any way to mainstream cryptographic notions? Fortunately, the answer is positive: security
of symmetric-key ciphertext can be faithfully captured this way,
making our program logic useful.  
This represents the ideal way to complete the picture, and Section~\ref{sec:compindep} 
is devoted to proving this result.


\renewcommand{\df}[1]{\top}

\section{Programs, and their Semantics}\label{sec:programs}
\revision{
  In this section, we define programs and their interpretation
  as transformers of \emph{efficiently samplable} families
  of distributions~\cite{Goldreich,KatzLindell}. While doing this,
  we also introduce other standard concepts such as sized types and
  with the aim of describing their interplay and their notation.
}

We start by introducing basic notions from probability theory in 
\Cref{subsubsec:probanddist}. In \Cref{subsubsec:typesandstores},
we define types, environments and stores. Later, in \Cref{sec:expressions} and 
in \Cref{sec:subprograms} we define the language of programs together with 
its semantics. In \Cref{subsec:kozeneq}, we show that the semantics of our programs
--- when interpreted on ordinary distributions --- is equivalent to that of
Kozen~\cite{Kozen}.
In section \Cref{subsection:crypto}, we introduce some basics concepts of computational
cryprography including \emph{computational indistinguishability}. Finally,
\Cref{subsec:cryptoprog} is devoted to showing cryptographical properties of
our programs.

\subsection{Preliminaries}
\label{sec:preliminaries}


\subsubsection{Probabilities and Distributions}
\label{subsubsec:probanddist}

In the following, we restrict our analysis to
probability distributions with finite support.
More precisely, each such object is a function
$\distrone: X \to \RRunit$ that satisfies the following properties:
\begin{itemize}
\item Its support 
  \(
    \supp(\distrone) \defsym \{x \in X \mid \distrone(x)> 0\}
  \)
  is finite.
\item $\sum_{x\in\supp(\distrone)} \distrone(x)=1$.
\end{itemize}
\noindent
We write $\distr X$ for the set of all distributions over $X$.

The tensor product of two distributions
$\distrone \in \distr X, \distrtwo \in \distr Y$ is another distribution, this
time over ($X \times Y$) that samples independently from $\distrone$ and 
$\distrtwo$, i.e.,
$(\distrone \tensprod \distrtwo)\defsym (x, y) \mapsto
\distrone(x) \cdot \distrtwo(y)$.

Probability distributions with finite support form a monad \cite{Giry82}, where 
the unit function $\unit[X] \cdot: X\to \distr X$ associates to
its argument $x\in X$ the Dirac distribution centered on $x$ and the 
multiplication $\bind[X] \cdot \cdot: \distr {X} \times  (X \to \distr Z) \to 
\distr Z$ sequences $\distrone \in \distr X$ with a continuation
$f: X \to \distr Z$.
\commentout{Formally:\\[1.5ex]
\resizebox{0.5\textwidth}{!}{$
\begin{aligned}
  \unit[X] {x}(y) &\defsym \begin{cases}
                          1 &\text{if } y = x\\
                          0 & \text{otherwise}
                        \end{cases} & 
  \bind[X] {\distrone}{f} (z)&\defsym \sum_{x \in X} \distrone(x) \cdot f(x)(z)
\end{aligned}
$}\\[1.5ex]%
}
This property induces the existence of the Kleisli extension
of each function $f: X \to \distr Y$: this is a function
$\kleisli f: \distr X \to \distr Y$ such that
$\kleisli f(\distrone) = \bind[X] \distrone f$.
Given a distribution $\distrone \in \distr X$ and a function $f: X \to Y$, we write
$f(\distrone)$ to denote $y \mapsto \sum_{x \in f^{-1}(y)}\distrone(x)$.

\subsubsection{Types}
\label{subsubsec:typesandstores}

In this section, we define types and environments, together with
their interpretation as \emph{families} of distributions.
Programs and stores are indeed typed to enforce size bounds to variables, depending 
on the value of the so-called \emph{security parameter}.
For instance, if the type associated to the variable $\rvarone$ by
the environment $\renvone$ is $\stng{2n}$,
the size of the binary string associated to
$\renvone$ by a store of type $\rvarone$ must be
\emph{exactly} $2n$, i.e.
two times the length of the security parameter. 
By allowing only types whose size is polynomial in
the size of the security parameter, we ensure that
the evaluation of \emph{well-typed} programs on \emph{well-typed} stores
is polytime with respect to the size of the security parameter.

\emph{Univariate positive polynomials} are indicated by metavariables
like $\polyone$ and $\polytwo$; we write $\polyone(n)$ for the natural number 
obtained by evaluating the polynomial $\polyone$ on the argument $\nat\in\NN$.
\emph{Types} 
are defined from the following grammar:
$$
\tyone\bnf\stng{\polyone}\midd\bool.
$$
where $p$ is a univariate polynomial. Let $\tyset$ be the set of all types. 
Given $\tyone$ and $\nat \in \NN$, 
$\detsem{\stng{\polyone}}{\nat}=\BB^{\polyone(\nat)}$,
$\detsem{\bool}{\nat}=\BB$,
and $\BB=\{0,1\}$ is the set of binary digits.
%

The set of program variables is $\rvarset$, and it is
ranged over by metavariables like $\rvarone,\rvartwo$.
An \emph{environment} $\renvone: \rvarset \rightharpoonup \tyset$
is a partial function with finite domain indicated as $\dom(\renvone)$.
We write $\varepsilon$ for the empty environment, and we
also adopt standard abbreviations for environments in the form
of lists of assignments of types to variables, e.g.
$\rvarone_1:\tyone_1, \ldots,\rvarone_k:\tyone_k$.
%
We introduce a partial order relation between
environments: $\renvtwo \ext \renvone$ holds
whenever $\dom(\renvtwo)\subseteq \dom(\renvone)$
and for every  $\rvarone \in \dom(\renvtwo)$
we have $\renvtwo(\rvarone)=\renvone(\rvarone)$.
The (disjoint) union of two environments $\renvone, \renvtwo$ is
indicated as $\join \renvone \renvtwo$ and is
defined \emph{only} if
$\dom(\renvone)\cap\dom(\renvtwo)=\emptyset$;
when this happens, it is defined in the standard way.
\commentout{$$
(\join \renvone \renvtwo)(\rvarone) \defsym
\begin{cases}
  \renvone (\rvarone) & \text{if }\rvarone \in \dom(\renvone)\\
  \renvtwo (\rvarone) & \text{if }\rvarone \in \dom(\renvtwo).
\end{cases}
$$
}

It is now time to turn to the interpretation of types, which will be built 
around the notion of an efficiently samplable distribution family.
\revision{This is a standard notion in complexity-based cryptography (see Subsection~\ref{subsection:crypto}) and is captured by the following definition}:
\begin{definition}
\label{def:samp}
  A family of distributions $\{\distrone_\nat\}_{\nat\in \NN}$
  is \emph{efficiently samplable} when
  there is \emph{one} probabilistic polytime algorithm
  that outputs samples distributed according to $\distrone_\nat$ on input
  $1^\nat$. In this case, we write $\es {\{\distrone_\nat\}_{\nat \in \NN}}$.
  $\hfill\qed$
\end{definition}

The semantics of a type $\tyone$,
which we call $\sem{\tyone}{}$ is the class 
of (efficiently samplable) families
$\{\distrone_\nat\}_{\nat\in \NN}$ such that
$\distrone_\nat\in\distr{\detsem{\tyone}{\nat}}$ for 
every $\nat \in \NN$. Formally:
$$
\sem{\tyone}{}\defsym\{\{\distrone_\nat\}_{\nat\in\NN}\mid
                \es {\{\distrone_\nat\}_{\nat\in \NN}} \land \forall\nat \in \NN. \distrone_\nat\in\sem{\tyone}{\nat}\}.
$$
In the following, we will often abbreviate $\{\distrone_\nat\}_{\nat \in \NN}$
as just $\distrone$.

We can extend the semantics above from types to environments.
An environment  $\renvone$ is interpreted as a set of \emph{distribution ensembles}
in the following way:
\begin{align*}
\sem{\renvone}{}&\defsym\{\{\rsone_\nat\}_{\nat}\mid
                  \es \rsone \land \forall\nat \in \NN. \rsone_\nat\in\sem{\renvone}{\nat}\}
\end{align*}
where $\sem \renvone \nat = \distr {\detsem \renvone \nat}$, and
$\detsem \renvone \nat$ is the set of
well-typed maps from variables to values, i.e.,
$$
\detsem \renvone \nat \defsym \{ \mem \in \left(\BB^*\right)^{\dom(\renvone)} \mid \forall \rvarone \in {\dom(\renvone)}. \mem(\rvarone) \in \detsem {\renvone(\rvarone)}\nat\}.
$$

\begin{example}
  If $\renvone = \rvarone: \stng \nat, \rvartwo: \stng {2\nat}$,
  then the semantics of $\renvone$ is the set of all distribution ensembles
  $\{\rsone_\nat\}_{\nat\in \NN}$ such that, for each $\nat$,
  $\rsone_\nat$ is a distribution in $\sem \renvone \nat$.
  This means that every member of this family  $\rsone_\nat$ is a distribution of  
  functions $\memone:\{\rvarone, \rvartwo\}\to{\BB^*}$ such that
  $\memone(\rvarone)$ is a binary string of length $\nat$ and 
  $\memone(\rvartwo)$ is a binary string of length $2\nat$.
  We also require $\{\rsone_\nat\}_{\nat\in \NN}$ to be
  \emph{effectively samplable}, so there must be a polytime
  probabilistic algorithm that, on input $1^n$, produces
  samples according to $\rsone_\nat$. \hfill\qed
\end{example}

In our logic, to state that some expression of type $\tyone$
is pseudorandom, we compare it against the family of uniform distributions
of strings of that type, and we impose that the advantage of any
polytime distinguisher for the two distributions is a negligible function (\Cref{def:nf}).
\revision{
  In turn, the (effectively samplable) family of 
uniform distributions strings of a given type $\tyone$ is defined as follows:
  \begin{align*}
    \unif{\bool}_\nat (b) &\defsym \frac 1 2 & &\text{with } b \in \BB\\
    \unif{\stng {p(\nat)}}_\nat(x) &\defsym {\frac 1 {2^{p(\nat)}}} & &\text{with } x \in \BB^{p(n)}.
  \end{align*}
}

\subsubsection{Stores}
Intuitively, stores can be defined just as elements of 
$\sem{\renvone}{}$. But, since we want to reason about 
stores with \emph{different} domains, we also make the
domain $\renvone$ explicit in the notation by
defining \emph{stores} as pairs in the form $\tyf \rsone \renvone$
where $\rsone\in\sem{\renvone}{}$.
We write $\emst{}$ for the store 
$\tyf{\{\{\varepsilon^1\}\}_{\nat \in \NN}}{\varepsilon}$
where $\{\varepsilon^1\}$ is the distribution assigning probability $1$ to 
$\varepsilon$, namely the only element of $\sem{\varepsilon}{\nat}$. The 
set 
$\stset$ stands for the set of all stores.
In the following, we use metavariables like 
$\stone,\sttwo,\stthree$ for stores.

The partial order $\ext$
can be lifted to distribution
ensembles as follows:
we say that a probabilistic store
$\stone =\tyf \rsone \renvone$ extends a
store $\sttwo= \tyf\rstwo\renvtwo$,
written $\sttwo \ext \stone$, if and only if $\renvtwo \ext \renvone$ and
$\rstwo$ can be obtained by computing the marginal distribution of
$\rsone$ on the domain of $\renvtwo$, itself often indicated
as $\rsone_{\renvone\to \renvtwo}$.
%
In turn, the definition of projection on distribution
ensembles is lifted to stores by setting
$\left(\tyf \rsone \renvone\right)_{\renvone\to \renvtwo}=\tyf {\rsone_{\renvone\to\renvtwo}} \renvtwo$.
\begin{definition}[Projection]
  Given a store $\tyf \rsone \renvone$ and $\renvtwo$ such that $\renvtwo\ext \renvone$, we define the projection $\left(\rsone\right)_{\renvone \to \renvtwo}\in \sem \renvtwo{}$, as follows:
  \[
    \left(\rsone\right)_{\renvone \to \renvtwo} \defsym \left\{\bind{\rsone_\nat}{\memone \mapsto \unit{\memone\restr{\dom(\renvtwo)}}}\right\}_{\nat\in \NN}
  \]
  where for every set of random variables $\rvarsetone$,
  $\memone\restr{\rvarsetone}$ is the function obtained
  by restricting the domain of
  $\memone$ to $\rvarsetone$.
\end{definition}

The relation $\ext$
is widely employed
in Section \ref{sec:logic},
because it allows our logic to describe
\emph{local} properties, i.e. properties that are valid in
\emph{portions of stores}.

%
We can define the tensor product of two
stores
as follows:

\begin{definition}
  The pointwise tensor product of
  two distribution ensembles $\tyf \rsone \renvone$
  and $\tyf \rstwo \renvtwo$ such that $\dom(\renvtwo)\cap
  \dom (\renvthree)=\emptyset$
  is defined as a distribution ensemble ${\tyf \rsone 
  \renvone}\tensprod{\tyf \rstwo \renvtwo}$ in
  $\sem{\join{\renvone}{\renvtwo}}{}$, as follows:
  \[
    ({\tyf \rsone \renvone}\tensprod{\tyf \rstwo \renvtwo})_\nat(\memone) 
    \defsym 
    \rsone_\nat(\memone\restr{\dom(\renvone)}) \cdot 
    \rstwo_\nat(\memone\restr{\dom(\renvtwo)}).
  \]
  We write $\stone \comp \sttwo \defined$ when the tensor product of two stores 
  is defined.
  \hfill\qed
\end{definition}

\subsection{Expressions}
\label{sec:expressions}
\begin{figure*}[t]
  \centering
    \begin{align*}
    \parsem {\ternjudg{\renvone}{\rvarone}{\tyone}} {\nat}(m) &\defsym {m(r)}\\
    \parsem {\ternjudg{\renvone}{\rfunone(\deone_1, \ldots, \deone_k)}{\tyone}} {\nat}(m) &\defsym \daof{\rfunone}(1^\nat)(\parsem{\ternjudg\renvone{\deone_1}{\tyone_1}}{\nat}(m), \ldots, \parsem{\ternjudg\renvone{\deone_k}{\tyone_k}}{\nat}(m))\\
    \detsem {\ternjudg{\renvone}{\deone}{\tyone}} {\nat}(m) &\defsym \unit{\parsem {\ternjudg \renvone \deone \tyone}{\nat}(m)}\\
    \detsem{\ternjudg{\renvone}{\rfunone(\reone_1,\ldots,\reone_k)} \tyone}{\nat}(m) & \defsym \kleisli{\raof \rfunone{}(1^\nat)} ({\detsem{\ternjudg{\renvone}{\reone_1}{\tyone_1}}\nat(m)\tensprod \ldots\tensprod \detsem{\ternjudg{\renvone}{\reone_k}{\tyone_k}}\nat(m)})\\[0.5ex]
    \sem{\binjudg{\renvone}{\pskip}}{\nat} &\defsym \distrone \mapsto \distrone\\
    \sem {\binjudg{\renvone}{\ass{\rvarone}{\reone}}} {\nat} &\defsym \distrone \mapsto
                                                           \bind {\distrone}{m \mapsto \bind {\detsem{\ternjudg{\renvone}{\reone}{\tyone}} \nat(m)} {t \mapsto \unit{m[\rvarone\mapsto t]}}}\\
    \sem{\binjudg{\renvone}{\seq {\prgone}{\prgtwo}}}{\nat}&\defsym \distrone \mapsto \sem{{\binjudg{\renvone}{\prgtwo}}}{\nat}\left(\left(\sem{\binjudg{\renvone}{\prgone}}{\nat}\right)(\distrone)\right)\\
    \sem{\binjudg{\renvone}{\ifr{\rvarone}{\prgone}{\prgtwo}}}{\nat} & \defsym \distrone \mapsto
                                                                 \gbind{\distrone}{m \mapsto\begin{cases}
                                                                                                    \sem {\binjudg \renvone\prgone}{\nat}(\unit {m}) & \text{if }m(\rvarone)=1\\  
                                                                                                    \sem {\binjudg \renvone\prgtwo}{\nat}(\unit {m}) & \text{if }m(\rvarone)=0  
                                                                                                    \end{cases}}
  \end{align*}
  \caption{Semantics of Expressions and Programs}
  \label{fig:exprgsem}
\end{figure*}

Before introducing programs,
we describe the
\emph{expressions} that occur in them. The set $\funset$ of
all \emph{function symbols} is assumed to be given, and
partitioned into two subsets $\dfset$ and $\rfset$, of 
deterministic and randomized function symbols, respectively.
The function $\tof \cdot$ provides the type of every function symbol
in $\funset$, i.e, 
$$
\tof{\rfunone}=\tyone_1,\ldots,\tyone_k\rightarrow\tytwo.
$$
\commentout{where $\tyone_1,\ldots,\tyone_k,$ and $\tytwo$ are types.
}
For every deterministic function symbol $\rfunone\in \dfset$,
such that
$\tof{\rfunone}=\tyone_1,\ldots,\tyone_k\rightarrow\tytwo$,
we assume the existence of a deterministic algorithm $\daof{\rfunone}$
that, on input $(1^\nat, s)$ where $s \in \detsem {\tyone_1} \nat
\times \ldots\times \detsem {\tyone_k} \nat$,
returns a value $t \in \detsem \tytwo \nat$. \revision{In order
not to violate the time constraint}, we assume
that $\daof{\rfunone}$ is polytime in the size of its first argument
for every $\rfunone \in \dfset$.
Similarly, for every $f \in \rfset$, we 
assume the existence 
of an efficient \emph{randomized} algorithm $\raof{\rfunone}$ 
that, on input $(1^\nat,s)$ where 
$s\in\detsem{\tyone_1}{\nat}\times\ldots\times\detsem{\tyone_k}{\nat}$, returns a 
distribution in $\sem{\sigma}{\nat}$.

The following primitive operations, with self-explanatory behavior, will be useful below:
\begin{gather*}
  \begin{aligned}
    \rand&:\varepsilon\arr\stng{n} & \tail_{\polyone+1} &:\stng{\polyone(n)+1}\arr\stng{\polyone(n)}\\
      \head_\polyone&:\stng{\polyone(n)+1}\arr\bool & \exor &:\stng{n}\times\stng{n}\arr\stng{n}\\
  \end{aligned}\\
  \concat_{\polyone, \polytwo}:\stng{\polyone (n)}\times\stng{\polytwo (n)} \arr\stng{\polyone (n)+\polytwo (n)} \\
  \setzero_{\polyone}:\varepsilon\arr\stng{\polyone(n)}
\end{gather*}
The only randomized primitive operation among these is $\rand$, which simply returns 
$n$ independently and uniformly distributed random bits. Noticed that these functions
may be indexed by some polynomials that describe the size of their arguments.
We omit these annotations, when the sizes of the arguments are evident from the context.

\emph{Deterministic and random expressions} are defined through the following grammars:
\begin{align*}
\deone&\bnf\rvarone\midd\rfuntwo(\deone_1,\ldots,\deone_k) &
\reone&\bnf\deone\midd\rfunone(\reone_1,\ldots,\reone_k)
\end{align*}
where $\rvarone\in\rvarset, \rfuntwo \in \dfset, \rfunone \in \rfset$.
A simple type system with 
judgments in the form 
$\ternjudg{\renvone}{\reone}{\tyone}$ can be given as follows:
\begin{gather*}
\infer{\ternjudg{\renvone}{\rvarone}{\tyone}}
{\renvone(\rvarone)=\tyone}\quad
\infer{\ternjudg{\renvone}{\rfunone(\reone_1,\ldots,\reone_k)}{\tytwo}}
{\deduce{\tof{\rfunone}=\tyone_1,\ldots,\tyone_k\rightarrow\tytwo}{
\left\{\ternjudg{\renvone}{\reone_i}{\tyone_i}\right\}_{1\leq
i\leq k}}	
}
\end{gather*}
\noindent
We write $\fv{\reone}$ for the set of free variables of $\reone$.
%
One can define the semantics of typed expressions in a standard way, as 
formalized in the following definition.
\begin{definition}[Semantics of Expressions]
  The semantics of a \emph{deterministic expression} $\deone$
  is a family of functions $\{\parsem {\ternjudg \renvone \deone \tyone} \nat\}_{\nat\in \NN}$
  such that for every $\nat \in \NN$, we have
  $\parsem {\ternjudg \renvone \deone \tyone} \nat: \detsem \renvone \nat \to \detsem \tyone \nat$.
  The \emph{pre-semantics} of a general expression $\ternjudg{\renvone}{\reone}{\tyone}$, is a family $\left\{\detsem {\ternjudg{\renvone}{\reone}{\tyone}}{\nat}\right\}_{\nat \in \NN}$ where for every $\nat \in \NN$,  $\detsem {\ternjudg{\renvone}{\reone}{\tyone}}{\nat}: {\detsem \renvone \nat}\to  {\sem {\tyone}\nat}$, see Figure \ref{fig:exprgsem}.
  The semantics of a randomized expression $\reone$, indicated as 
  $\sem{\ternjudg \renvone  \reone \tyone}{}$, is the family of functions 
  $\sem{\ternjudg{\renvone}{\reone}{\tyone}}\nat: \sem \renvone\nat\to  
  {\sem\tyone\nat}$, indexed by $\nat \in \NN$, where every $\nat\in \NN$, 
  $\sem{\ternjudg{\renvone}{\reone}{\tyone}}\nat$ is the Kleisli extension of
  $\detsem {\ternjudg{\renvone}{\reone}{\tyone}}{\nat}$.
  \hfill\qed
\end{definition}

\subsection{Programs}
\label{sec:subprograms}

We work within a probabilistic imperative language without iteration.
(We discuss in Section~\ref{sec:limitations} below why dealing with loops is nontrivial.)

\emph{Programs} are defined by the following grammar:
\begin{align*}
  \prgone, \prgtwo\bnf\;\pskip\midd
               \ass{\rvarone}{\reone}\midd
               \seq{\prgone}{\prgone}\midd
               \ifr{\rvarone}{\prgone}{\prgone}
\end{align*}
A type system with judgments in the form  $\binjudg{\renvone}{\prgone}$
for the just defined programs can be 
given in a standard way. 
Programs can be endowed with a denotational semantics.
\begin{definition}
  \label{def:prgsem}
  Whenever $\binjudg \renvone \prgone$, the semantics of $\prgone$ is
  the family $\left\{ \sem {\binjudg \renvone \prgone} \nat \right\}_{\nat \in \NN}$,
  where for every $\nat \in \NN$,
  $\sem {\binjudg \renvone \prgone} \nat:\sem \renvone\nat \to \sem\renvone\nat$ 
  as in Figure \ref{fig:exprgsem}. 
  This function induces an endofunction on $\sem \renvone{}$,
  which we indicate as $\sem{\binjudg{\renvone}{\prgone}}{}$ and which is defined as 
  $\left(\sem{\binjudg{\renvone}{\prgone}}{}(\rsone)\right)_\nat 
  \defsym\sem{\binjudg{\renvone}{\prgone}}{\nat}(\rsone_\nat)$.
  Very often, $\sem{\binjudg{\renvone}{\prgone}}{\nat}$ (resp. $\sem{\binjudg{\renvone}{\prgone}}{}$)
  will be abbreviated just as $\sem{\prgone}{\nat}$ (resp. $\sem{\prgone}{}$). 
  \hfill\qed
\end{definition}

The interpretations of $\pskip$ and of the composition $\prgone \seq \prgtwo$
are trivial. \revision{The interpretation of assignments is similar to that of
  sampling statements in \cite{PSL}; it composes the initial distribution
  with the continuation that associates every $\mem$ in the support of
  the initial distribution 
with the store that is obtained by evaluating $\reone$ on $\mem$
and updating the value of $\rvarone$ in $\mem$ accordingly.}
Finally, the interpretation of the conditional constructs
composes the distribution with the function that,
given an element of its support, executes $\prgone$ on
it if this program satisfies the guard,
and executes $\prgtwo$ otherwise.
Overall, the definition of our semantics is close to that of~{Kozen}~\cite{Kozen},
but that the semantics of the conditional construct
is not given in terms of conditional distributions.
This is due to the need of having an endofunction on
\emph{efficiently samplable} distribution ensembles
as the semantics of our language.
\revision{
  Specifically, we cannot use conditioned distributions
  because, given a polytime sampler for the initial distribution
  ensemble $\rsone$, we may not be able to identify a polytime
  sampler for the distribution ensemble obtained by conditioning $\rsone$  on
  the value of the guard. For instance, although a construction
  based on rejection sampling would produce the desired output distribution,
  it would not be \emph{polytime} in general. 
}
In \Cref{subsec:kozeneq} below, we show that the two definitions are
equivalent when evaluated \emph{on distributions}.

Notice that, by construction, for every program $\prgone$,
$\sem \prgone {}$ can be seen as a family of circuits,
and for every $\nat\in \NN$,
the circuit $\sem{{\prgone}}{\nat}$
has polynomial size in $\nat$.

\subsection{Notions from Computational Cryptography}
\label{subsection:crypto}
One of our main goals is to provide a formal framework for reasoning about definitions and constructions in complexity-based cryptography, in this work we follow a standard approach (e.g., as presented in \cite{KatzLindell} or \cite{Goldreich}). A fundamental characteristic of this setting is that all parties, both honest and adversarial, are modeled as efficient (randomized) agents, in particular as \emph{probabilistic polynomial time} (PPT) algorithms. This framework underlies the notion of samplablability given in Definition~\ref{def:samp}, as such distributions are exactly those sampled by some PPT algorithm.

While it is well-known that under reasonable restrictions (e.g., on the length of keys) it is impossible to provide absolute guarantees against adversarial success, definitions of security in the computational setting require \emph{negligible} adversarial advantage, where negligibilty is measured as a function of a \emph{security parameter} $n \in \NN$. This value measures the size of the secret information used by honest parties (e.g. the length of a secret key). In keeping with the asymptotic complexity bounds on adversaries, this leads to the following:
\begin{definition}
  \label{def:nf}
A function $\nu:\NN \to \NN$ is negligible if for any $k \in \NN$, for sufficiently large $n$, $\nu(n)\le 1/n^k$.
\hfill\qed
\end{definition}

A fundamental concept in cryptography is a binary relation on distribution 
ensembles called computational indistinguishability, by which we mean that no 
probabilistic polynomial time algorithm
can distinguish between the two ensembles with non-negligible
probability of success. Here, we can 
introduce computational
indistinguishability as a relation on stores:
\begin{definition}
  \label{def:ensind}
  We say that two distribution ensembles $\{\rsone_n\}_{\nat \in \NN},
  \{\rstwo_\nat\}_{\nat \in \NN}\in \sem {\renvone}{}$
  (resp. in $\sem \tyone{}$) are indistinguishable,
  and we write $\{\rsone_\nat\}_\nat\ind
  \{\rstwo_\nat\}_{\nat}$, if and only if
  for every probabilistic polynomial time algorithm $\distone$,
  there is a negligible function $\negl: \NN \to \NN$ such that
  for every $\nat \in \NN$:
  \[
    \left|\Prob[x\leftarrow \rsone_\nat]\left[\distone(1^\nat, x)=1\right]-\Prob[x\leftarrow \rstwo_\nat]\left[\distone(1^\nat, x)=1\right]\right|\le \negl(\nat).
    \qedeqn
  \]
\end{definition}

With this definition at hand, we are able to define a central notion in complexity-based cryptography: \emph{pseudorandomness}.
Let $\{\mathcal{U}_n\}_{\nat \in \NN}$ denote the distribution ensemble where $\mathcal{U}_n$ is the uniform distribution on $\BB^n$. 

\begin{definition}
\label{def:pseudo}
A distribution ensemble $\{\rsone_n\}_{\nat \in \NN}$ is pseudorandom if
$\{\rsone_\nat\}_\nat\ind \{\mathcal{U}_\nat\}_\nat$.
A (deterministic) poly-time function $f:\BB^*\to\BB^*$ is a pseudorandom generator if there is a function
$\ell:\NN\to\NN$ with $\ell(n)>n$
such that for $x\in \BB^n$, $f(x)\in \BB^{\ell(n)}$ and 
$\{f(\mathcal{U}_\nat)\}_\nat \ind  \{\mathcal{U}_{\ell(n)}\}_\nat$.
\hfill\qed
\end{definition}

Thus, a distribution ensemble is pseudorandom if no PPT algorithm can distinguish between samples from the distribution and samples from a truly random distribution.

A more detailed discussion of complexity-based cryptography is given in Section~\ref{sec:compindep}.

\subsection{Semantic Properties of Programs}
\label{subsec:cryptoprog}

This section is devoted to showing some remarkable properties
of our programs' semantics. We start by showing that
it is equivalent to the one proposed by Kozen~\cite{Kozen}.

In \Cref{subsec:propofexpr,subsec:propofprog}, we show
some other general properties of the expression of expressions
and programs, notably these include \Cref{lemma:projexpsem,lemma:semextproj},
stating that the outcome of the evaluation of expressions
and programs does not depend on those variables that do
not appear in the program, and \Cref{lemma:prgind,lemma:exprind}, where we show that
the distribution ensemble that is obtained by evaluating the
same program (resp. expression) on two computationally
indistinguishable ensembles are themselves computationally indistinguishable.

\subsubsection{Equivalence with Kozen's Semantics}
\label{subsec:kozeneq}

In this subsection, we are interested in showing that our semantics
is equivalent to that by Kozen~\cite{Kozen}.
A subtle difference within these two semantics is that,
while that by Kozen is defined on \emph{ordinary} distributions, ours is
defined on efficiently samplable distributions. This means that,
with our semantics,
if we take an efficiently samplable distribution,
and we compute the semantics of a program on it, we obtain
another efficiently samplable distribution.
In contrast to our semantics, Kozen's one is defined
on \emph{ordinary} distributions, so
it does not provide any guarantee on the
existence of such a sampler.
However, the equivalence between these two semantics
\emph{seen on regular distributions} has as a consequence that
the loopless fragment of Kozen's language (and semantics) transforms
efficiently samplable distributions in other
efficiently samplable distributions.
This result is not trivial because, on conditional statements,
Kozen's semantic of the $\ifr \rvarone \prgone \prgtwo$
computes the conditioned distributions of $\stone$ on the
$\rvarone=1$ and $\rvarone =0$ obtaining two states $\stone_1$ and $\stone_0$,
Finally, it outputs the convex combination of $\sem \prgone{} ({\stone_1}){}$
and $\sem \prgone{} ({\stone_0}){}$ according to $\sem \rvarone{} ({\stone}){}$.
In particular, as far as we know, it is not known whether
efficiently samplable distributions are closed under conditioning.

The first step we do in order to establish this claim,
is to define the conditioning operation:
given a distribution ensemble $\rsone \in \sem \renvone{}$,
we write $\condee {{\rsone_\nat}} \rvarone b$
for the distribution that is obtained by conditioning $\rsone_\nat$
on the event $\mem(\rvarone)=\bool$. Formally, this distribution is
defined as follows:
\[
   \condee {{\rsone_\nat}} \rvarone b(\mem) \defsym
  \begin{cases}
    \frac {{\rsone_\nat}(\mem)}{{\rsone_\nat}(S)} & \text{if } \mem(\rvarone) = b \text{ and } S=\{\mem\in \supp(\rsone_\nat) \mid \mem(\rvarone) = b\}\neq \emptyset\\
    0 &\text{otherwise.}
  \end{cases}
\]
Finally, we define the convex combination of two distributions
$\distrone, \distrtwo \in \distr X$ for $\nat\in \NN$ as follows:
$$
\left(\distrone\conv d \distrtwo\right)(\vec z) \defsym d(1)\cdot \distrone(x) + d(0)\cdot \distrtwo(x).
$$
\begin{definition}
  The Kozen-style definition of the semantics of loopless program ($\ksem{\cdot}{\cdot}$) is defined identically to that of\Cref{fig:exprgsem}, except for $\ifr \rvarone \prgone \prgtwo$, which is defined as follows:
    \[
      \ksem{\ifr \rvarone \prgone \prgtwo}{\nat} \defsym
        \begin{cases}
          \ksem \prgone\nat(\rsone_\nat) & \text{if } \sem \rvarone\nat(\rsone_\nat)(1) = 1\\ 
          \ksem \prgtwo\nat(\rsone_\nat) & \text{if } \sem \rvarone\nat(\rsone_\nat)(1) = 0 \\
          \ksem \prgone\nat(\condee {{\rsone_\nat}} {\rvarone}{1})\conv{\sem\rvarone\nat(\rsone_\nat)}\ksem \prgtwo\nat(\condee {{\rsone_\nat}} {\rvarone}{0}) & \text{otherwise.}
        \end{cases}        
      \]
    \end{definition}
In order to show the equivalence we are aiming to,
we also need to introduce some results on measures:
given two measures $\mu, \nu$, we write $\mu + \nu$ as
a shorthand for $ x \mapsto \mu(x)+\nu(x)$, and for
every $k \in [0,1]$, we write $k \cdot \mu$ for
$x \mapsto k \cdot \mu (\nat, x)$.
By considering the extension of our program's semantics
to measures, we obtain the following result:
\begin{rem}
  \label{remark:sumprodsemcommute}
  Let $\prgone$ be a program, and $\nat$ a natural number. It holds that:
  \begin{itemize}
  \item $\sem\prgone{}(\mu+\nu)= \sem \prgone{}(\mu)+\sem \prgone{}(\nu)$
  \item $\forall 0\le k\le 1.\sem\prgone{}(k\cdot\mu)= k \cdot\sem \prgone{}(\mu)$
  \end{itemize}
\end{rem}
\begin{proof}
  By induction on $\prgone$.
\end{proof}

Leveraging this Remark, it is possible to show the
main result of this section, the following lemma:

\begin{lemma}
  \label{lemma:semcons}
  The semantics of programs, seen as in Definition \ref{def:prgsem}, if interpreted on \emph{ordinary} distributions,
  is equivalent to that by Kozen,~\cite{Kozen}.
\end{lemma}

\begin{proof}
  We go by induction on $\prgone$. Whenever $\prgone$ is not a conditional,
  the conclusion is trivial because the two semantics are identical by
  definition. For this reason, we take in exam only the case where the
  program is $\ifr \rvarone \prgone \prgtwo$. We first assume that the
  guard is identically 1 or 0. In those cases ewe are required to verify that:
  \small
  \[
    \gbind{\rsone_\nat}{\memone \mapsto\begin{cases}
                                       \sem \prgone{\nat}(\unit {\memone}) & \text{if }\memone(\rvarone)=1\\  
                                       \sem \prgtwo{\nat}(\unit {\memone}) & \text{if }\memone(\rvarone)=0  
                                     \end{cases}}
  \]
  \normalsize
  Is equal to $\sem \prgone\nat(\rsone_\nat)$ or respectively to $\sem \prgtwo\nat(\rsone_\nat)$
  we assume without lack of generality that the distribution of the guard is identically 1.
  In this case we are asked to show that:
  \[
    \bind{\rsone_\nat}{\memone \mapsto \sem \prgone{\nat}(\unit {\memone})}=\sem \prgone\nat(\rsone_\nat).
  \]
  A formal proof of this result is in \Cref{lemma:unitsem} below; notice that the proof of that lemma does not rely on this lemma.

  Now, we assume that the guard is not identically 1 or 0.
  In this case, the claim is:  $\forall \renvone.\forall \nat \in \NN.\forall\rsone_\nat\in \sem\renvone\nat.\sem {\ifr \rvarone \prgone \prgtwo}\nat(\rsone_\nat)= \sem \prgone\nat(\condee {{\rsone_\nat}} {\rvarone}{1})\conv{\sem\rvarone\nat(\rsone_\nat)}\sem \prgtwo\nat(\condee {{\rsone_\nat}} {\rvarone}{0})$. 
  We start by fixing $\nat \in \NN$ and $\memone \in \supp (\rsone_\nat)$,
  and we observe that, by expanding the definition of the semantics, we obtain that
  $\sem{\ifr \rvarone \prgone \prgtwo}{}(\rsone)(\nat, \memone)$ is equal to:
  \begin{align*}
    \smashoperator[r]{\sum_{\memtwo :\memtwo\in \detsem \renvone \nat \land \memtwo(\rvarone)=1}}\;\;{\rsone(\nat, \memtwo)} \cdot \sem \prgone\nat(\unit {\memtwo})(\memone) +\smashoperator[r]{\sum_{\memtwo :\memtwo\in \detsem \renvone \nat \land \memtwo(\rvarone)=0}}\; \; {\rsone(\nat, \memtwo)}\cdot \sem \prgtwo\nat(\unit {\memtwo})(\memone)
  \end{align*}
  On the other hand, we can expand the definition of
  \[
    \left(\sem \prgone\nat(\condee {{\rsone_\nat}} {\rvarone}{1})\conv{\sem\rvarone\nat(\rsone_\nat)}\sem \prgtwo\nat(\condee {{\rsone_\nat}} {\rvarone}{0})\right)(\memone)
  \]
  as follows:
  \begin{align*}
     \smashoperator[r]{\sum_{\memtwo :\memtwo\in \detsem \renvone \nat \land \memtwo(\rvarone)=1}}\;\; \rsone_\nat(\memtwo)\cdot\sem \prgone{\nat}(\condee {{\rsone_\nat}} {\rvarone}{1})(\memone) +
    \smashoperator[r]{\sum_{\memtwo :\memtwo\in \detsem \renvone \nat \land \memtwo(\rvarone)=0}}\;\;\rsone_\nat(\memtwo)\cdot\sem \prgtwo{\nat}(\condee {{\rsone_\nat}} {\rvarone}{0})(\memone)
  \end{align*}
  Now, we can simply show that for every program $\prgone$ and Boolean $b$:
  \begin{align*}
    \smashoperator[r]{\sum_{\memtwo :\memtwo\in \detsem \renvone \nat \land \memtwo(\rvarone)=b}}\;\;\rsone_\nat( \memtwo)\cdot\sem \prgone\nat(\unit {\memtwo})(\memone) = 
    \smashoperator[r]{\sum_{\memtwo :\memtwo\in \detsem \renvone \nat \land \memtwo(\rvarone)=b}}\;\;\rsone_\nat( \memtwo)\cdot\sem \prgone{\nat}(\condee \stone {\rvarone}{b}(\nat))(\memone)
    \tag{$*$}
  \end{align*}
  We start by simplifying the expression on the right. First, we observe that, for Remark \ref{remark:sumprodsemcommute}, $\condee {{\rsone_{\nat}}} {\rvarone}{b}(\memone)$
  is equal to:
  \footnotesize
  \[
    \frac{1}{\sum_{\memthree :\memthree\in \detsem \renvone \nat \land \memthree(\rvarone)=b}\rsone_\nat (\memthree)}\cdot    \begin{cases}
                                                                                                                                  {\rsone_\nat(\memone)} & \text{if } \memone(\rvarone)=b\\
                                                                                                                                  0 & \text{otherwise}
                                                                                                                                \end{cases}
                                                                                                                              \]
  \normalsize
  Let $k_b \in [0,1]$ and $\rsthree_b$ be defined as follows:
  \[
    k_b \defsym \frac{1}{\sum_{\memthree :\memthree\in \detsem \renvone \nat \land \memthree(\rvarone)=b}\rsone_\nat(\memthree)}
    \quad\quad\quad
    \rsthree_b(\mem)\defsym\begin{cases}
                 {\rsone_\nat(\memone)} & \text{if } \memone(\rvarone)=b\\
                 0 & \text{otherwise.}
               \end{cases}
             \]
  With these definitions, the expression we were reducing can be rewritten as
  \[
    \ldots =\sum_{\memtwo :\memtwo\in \detsem \renvone \nat \land \memtwo(\rvarone)=b}\rsone_\nat(\memtwo)\cdot\sem \prgone\nat(k_b\cdot\rsthree_b)(\memone)\\
  \]
  from the observation above and Remark \ref{remark:sumprodsemcommute}, we obtain
  \small
  \begin{align*}
    \ldots =\sum_{\memtwo :\memtwo\in \detsem \renvone \nat \land \memtwo(\rvarone)=b}\rsone_\nat(\memtwo)\cdot k_b\cdot\sem \prgone\nat(\rsthree_b)(\memone)
      =\sem \prgone\nat(\rsthree_b)(\memone),
  \end{align*}
  \normalsize
  from the definition of $k_b$.
  We call $\rstwo^{\memone}$ the distribution ensemble defined as follows:
  \[
    \rstwo^{\memone}_\nat(\memthree) \defsym \begin{cases}
                                            \rsone_\nat({\memone}) & \text{if }\memone = \memthree\\
                                            0 & \text{otherwise}
                                          \end{cases}
  \]
  Observe that:
  $$
  \rsthree_b(\memthree)=\smashoperator[r]{\sum_{\memtwo :\memtwo\in \detsem \renvone \nat \land \memtwo(\rvarone)=b}}\;\;\rstwo^{\memtwo}_\nat(\memthree),
  $$
  so we can continue the reduction of the term observing that:
  \[
    \sem \prgone\nat(\rsthree_b)=\sem \prgone \nat\left(\sum_{\memtwo :\memtwo\in \detsem \renvone \nat \land \memtwo(\rvarone)=b}\rstwo^{\memtwo}_{\nat}\right)
  \]
  With an application of Remark \ref{remark:sumprodsemcommute} we can state the following equivalence:
  \begin{align*}
    \sem \prgone\nat(\rsthree_b)= \smashoperator[r]{\sum_{\memtwo :\memtwo\in \detsem \renvone \nat \land \memtwo(\rvarone)=b}}\;\;\sem \prgone \nat\left(\rstwo^{\memtwo}_\nat\right),
  \end{align*}
  finally, we observe that $\rstwo^{\memtwo}_\nat(\mem)=\rsone_\nat(\mem)\cdot \unit{\memtwo}(\memone)$ so, we have:
  \begin{align*}
    \sem \prgone\nat(\rsthree_b)= \smashoperator[r]{\sum_{\memtwo :\memtwo\in \detsem \renvone \nat \land \memtwo(\rvarone)=b}}\;\;\sem \prgone \nat\left(\rstwo^{\memtwo}_\nat\right)
    = \smashoperator[r]{\sum_{\memtwo :\memtwo\in \detsem \renvone \nat; \land \memtwo(\rvarone)=b}}\;\;\sem \prgone \nat\left(\mem \mapsto \rsone_\nat(\mem)\cdot \unit{\memtwo}(\memone)\right).
  \end{align*}
  So, in order to show our claim, it suffices to observe that for every $\memtwo \in \detsem\renvone\nat$ such that $\memtwo(\rvarone)=b$, we have:
  \[
    \sem \prgone \nat\left(\mem \mapsto\rsone_\nat(\mem)\cdot \unit{\memtwo}(\memone)\right) =
    \rsone_\nat(\memtwo) \cdot \sem \prgone \nat\left(\unit{\memtwo}\right).
  \]
  But this is a consequence of $\mem \mapsto\rsone_\nat(\mem)\cdot \unit{\memtwo}(\memone) =
  \mem \mapsto\rsone_\nat(\memtwo)\cdot \unit{\memtwo}(\memone) $ and of \Cref{remark:sumprodsemcommute}.
\end{proof}

\subsubsection{Properties of Expressions}
\label{subsec:propofexpr}

We start by showing that if $\stone \ind \sttwo$, then 
$\sem \reone {}(\stone)\ind\sem \reone {}(\sttwo)$ (\Cref{lemma:exprind})
and that, if $\stone \ext \sttwo$ then
$\sem \reone {}(\stone)=\sem \reone {}(\sttwo)$ (\Cref{lemma:projexpsem}).
These properties in turn, are useful later
to show that the semantic of specific \CBI's formulas
is closed with respect to $\ind\ext$ (\Cref{lemma:indcbi}).
Leveraging this result, we show that our model also enjoys
a form Restriction Property (\Cref{cor:restriction}), and
that some standard axiom schema for atomic formulas hold, \Cref{lemma:apax}.

\begin{lemma}[Indistinguishability of Expressions' Semantic]
  \label{lemma:exprind}
  For every expression $\reone$, and every pair of stores ${\tyf \rsone \renvone} \ind {\tyf \rstwo \renvtwo}$ such that $\ternjudg{\renvone}{\reone}{\tyone}, \ternjudg{\renvtwo}{\reone}{\tyone}$ it holds that $\sem\reone{}(\rsone)\ind \sem\reone{}(\rstwo)$.
\end{lemma}
\begin{proof}
  Assume the claim not to hold, this means that there is an adversary $\distone$
  which is capable to distinguish $\sem\reone{}(\rsone)$ and $\sem\reone{}(\rsone)$
  with non-negligible advantage.
  This adversary can be used to define an adversary $\disttwo$
  with the same advantage on $\rsone$ and $\rstwo$.
  On input $(1^\nat,\memone)$, it samples from $\detsem\reone\nat(\memone)$ and then
  executes $\distone$ on this sample.
  The advantage of $\disttwo$ on $\renvone, \renvtwo$ is the same of $\distone$
  on $\sem\reone\nat(\rsone)$ and $\sem\reone\nat(\rstwo)$.
  Notice that since $\reone$ is polytime computable, $\disttwo$ is also
  polytime.
\end{proof}

Then, we show that evaluation expression $\reone$ such that
$\ternjudg \renvone \reone \tyone$ in a state $\tyf \rsone \renvone$
yields the same result of evaluating it in a state
$(\tyf \rsone \renvone)_{\renvone \to \renvtwo}$ if $\ternjudg \retwo \reone \tyone$.

\begin{lemma}
  \label{lemma:parsemproj}
  For every deterministic expression $\deone$ each pair of environments $\renvone\ext \renvtwo$ such that $\ternjudg{\renvone}{\deone}{\tyone}$ and $\ternjudg{\renvtwo}{\deone}{\tyone}$, every $\nat \in \NN$ and every tuple $\memone \in \supp(\sem \renvtwo \nat)$, it holds that $\parsem {\ternjudg{\renvtwo}{\deone}{\tyone}}{\nat}(\memone)=\parsem {\ternjudg{\renvone}{\deone}{\tyone}}{\nat}(\memone\restr{\dom(\renvone)})$.
\end{lemma}

\begin{proof}
  The proof is by induction on the syntax of $\deone$
  \begin{itemize}
  \item If $\deone$ is $\rvarone$, we observe that
    for every $\nat \in \NN$,
    \begin{align*}
      \parsem {\ternjudg{\renvone}{\reone}{\tyone}}{\nat}(\memone\restr{\dom(\renvone)}) &= (\memone\restr{\dom(\renvone)})(\rvarone) = \memone(\rvarone)
    \end{align*}
    because $\rvarone$ in $\dom(\renvone)$ from the assumption $\ternjudg{\renvone}{\reone}{\tyone}$.
  \item If $\deone$ is a complex expression, the claim is a consequence of the IHs.
  \end{itemize}
\end{proof}

\begin{lemma}
  \label{lemma:detsemproj}
  For every expression $\reone$ each pair of environments $\renvone\ext \renvtwo$ such that $\ternjudg{\renvone}{\reone}{\tyone}$ and $\ternjudg{\renvtwo}{\reone}{\tyone}$, and every tuple $\memone \in \supp(\sem \renvtwo \nat)$, it holds that $\detsem {\ternjudg{\renvtwo}{\reone}{\tyone}}{}(\memone)=\detsem {\ternjudg{\renvone}{\reone}{\tyone}}{}(\memone\restr{\dom(\renvone)})$.
\end{lemma}

\begin{proof}
  The proof is by induction on the syntax of $\reone$
  \begin{itemize}
  \item If $\reone$ is a deterministic expression, we observe that
    for every $\nat \in \NN$,
    \begin{align*}
      \detsem {\ternjudg{\renvone}{\deone}{\tyone}}{\nat}(\memone\restr{\dom(\renvone)}) &= \unit{\parsem {\ternjudg \renvone \deone \tyone} \nat (\memone\restr{\dom(\renvone)})} \\
                                                                                   &= \unit{\parsem {\ternjudg \renvtwo \deone \tyone} \nat (\memone)}.
  \end{align*}
  \item If $\reone$ is a complex expression, the claim is a consequence of the IHs.
  \end{itemize}
\end{proof}

\begin{lemma}[Expressions' Semantic on projection]
  \label{lemma:projexpsem}
  For every expression $\reone$ such that $\ternjudg{\renvone}{\reone}{\tyone}$ and $\ternjudg{\renvtwo}{\reone}{\tyone}$, and every store $\rsone \in \sem \renvone{}$, $\sem \reone {}(\rsone_{\renvone \to \renvtwo})=\sem \reone {}(\rsone)$.
\end{lemma}

\begin{proof}
  We start by expanding the definition of $\sem \reone \nat(\rsone_{\renvone \to \renvtwo})$ for some $t \in \sem \tyone \nat$:
  \begin{align*}
    \ldots &= \bind{\bind{\rsone_\nat}{\memone \mapsto \unit {\memone\restr{\dom(\renvtwo)}}}}{\vec y \mapsto\detsem \reone \nat (\vec y)}\\
           &= \bind{\rsone_\nat} {\memone \mapsto \bind {\unit {\memone\restr{\dom(\renvtwo)}}}{\vec y \mapsto\detsem \reone \nat (\vec y)}}\\
           &= \bind{\rsone_\nat} {\memone \mapsto {\detsem \reone \nat ({\memone\restr{\dom(\renvtwo)}})}}
  \end{align*}
  It holds that $\detsem \reone \nat({\memone\restr{\dom(\renvtwo)}}) = \detsem \reone \nat (\memone)$ for Lemma \ref{lemma:detsemproj}, so we can conclude the proof.
\end{proof}

We conclude this section by showing a technical result that is employed in the proof of soundness of the rules 

\begin{rem}
  \label{rem:detsemmemupdate}
  For every $\nat \in \NN$, deterministic expression $\deone$, every $\renvone$ such that $\ternjudg \renvone \deone \tyone$, and every $\rvarone \in \dom (\renvone)$ such that $\rvarone \notin \fv \deone$, every $\memone \in \detsem \renvone \nat$, and every $t \in \detsem \tyone \nat$, we have that $\parsem \deone \nat (\memone) = \parsem \deone \nat (\memone[\rvarone \mapsto t])$. 
\end{rem}
\begin{proof}
  We start by fixing $\nat \in \NN$. The proof is by induction on $\deone$. If it is a variable, we are sure that it must be different from $\rvarone$, so the conclusion is trivial. In the inductive case, the conclusion is a direct consequence of the IHs on the sub-expressions.
\end{proof}


\subsubsection{Semantic Properties of Programs}
\label{subsec:propofprog}

In this section, we show some remarkable properties of our programs' semantics. We start with \Cref{lemma:semextproj}, where we show that the semantic of a program depends only on those variables that are used to type that program.
In \Cref{lemma:assorth,lemma:mvproj}, we show the soundness of the function $\mv\cdot$: in particular, we show that the program variables that do not belong to $\mv\cdot$ are not affected by the evaluation of $\prgone$.
Another result with the same spirit is \Cref{lemma:sempartcomm}, where we show that if a program $\prgone$ is evaluated on the tensor product of two distribution ensembles $\rsone$ and $\rstwo$, and the domain on one of them contains only variables that $\prgone$ does not use, then in the resulting distribution ensemble, these variables are still independent from the others.

\Cref{lemma:prgind} states that by evaluating the same program on two computationally indistinguishable stores, we obtain two computationally indistinguishable outputs. Crucially, it is possible to show this result because our programs are polytime.  

In \Cref{lemma:exprcompsem}, we show that after an assignment of an expression to a variable $\rvarone$ , such that $\rvarone \notin \fv\reone$ the values obtained by evaluating $\rvarone$ and $\reone$ are the same. \Cref{lemma:unitsem,lemma:suppbackcomp,lemma:detexprcompsem} are devoted to show a similar claim, but for deterministic expressions. These results and \Cref{lemma:sepassntech} are employed in the proof of soundness of the rules for assignments in \Cref{sec:inference}.

\begin{lemma}
  \label{lemma:semextproj}
  For every program $\prgone$, and every store ${\tyf \rsone \renvone}$ such that $\renvtwo \ext \renvone$,  $\binjudg{\renvtwo}{\prgone}$, it holds that $\sem\prgone{}(\rsone_{\renvone\to\renvtwo})= \sem\prgone{}(\rsone)_{\renvone\to\renvtwo}$.  
\end{lemma}
\begin{proof}
  The proof goes by induction on the program.
  \begin{itemize}
  \item If $\prgone = \pskip$, the conclusion is trivial.
  \item If $\prgone = \ass \rvarone \reone$, observe that:
    \begin{align*}
      (\sem {\prgone}\nat(\rsone))_{\renvone\to\renvtwo}&= \bind{\sem {\ass \rvarone \reone}\nat(\rsone)}{\memone\mapsto \memone\restr{\dom(\renvtwo)}}.
    \end{align*}
    By expanding the definition of the semantics and applying some monadic laws, it is possible to deduce that it is equal to:\\[1.5ex]
      $$\bind {\rsone_\nat} {\memone \mapsto\bind {\detsem {\binjudg \renvone \reone} {\nat} (\memone)} {t \mapsto \unit{{\memone[\rvarone\mapsto t]}\restr{\dom(\renvtwo)}}}},
      $$
    on the other hand, with similar expansions, we obtain that $(\sem {\prgone}\nat(\rsone_{\renvone\to\renvtwo}))$ is equal to:
      $$
      \bind {\rsone_\nat} {\memone \mapsto\bind {\detsem {\binjudg \renvtwo \reone} {\nat} (\memone\restr{\dom(\renvtwo)})} {t \mapsto \unit{{\memone\restr{\dom(\renvtwo)}[\rvarone\mapsto t]}}}},
      $$
    and with an application of \Cref{lemma:detsemproj}, we observe that, this is equal to:
    $$
    \bind {\rsone_\nat} {\memone \mapsto\bind {\detsem {\binjudg \renvone \reone} {\nat} (\memone)} {t \mapsto \unit{{\memone\restr{\dom(\renvtwo)}[\rvarone\mapsto t]}}}}.
    $$
    By observing $\binjudg \renvtwo \prgone$, we obtain that $\rvarone \in \dom(\renvtwo)$, which is enough to state the equivalence. 
  \item The case of composition is a trivial consequence of the IH.
  \item If $\prgone$ is a conditional we are required to show that:\\[1.5ex]
    $$
      \sem{\ifr \rvarone \prgone \prgtwo}{}(\rsone)_{\renvone \to \renvtwo} = \sem{\ifr \rvarone \prgone \prgtwo}{}(\rsone_{\renvone \to \renvtwo})
    $$
    The term on the right can be reduced as follows:
    $$
    \bind{\bind{\rsone_\nat} {\memone \mapsto h(\memone)}}{\memtwo \mapsto \unit{\memtwo\restr{\dom(\renvtwo)}}}
    $$
    where
    $$
    h(\memone) =
    \begin{cases}\sem \prgone \nat(\unit {\memone}) & \text{if } \memone(\rvarone)=1\\
      \sem \prgtwo \nat(\unit {\memone}) &\text{if } \memone(\rvarone)=0
    \end{cases}
    $$
    which is equal to:
    \begin{align*}
      \bind{\rsone_\nat} {\memone \mapsto \bind{h (\memone)}{\memtwo \mapsto \unit{\memtwo\restr{\dom(\renvtwo)}}}}
    \end{align*}
    and, by definition of projection, to:
    \begin{align*}
      \bind{\rsone_\nat} {\memone \mapsto h(\memone)_{\renvone \to \renvtwo}}.
    \end{align*}
    As we did of assignments, we can conclude the
    proof by assessing that the term above is equal to:
    \begin{align*}
      \bind{\rsone_\nat} {\memone \mapsto h(\memone\restr{\dom(\renvtwo)})}
      \tag{$*$}
    \end{align*}
    Observe that:
    $$
    h(\memone\restr{\dom(\renvtwo)}) =
    \begin{cases}\sem \prgone \nat(\unit {\memone}_{\renvone\to \renvtwo}) & \text{if } \memone(\rvarone)=1\\
      \sem \prgtwo \nat( {\unit {\memone}_{\renvone\to \renvtwo}}) &\text{if } \memone(\rvarone)=0
    \end{cases}
    $$
    We continue by proceed applying the IHs, to conclude that ($*$) is equal to:
    \begin{align*}
      \bind{\rsone_\nat} {\memone \mapsto h(\memone)_{\renvone \to \renvtwo}}.
    \end{align*}
  \end{itemize}
\end{proof}

\begin{lemma}
  \label{lemma:assorth}
  \[
    \forall \renvone.\forall \renvtwo \ext \renvone. \binjudg{\renvone}{\ass \rvarone \reone}\Rightarrow \rvarone \notin \dom(\renvtwo)\Rightarrow
    \forall \rsone \in \sem \renvone{}. {{\tyf \rsone \renvone}}_{\renvone \to \renvtwo} = ({\sem {\ass \rvarone \reone}{} {\tyf \rsone \renvone}})_{\renvone \to \renvtwo}.
  \]
\end{lemma}

\begin{proof}
  Observe that:
  \small
  \begin{align*}
    (\sem {\ass \rvarone \reone}\nat(\rsone))_{\renvone\to\renvtwo}&= \bind{\sem {\ass \rvarone \reone}{}(\rsone)}{\memone\mapsto \unit{\memone\restr{\dom(\renvtwo)}}}.
  \end{align*}
  \normalsize
  By expanding the definition of the semantics and applying standard monads' properties, it is possible to deduce that it is equal to:
  \small
  \[
    \bind {\rsone(\nat)} {\memone \mapsto\bind {\detsem \reone {} (\nat, \memone)} {t \mapsto \unit{{{\memone[\rvarone\mapsto t]}\restr{\dom(\renvtwo)}}}}},
  \]
  \normalsize
  which, in turn, is equal to:
  $$
  \bind {\rsone(\nat)}{\memone \mapsto\bind {\detsem \reone {} (\nat, \memone)} {t \mapsto \unit{{\memone}\restr{\dom(\renvtwo)}}}}
  $$
  for the assumption on $\rvarone$.
  It is easy to verify that, since ${t \mapsto \unit{{\memone}\restr{\dom(\renvtwo)}}}$ is a constant function, the expression above is equal to:
  $$
  \bind {\rsone(\nat)}{\memone \mapsto \unit{\memone\restr{\dom(\renvtwo)}}}.
  $$
\end{proof}

\begin{lemma}[Modified Variables and Projection]
  \label{lemma:mvproj}
  It holds that:
  \begin{multline*}
    \forall \prgone, \renvone. \forall \renvtwo \ext \renvone. \binjudg{\renvone}{\prgone}\Rightarrow \dom (\renvtwo)\cap\mv\prgone=\emptyset\Rightarrow
    \forall \rsone \in \sem \renvone. {{\tyf \rsone \renvone}}_{\renvone \to \renvtwo} = {\sem \prgone{} ({\tyf \rsone \renvone})}_{\renvone \to \renvtwo}.
  \end{multline*}
\end{lemma}
\begin{proof}
  This proof is by induction on $\prgone$.
  \begin{itemize}
  \item If $\prgone$ is $\pskip$, then the premise coincides with the conclusion.
  \item If $\prgone$ is $\seq \prgone \prgtwo$, then the claim comes from the application of the IH on $\prgone$ and of that on $\prgtwo$. 
  \item If the program is $\ifr \reone \prgone \prgtwo$, we start by expanding $      \sem{\ifr \rvarone \prgone \prgtwo}{\nat}(\rsone)_{\renvone \to \renvtwo}$ as follows:
        $$
        \bind{\bind{\rsone_\nat} {\memone \mapsto h(\memone)}}{\memtwo \mapsto \unit{\memtwo\restr{\dom(\renvtwo)}}}
        $$
        where
        $$
        h(\memone) =\begin{cases}\sem \prgone \nat(\unit {\memone}) & \text{if } \memone(\rvarone)=1\\
                                                                                                                               \sem \prgtwo \nat(\unit {\memone}) &\text{if } \memone(\rvarone)=0
                   \end{cases}
                   $$
        which is equal to:
        \begin{align*}
          \bind{\rsone_\nat} {\memone \mapsto \bind{h (\memone)}{\memtwo \mapsto \unit{\memtwo\restr{\dom(\renvtwo)}}}}\\
        \end{align*}
        Now, it is possible to observe that $\bind{h (\memone)}{\memtwo \mapsto \unit{\memtwo\restr{\dom(\renvtwo)}}}$ is equal to:
        $$
        \begin{cases}\bind{\sem \prgone \nat(\unit {\memone})}{\memtwo \mapsto \unit{\memtwo\restr{\dom(\renvtwo)}}} & \text{if } \memone(\rvarone)=1\\
                                                         \bind{\sem \prgtwo \nat(\unit {\memone})}{\memtwo \mapsto \unit{\memtwo\restr{\dom(\renvtwo)}}} &\text{if } \memone(\rvarone)=0,
                                                       \end{cases}
                                                       $$
                                                       \normalsize
                                                       which in turn is equal to 
        $$
\begin{cases} {\sem \prgone \nat(\unit {\memone})}_{\renvone \to \renvtwo} & \text{if } \memone(\rvarone)=1\\
                                                         {\sem \prgtwo \nat(\unit {\memone})}_{\renvone \to \renvtwo} &\text{if } \memone(\rvarone)=0.
                                                       \end{cases}                                                       $$
      The conclusion follows from the IHs.
    \item $\prgone$ is $ \ass \rvarone \reone$, the claim is a consequence of Lemma \ref{lemma:assorth}.
    \end{itemize}
\end{proof}

\begin{lemma}
  \label{lemma:sempartcomm}
  For every program $\prgone$, and every pair of stores $\stone = {\tyf \rsone \renvone}$, $\sttwo = {\tyf \rstwo \renvtwo}$ such that $\binjudg\renvone \prgone$ and $\stone \comp \sttwo \defined$, it holds that $\sem \prgone {}(\stone \comp \sttwo) = \sem \prgone {}(\stone) \comp \sttwo$.
\end{lemma}

\begin{proof}
  We will show that
  \[
  \sem {\binjudg{\join \renvone \renvtwo} \prgone}{\nat}((\rsone\tensprod \rstwo)_\nat) = \left(\sem {\binjudg \renvone \prgone}{}(\rsone) \tensprod \rstwo\right)_\nat.
  \tag{C}
  \]
  The proof goes by induction on $\prgone$.
  \begin{itemize}
  \item The case where $\prgone = \pskip$ is trivial.
  \item When $\prgone = \ass \rvarone \reone$, we observe that:
      \begin{align*}
      \sem {\ass \rvarone \reone}\nat((\rsone \tensprod \rstwo)_\nat) &=\bind{(\rsone \tensprod \rstwo)_\nat} {\memone \mapsto \bind {\detsem {\ternjudg{\join\renvone\renvtwo}\reone \tyone} {\nat}(\memone)}{t\mapsto \unit {\memone[\rvarone \mapsto t]}}}\\
      &=\bind{(\rsone \tensprod \rstwo)_\nat} {\memone \mapsto \bind {\detsem {\ternjudg{\renvone}\reone\tyone} {\nat}(\memone\restr{\dom(\renvone)})}{t\mapsto \unit {\memone[\rvarone \mapsto t]}}}
      \end{align*}
    Observe that $\unit {\memone[\rvarone \mapsto t]}$ is equal to $\memtwo \mapsto\unit {\memone\restr{\dom(\renvone)}[\rvarone \mapsto t]}(\memtwo\restr{\dom(\renvone)})  \unit {\memone\restr{\dom(\renvtwo)}}(\memtwo\restr{\dom(\renvtwo)})$, so we can expand the expressions and obtain:
    \small
    \begin{multline*}
      \memtwo\mapsto\;\;\smashoperator[lr]{\sum_{\memone \in \supp(\sem {\join \renvone \renvtwo}\nat)}}\;\;\rsone_\nat(\memone\restr{\dom(\renvone)})\cdot\rstwo_\nat(\memone\restr{\dom(\renvtwo)})\cdot\\ \sum_{t \in \sem \tyone \nat}\detsem \reone \nat(\memone\restr{\dom(\renvone)})(t)\cdot \unit {\memone\restr{\dom(\renvone)}[\rvarone \mapsto t]}(\memtwo\restr{\dom(\renvone)}) \cdot \\\unit {\memone\restr{\dom(\renvtwo)}}(\memtwo\restr{\dom(\renvtwo)})   
    \end{multline*}
    \normalsize
    By using the distributive property of multiplication, it is possible to observe that this is equal to:
    \begin{align*}
      \memtwo\mapsto\;\;\sem {\binjudg{\renvone}{\ass \rvarone \reone}}\nat(\rsone_\nat)({\memtwo\restr{\dom(\renvone)}})\cdot\smashoperator[lr]{\sum_{\memone \in \supp(\sem {\renvtwo}\nat)}}\;\;\rstwo_\nat(\memone)\cdot\unit {\memone}(\memtwo\restr{\dom(\renvtwo)}),
    \end{align*}
    which, in turn, is equal to:
    \begin{multline*}
      \memtwo\mapsto\sem {\binjudg{\renvone}{\ass \rvarone \reone}}\nat(\rsone_\nat) ({\memtwo\restr{\dom(\renvone)}})\cdot\bind{\rstwo_\nat}{\unit\cdot}(\memtwo\restr{\dom(\renvtwo)})=\\
      \memtwo\mapsto\sem {\binjudg{\renvone}{\ass \rvarone \reone}}\nat(\rsone_\nat)({\memtwo\restr{\dom(\renvone)}})\cdot{\rstwo_\nat}(\memtwo\restr{\dom(\renvtwo)}).
    \end{multline*}
    This, in turn, is equal to:
    \[
      \left(\sem {\binjudg{\renvone}{\ass \rvarone \reone}}{}(\rsone)\tensprod{\rstwo}\right)_\nat.
    \]
  \item When the program is $\seq\prgone\prgtwo$, the claim is a direct consequence of the IHs.
  \item When the program is $\ifr \reone \prgone \prgtwo$, $\sem{\ifr \reone \prgone \prgtwo}\nat((\rsone\tensprod \rstwo)_\nat)$ is equivalent to:
        $$
        \memtwo \mapsto\smashoperator[r]{\sum_{\memone \in \supp{\sem{\join\renvone \renvtwo}\nat}}}\;\;
        \rsone_\nat(\memone\restr{\dom(\renvone)})\cdot\rstwo_\nat(\memone\restr{\dom(\renvtwo)})\cdot
        h(\memone)(\memtwo)
        $$
        where
        $$
        h(\memone) =\begin{cases}\sem \prgone \nat(\unit {\memone}) & \text{if } \memone(\rvarone)=1\\
                                                                                                                               \sem \prgtwo \nat(\unit {\memone}) &\text{if } \memone(\rvarone)=0
                   \end{cases}
        $$
        Observe that $h(\memone)=h'(\memone\restr{\dom(\renvone)},\memone\restr{\dom(\renvtwo)})$, where
        $$
        h'(\memtwo_1, \memtwo_2) =\
        \begin{cases}\sem \prgone \nat(\unit {\memtwo_1}\tensprod\unit {\memtwo_2}) & \text{if } \memtwo_1(\rvarone)=1\\
          \sem \prgtwo \nat(\unit {\memtwo_1}\tensprod\unit {\memtwo_2}) & \text{if } \memtwo_1(\rvarone)=0\\
        \end{cases}
        $$
        By using the IH, we conclude that this is equal to:
        $$
        h'(\memtwo_1, \memtwo_2) =\
        \begin{cases}\sem \prgone \nat(\unit {\memtwo_1})\tensprod\unit {\memtwo_2} & \text{if } \memtwo_1(\rvarone)=1\\
          \sem \prgtwo \nat(\unit {\memtwo_1})\tensprod \unit {\memtwo_2} & \text{if } \memtwo_1(\rvarone)=0\\
        \end{cases}
        $$
        So we can rewrite the expression we wrote at the beginning as follows:
        \begin{equation*}
          \memtwo \mapsto\smashoperator[lr]{\sum_{\memone \in \supp{\sem{\renvone}\nat}}}\;\;
          \rsone_\nat(\memone)\cdot
          h(\memone)(\memtwo\restr{\dom(\renvone)})\cdot\\
          \smashoperator[r]{\sum_{\memone \in \supp{\sem{\renvtwo}\nat}}}\;\;
          \rstwo_\nat(\memone)\cdot\unit {\memone}(\memtwo\restr{\dom(\renvtwo)}),
        \end{equation*}
        which is:
        \begin{equation*}
          \memtwo \mapsto\smashoperator[lr]{\sum_{\memone \in \supp{\sem{\renvone}\nat}}}\;\;
          \rsone_\nat(\memone)\cdot
          h(\memone)(\memtwo\restr{\dom(\renvone)})\cdot \rstwo_\nat(\memone)(\memtwo\restr{\dom(\renvtwo)}).
        \end{equation*}
        This last term is equal, by definition, to:
        $$
        \left(\sem{\ifr \reone \prgone \prgtwo}{}(\rsone)\tensprod \rstwo\right)_\nat.
        $$
          
      \end{itemize}
      The conclusion comes by observing that
      $\sem {\binjudg{\join\renvone\renvtwo}\prgone} {}(\stone \comp \sttwo)=
      \tyf{\sem {\binjudg{\join\renvone\renvtwo}\prgone} {}(\rsone\tensprod\rstwo)}{\join\renvone\renvtwo}$ and that $\sem \prgone {}(\stone) \comp \sttwo=\tyf{\sem {\binjudg{\renvone}\prgone} {}(\rsone)\tensprod \rstwo}{\join\renvone\renvtwo}$. So, in particular:
      \[
        \sem {\binjudg{\join\renvone\renvtwo}\prgone} {}(\rsone\tensprod\rstwo)=
        \sem {\binjudg{\renvone}\prgone} {}(\rsone)\tensprod \rstwo
      \]
      is a consequence of (C).
\end{proof}

\begin{lemma}[Indistinguishability of Programs' Semantic]
  \label{lemma:prgind}
  For every program $\prgone$, and every pair of stores ${\tyf \rsone \renvone} \ind {\tyf \rstwo \renvtwo}$such that $\binjudg{\renvone}{\prgone}, \binjudg{\renvtwo}{\prgone}$, it holds that $\sem\prgone{}(\rsone)\ind \sem\prgone{}(\rstwo)$.
\end{lemma}
\begin{proof}
  Observe that the semantics of programs we gave can be defined in terms of the Kleisli extension of a distribution ensemble with poly-time samplable continuations. In particular, this can be shown by induction on the semantics of programs observing that:
  \begin{itemize}
  \item the claim is trivial for $\pskip$
  \item it is a consequence of the poly-time samplability property of the expressions' semantics in the case of assignments, and on the observation that updating a sample can be done in polynomial time
  \item it is a direct consequence of the IH in the case of the composition,
  \item in the case of the conditional statement, it is a consequence of the IHs, and of the observation that testing a variable on that sample can be done in polynomial time 
  \end{itemize}
  This means that for every program $\prgone$, there is a polynomial transformation on samples $t_\prgone$ such that:
  \[
    \sem {\binjudg \renvone \prgone}\nat (\rsone_\nat)={\bind{\rsone_\nat} {t_\prgone}}
  \]
  Assume the claim not to hold, this gives an adversary $\distone$ which is capable to distinguish $\sem\prgone{}(\rsone)$ and $\sem\prgone{}(\rstwo)$ with non-negligible advantage. This adversary can be used to define an adversary $\disttwo$ with the same advantage that given on input $\rsone$, it first computes the continuation $t_\prgone$ on the sample $\sem\prgone{}(\rsone)$ and then feeds this value into $\distone$. The advantage of $\disttwo$ is the same of $\distone$.
\end{proof}

\begin{lemma}
  \label{lemma:exprcompsem}
  If $\binjudg \renvone \ass \rvarone \reone$, $\rsone \in \sem \renvone{}$ and $\rvarone \notin \fv \reone$, it holds that: 
  \begin{align*}
    \sem \rvarone {}\left(\sem {\ass \rvarone \reone} {}(\rsone)\right) = \sem {\reone} {}(\rsone) = \sem \reone {}\left(\sem {\ass \rvarone \reone} {}(\rsone)\right).
  \end{align*}
\end{lemma}
\begin{proof}
    We start by fixing $\nat\in \NN$ and by observing that:
  \begin{align*}
    \sem \rvarone {\nat}\left(\sem {\ass \rvarone \reone} {}(\rsone_\nat)\right) &= \bind{\sem {\ass r \reone}{}(\rsone_\nat)}{\memone \mapsto \unit{\memone(\rvarone)}}
  \end{align*}
  the expression $\sem {\ass r \reone}{}(\rsone_\nat)$ can be expanded as follows:\\[0.5em]
  \[
    \bind {\rsone_\nat}{\memtwo \mapsto \bind {\detsem{\ternjudg{\renvone}{\reone}{\tyone}} \nat(\memtwo)} {t \mapsto \unit{\memtwo[\rvarone\mapsto t]}}}
  \]
  using the associative property of the distribution monad and the property of $\unit\cdot$, the whole expression can be rewritten as follows: 
  $$
  \bind{\rsone_\nat} {\memtwo \mapsto\bind {\detsem{\ternjudg{\renvone}{\reone}{\tyone}} \nat(\memtwo)} {t \mapsto \bind {\unit{\memtwo[\rvarone\mapsto t]}} {\memone \mapsto \unit{\memone(\rvarone)}}}}.
  $$
  The bind-unit pattern can be simplified as follows:
  $$
  \bind{\rsone_\nat} {\memtwo \mapsto\bind {\detsem{\ternjudg{\renvone}{\reone}{\tyone}} \nat(\memtwo)}{t \mapsto {\unit{\memtwo[\rvarone\mapsto t](\rvarone)}}}}.
  $$
  If we observe that ${\memtwo[\rvarone\mapsto t]}(\rvarone)=t$, it becomes
  \[
    \bind{\rsone_\nat} {\memtwo \mapsto\bind {\detsem{\ternjudg{\renvone}{\reone}{\tyone}} \nat(\memtwo)}{t \mapsto {\unit{t}}}},
  \]
  which, in turn equals 
  \[
    \bind{\rsone_\nat} {\memtwo \mapsto {\detsem{\ternjudg{\renvone}{\reone}{\tyone}} \nat(\memtwo)}}=\sem \reone\nat(\rsone_\nat).
  \]
  With a similar chain of equivalences, it is possible to rewrite $\sem \reone \nat\left(\sem {\ass \rvarone \reone} {}(\rsone)\right)$ as follows::
  $$
  \bind{\rsone_\nat} {\memtwo \mapsto\bind {\detsem{\ternjudg{\renvone}{\reone}{\tyone}} \nat(\memtwo)}{t \mapsto \detsem \reone \nat(\memtwo[\rvarone\mapsto t])}}.
  $$
  With an application of Lemma \ref{lemma:detsemproj}, we observe that this is equal to:   
  $$
  \bind{\rsone_\nat} {\memtwo \mapsto\bind {\detsem{\ternjudg{\renvone}{\reone}{\tyone}} \nat(\memtwo)}{t \mapsto \detsem \reone \nat(\memtwo)}}.
  $$
  The assumptions of Lemma \ref{lemma:detsemproj} can be matched by taking $\renvtwo$ as the environment obtained removing $\rvarone$ from $\renvone$. The inner bind can be expanded as:
  \[
    \sum_{x \in \sem \tyone\nat} \detsem{\ternjudg{\renvone}{\reone}{\tyone}}\nat (t) \cdot \detsem \reone \nat(\memtwo)= 1 \cdot \detsem \reone \nat(\memtwo).
  \]
  For this reason the expression above can be rewritten as
  $
  \bind{\rsone_\nat} {\memtwo \mapsto {\detsem{\ternjudg{\renvone}{\reone}{\tyone}} \nat(\memtwo)}}=\sem\reone \nat(\rsone_\nat).
  $
\end{proof}

\begin{lemma}
  \label{lemma:unitsem}
  For every $\nat$, every $\prgone$ and every $\rsone_\nat$, we have
  $\sem \prgone \nat(\rsone_\nat)= \bind {\rsone_\nat}{\memone \mapsto\sem \prgone\nat(\unit \memone)}$.
\end{lemma}

\begin{proof}
  By induction on $\prgone$.
  \begin{proofcases}
    \proofcase[$\pskip$] In this case, the claim is $\rsone_\nat=\bind {\rsone_\nat}{\memone \mapsto\unit \memone}$, which is trivial.
    \proofcase[$\ass \rvarone \reone$] In this case, observe that:
    \begin{align*}
      \bind {\rsone_\nat}{\memone \mapsto\sem {\ass \rvarone \reone}\nat(\unit \memone)}
    \end{align*}
    is equal to:
    \begin{align*}
    \bind {\rsone_\nat}{\memone \mapsto      \bind {\unit \memone}{\memtwo \mapsto \bind{\detsem \reone \nat (\memtwo)}{t \mapsto \memtwo[\rvarone \mapsto t] }}},
    \end{align*}
    that, in turn, equals:
    \begin{align*}
      \bind {\rsone_\nat}{\memone \mapsto\bind{\detsem \reone \nat (\memone)}{t \mapsto \memone[\rvarone \mapsto t]}} = \sem {\ass \rvarone \reone}\nat(\rsone_\nat)
    \end{align*}
    \proofcase[$\seq \prgone\prgtwo$] The claim is
    \[
      \sem {\seq \prgone\prgtwo} \nat(\rsone_\nat)= \bind {\rsone_\nat}{\memone \mapsto\sem {\seq \prgone\prgtwo}\nat(\unit \memone)}.
    \]
    For the IH on $\prgone$, we can rewrite $\sem{\seq \prgone\prgtwo}\nat(\rsone_\nat)$ as follows:
    \[
      \sem {\prgtwo} \nat(\bind {\rsone_\nat}{\memtwo \mapsto \sem {\prgone} \nat(\unit \memtwo)})
    \]
    For the IH on $\prgtwo$, this claim can be rewritten again as:
    \[
      \bind{\bind {\rsone_\nat}{\memtwo \mapsto \sem {\prgone} \nat(\unit \memtwo)}}{\memthree \mapsto \sem {\prgtwo} \nat(\unit\memthree)}
    \]
    That can be simplified as follows:
    \[
      \bind {\rsone_\nat} {\memtwo \mapsto \bind{\sem {\prgone} \nat(\unit \memtwo)} {\memthree \mapsto \sem {\prgtwo} \nat(\unit\memthree)}}
    \]
    With the application of the IH on $\prgtwo$, we conclude.
    \proofcase[$\ifr \rvarone \prgone\prgtwo$] The claim is a direct consequence of the
      definition of the semantics of the construct.
\end{proofcases}
\end{proof}
\begin{lemma}
  \label{lemma:suppbackcomp}
  For every $\nat\in \NN$, and every $\memone\in \supp(\sem \prgone \nat (\rsone_\nat))$,
  there is $\overline \memone \in \supp (\rsone_\nat)$ such that
  $\memone \in \supp(\sem \prgone\nat (\unit{\overline \memone}))$.
\end{lemma}
\begin{proof}
  The proof goes by induction on $\prgone$.
  \begin{proofcases}
    \proofcase[$\pskip$] Trivial.
    \proofcase[$\ass \rvarone \reone$] Assume that the claim does not hold.
    This means that there is $\memone \in \supp (\sem {\ass \rvarone \reone}\nat(\rsone_\nat))$ such that for every $\memtwo \in \supp(\rsone_\nat)$, we have:
    \[
      \bind {\detsem \reone \nat(\memtwo)}{t \mapsto \unit{\memtwo[\rvarone\mapsto t]}}(\memone)=0.
    \]
    From this result we can deduce
    $
    \sem {\ass \rvarone \reone} \nat (\rsone_\nat)(\memone) = 0
    $,
    but this contradicts the premise.
    \proofcase[$\ifr \rvarone \prgone \prgtwo$] Assume that the claim does not hold.
      This means there is $\memone \in \supp (\sem {\ifr \rvarone \prgone \prgtwo}\nat(\rsone_\nat))$ such that
       for every $\memtwo \in \supp(\rsone_\nat)$, we have:
    \begin{align*}
      \memtwo(\rvarone)=0 &\Rightarrow \sem \prgone \nat(\unit \memtwo)(\memone)=0\\
      \memtwo(\rvarone)=1 &\Rightarrow \sem \prgtwo \nat(\unit \memtwo)(\memone)=0.
    \end{align*}
    This shows that
    $
    \sem{\ifr \rvarone \prgone \prgtwo}\nat(\rsone_\nat)(\memone) = 0
    $,
    but this contradicts the premise.
    \proofcase[$\seq \prgone \prgtwo$] Assume $\memone\in \supp(\sem \prgone \nat (\rsone_\nat))$,
    we must show that there is $\overline m$ is
    such that:
    \[
      \memone \in \supp(\sem {\prgtwo}\nat (\sem \prgone \nat(\unit{\overline \memone})))
    \]
    From the IH on $\prgtwo$, we deduce that there is $\overline \memtwo \in \supp(\sem \prgone \nat(\rsone_\nat))$ such that 
    \[
      \memone \in \supp(\sem {\prgtwo}\nat (\unit{\overline \memtwo})).
    \]
    With another application of this hypothesis, we deduce that there is $\overline \memthree \in \supp (\rsone_\nat)$ such that:
    \[
      \overline \memtwo \in \supp(\sem {\prgone}\nat (\unit{\overline \memthree})).
    \]
    With an application of \Cref{lemma:unitsem}, we observe that
    $
    \sem {\prgtwo}\nat (\sem \prgone \nat(\unit{\overline \memthree}))(\memone)
    $
    is equal to:
    \[
      \bind {\sem \prgone \nat(\unit{\overline \memthree})}{\memtwo \mapsto \sem {\prgtwo}\nat(\unit \memtwo)}(\memone),
    \]
    which in turn is equal to
    \[
      \sum_{\memtwo \in \supp(\sem \prgone \nat(\unit{\overline \memthree}))} \sem \prgone \nat(\unit{\overline \memthree})(\memtwo)\cdot \sem {\prgtwo}\nat(\unit \memtwo)(\memone)
    \]
    which is greater than:
    \[
      \sem \prgone \nat(\unit{\overline \memthree})(\overline \memtwo)\cdot \sem {\prgtwo}\nat(\unit {\overline \memtwo})(\memone)
    \]
    which is greater than $0$ for the assumptions above.
  \end{proofcases}
\end{proof}

\begin{lemma}
  \label{lemma:detexprcompsem}
  If $\binjudg \renvone \ass \rvarone \deone$, $\rsone \in \sem \renvone{}$ and $\rvarone \notin \fv \deone$, it holds that: 
  \begin{align*}
    \forall \nat \in \NN. \forall \memtwo \in \supp(\sem {\ass \rvarone \deone} \nat(\rsone_\nat)). \parsem {\deone} {\nat}(\memtwo) = \parsem \rvarone {\nat}(\memtwo).
  \end{align*}
\end{lemma}
\begin{proof}
Assume this to be false for some $\nat$ and some $\memtwo \in \sem {\ass \rvarone \deone} \nat(\rsone_\nat)$. From \Cref{lemma:suppbackcomp}, we deduce that there is that there is $\memone\in \supp({\rsone})$ such that $\memtwo \in \supp (\sem {\ass \rvarone \deone}\nat(\unit \memone))$.
  Observe that:
  \begin{align*}
    \sem {\ass \rvarone \deone}\nat(\unit \memone) &= \bind{\unit \memone} {\memone \mapsto \bind {\detsem \deone \nat (\memone)} {t \mapsto \unit{\memone[\rvarone \mapsto t]}}}\\
           &=  \bind {\detsem \deone \nat (\memone)} {t \mapsto \unit{\memone[\rvarone \mapsto t]}}\\
           &=  \bind {\unit {\parsem \deone \nat (\memone)}} {t \mapsto \unit{\memone[\rvarone \mapsto t]}}\\
           &=  \unit{\memone[\rvarone \mapsto \parsem \deone \nat (\memone)]}.
  \end{align*}
  From this observation, we deduce that $\memtwo = \memone[\rvarone \mapsto \parsem \deone \nat (\memone)]$. In particular, we have that $\parsem \rvarone \nat(\memtwo) = \memtwo(\rvarone) = \parsem \deone \nat (\memone)$, and from \Cref{rem:detsemmemupdate}, we also conclude that $\parsem \deone \nat(\memtwo) = \parsem \deone \nat(\memone[\rvarone \mapsto \parsem \deone \nat (\memone)]) = \parsem \deone \nat(\memone)$, but this is absurd because we were assuming that $\parsem \rvarone {\nat}(\memtwo)\neq \parsem \deone{\nat}(\memtwo)$.
\end{proof}

\begin{lemma}
  \label{lemma:sepassntech}
  Let $\renvone, \renvtwo$ and $\renvthree$
  be environments such that:
  \begin{itemize}
  \item $\join \renvone \renvtwo \defined$.
  \item $\ternjudg \renvone \reone \tyone$.
  \item $\binjudg{\join \renvone \renvtwo} {\ass \rvarone \reone}$.
  \end{itemize}
  and let $\rsone, \rstwo \in \sem {\join \renvone \renvtwo}{}$,
  such that
  $\rsone_{\join \renvone \renvtwo \to \renvone}= \rstwo_{\join \renvone \renvtwo\to \renvone}$
  it holds that
  $\sem {\ass \rvarone \reone}{}(\rsone)_{\join \renvone \renvtwo \to \renvone \cup \{\rvarone:\tyone \}}=\sem {\ass \rvarone \reone}{}(\rstwo)_{\join \renvone \renvtwo \to \renvone \cup \{\rvarone:\tyone \}}$.
\end{lemma}
\begin{proof}
  We start by expanding $\sem {\ass \rvarone \reone}{\nat}(\rsone)_{\join \renvone \renvtwo \to \renvone \cup \{\rvarone:\tyone \}}$ as $\bind{\bind {\rsone_\nat} {\mem \mapsto \bind{\detsem \reone \nat(\mem)} {t \mapsto \unit{\mem[\rvarone \mapsto t]}}}}{\memtwo \mapsto \unit{\memtwo\restr{\dom(\renvone) \cup \{\rvarone \}}}}$. This rewrites as follows:
  \[
    \bind {\rsone_\nat} {\mem \mapsto \bind{\detsem \reone \nat(\mem)} {t \mapsto \unit{\mem[\rvarone \mapsto t]\restr{\dom(\renvone) \cup \{\rvarone \}}}}},
  \]
  which is equal to:
  \[
    \bind {\rsone_\nat} {\mem \mapsto \bind{\detsem \reone \nat(\mem)} {t \mapsto \unit{\mem\restr{\dom(\renvone)}[\rvarone \mapsto t]}}},
  \]
  By \Cref{lemma:detsemproj}, we deduce that, in turn, it equals:
  \[
    \bind {\rsone_\nat} {\mem \mapsto \bind{\detsem \reone \nat(\mem\restr{\dom(\renvone)})} {t \mapsto \unit{\mem\restr{\dom(\renvone)}[\rvarone \mapsto t]}}},
  \]
  We can also observe that this is equal to:
  \[
    \bind {{(\rsone_{\join \renvone \renvtwo \to \renvone})}_\nat} {\memtwo \mapsto \bind{\detsem \reone \nat(\memtwo)} {t \mapsto \unit{\memtwo[\rvarone \mapsto t]}}}
  \]
  by expanding the definition of projection and applying standard properties of monads. From our premise, we deduce that it is equal to:
  \[
    \bind {{(\rstwo_{\join \renvone \renvtwo \to \renvone})}_\nat} {\memtwo \mapsto \bind{\detsem \reone \nat(\memtwo)} {t \mapsto \unit{\memtwo[\rvarone \mapsto t]}}}.
  \]
  By applying the same steps but in the reverse order, we are able to show the claim that we are aiming at. 
\end{proof}




\section{Cryptographic Probabilistic Separation Logic}\label{sec:logic}
This section is devoted to introducing a variation on
Barthe et al.'s probabilistic separation logic called
\emph{cryptographic probabilistic separation logic} (CPSL in the following).
Compared to PSL, but also to the logic of bunched 
implications (\BI)~\cite{Pym99,OHearnPym99}, CPSL is simpler 
in that the only connectives are multiplicative
and additive conjunctions, i.e. $\sep$ and $\land$.
In addition, instead of working with plain, untyped formulas, we 
work with \emph{well-typed} formulas. Thanks to this choice, the semantics of
separating conjunction is more precise than in PSL,
and as a 
byproduct gives rise to a lighter formulation of program logic rules
in Section~\ref{sec:inference}.


We first introduce the
syntax of \CPSL and its formulas,
then we proceed by defining a semantics for it
based on \emph{efficiently samplable} randomized stores.
Finally, we show that \CPSL enjoys some basic but useful
properties.

\subsection{The Syntax and Semantics of \CPSL}
\label{subsec:bi}
\newcommand{\srenv}[3]{#2\models^{#1}#3}

\begin{figure*}[t]
  \begin{center}
    \begin{align*}
      \sem{\tyf{\ud{\reone}}\renvone}{}\defsym\{{\tyf \rsone \renvone} \in \sem \renvone{}\mid&\;\sem{\ternjudg \renvone\reone\tyone}{}(\rsone)\ind\unif{\tyone}\}\\
      \sem{\tyf{\indp{\reone}{\retwo}}\renvone}{}\defsym\{{\tyf \rsone \renvone}\in \sem \renvone{}\mid&\;\sem{\reone}{}(\rsone)\ind\sem{\retwo}{}(\rsone)\}\\
      \sem{\tyf{\eq{\reone}{\retwo}}\renvone}{}\defsym\{{\tyf \rsone \renvone}\in \sem \renvone{}\mid&\;\sem{\reone}{}(\rsone)=\sem{\retwo}{}(\rsone)\}\\
      \sem{\tyf{\espl{\deone}{\detwo}}\renvone}{}\defsym\{{\tyf \rsone \renvone}\in \sem \renvone{}\mid&\;\forall \nat \in \NN.\forall \mem \in \supp (\rsone_\nat). \parsem{\deone}{\nat}(\mem)=\parsem{\detwo}{\nat}(\mem)\}
    \end{align*}
    
    \begin{tabular}{ll}
      $\srenv\renvone{\stone}{\tyf \apone \renvone}$&iff $\stone\in\sem{\tyf \apone \renvone}{}$ \\
      $\srenv\renvone{\stone}{\tyf \ftrue \renvone}$&always\\
      $\srenv\renvone{\stone}{\tyf \ffalse \renvone}$&never \\
      {$\srenv\renvone{\stone}{\tyf {\tyf {\fone} \renvtwo\land\tyf {\ftwo} \renvthree} \renvone}$} &{iff there are $\sttwo\ext \stone, \stthree\ext \stone$ such that $\srenv\renvtwo{\sttwo}{\tyf {\fone} \renvtwo}$ and $\srenv\renvthree{\stthree}{\tyf {\ftwo} \renvthree}$}\\
      $\srenv \renvone{\stone}{\tyf{\tyf \fone \renvtwo \sep \tyf \ftwo\renvthree}{\renvone}}$&iff there are $\sttwo, \stthree$ such that $\sttwo\comp \stthree\defined$, $\sttwo \comp \stthree \ind\ext \stone$, $\srenv \renvtwo{\sttwo}{\tyf {\fone} \renvtwo}$ and 
                                     $\srenv \renvthree{\stthree}{\tyf {\ftwo} \renvthree}$\\
    \end{tabular}
  \end{center}  
  \caption{Semantics of \CBI}
  \label{fig:logicsem}
\end{figure*}

\subsubsection{Syntax}
We write $\apset$ for the set of atomic formulas as defined in the following grammar:
\begin{align*}
  \apone\bnf\ud{\reone}\midd\indp{\reone}{\reone}\midd\eq{\reone}{\reone} \midd
  \espl\deone \deone
\end{align*}
The formula $\ud{\reone}$, present in \PSL, is interpreted differently here, 
and indicates that the expression $\reone$ evaluates to a distribution that is 
computationally indistinguishable from the uniform one. The three binary atomic 
formulas $\indp{\reone}{\retwo}$, $\eq{\reone}{\retwo}$ and $\espl{\deone} 
{\detwo}$ all speak of indistinguishability between the two expressions 
involved but at different levels, and their semantics will be explained in detail 
later. We also define a type system for atomic formulas in an environment.
Judgments for this type system are in the form $\binjudg \renvone \apone$, and 
are proved valid by the following rules:
\begin{gather*}
  \infer{\binjudg \renvone \ud \reone}{\ternjudg \renvone \reone 
  \tyone}\qquad\qquad
  \infer{\binjudg \renvone \indp \reone \retwo}{
    \ternjudg \renvone \reone \tyone &
    \ternjudg \renvone \retwo \tyone}\\[0.5ex]
    \infer{\binjudg \renvone \eq \reone \retwo}{
    \ternjudg \renvone \reone \tyone &
    \ternjudg \renvone \retwo \tyone}\qquad
    \infer{\binjudg \renvone \espl \deone \detwo}{
    \ternjudg \renvone \deone \tyone &
    \ternjudg \renvone \detwo \tyone}
\end{gather*}

On top of the set of atomic formulas, one can build the set of generic 
formulas, as follows:

\begin{definition}
  For a fixed environment $\renvone$, the set $\tyf \fset \renvone$
  of \emph{formulas for $\renvone$} is the smallest set of expressions closed under the following 
  rules:
  \begin{gather*}
    \infer{\tyf \apone \renvone \in \tyf \fset \renvone}{\binjudg \renvone \apone}\quad
    \infer{\tyf \ftrue \renvone \in \tyf \fset \renvone}{}\quad
    \infer{\tyf \ffalse \renvone \in \tyf \fset \renvone}{}\\[0.5ex]
    \infer{\tyf {\tyf \fone\renvtwo \sep \tyf \ftwo\renvthree } {\renvone} \in \tyf \fset \renvone}
    {\tyf {\fone} \renvtwo \in \tyf \fset \renvtwo &
      \tyf {\ftwo} \renvthree\in \tyf \fset \renvthree &
      \join \renvtwo\renvthree \ext \renvone
    }\\[0.5ex]
    \infer{\tyf {\tyf {\fone} \renvtwo\land\tyf {\ftwo} \renvthree} \renvone \in \tyf \fset \renvone}{\tyf {\fone} \renvtwo \in \tyf \fset \renvtwo &
      \tyf {\ftwo} \renvthree\in \tyf \fset \renvthree &
      \renvtwo\ext \renvone &
      \renvthree\ext \renvone}
  \end{gather*}
  The set of \CPSL formulas $\fset$ is the union of all the sets
  $\tyf \fset \renvone$ over all possible environments $\renvone$.
  \hfill\qed
\end{definition}

As already observed, the binary connectives we consider are just 
conjunctions. On the other hand, formulas are typed and, crucially, the way in 
which a separating conjunction splits the underlying environment is fixed by the 
formula itself.

\subsubsection{Semantics}
The semantics of \CPSL formulas is a family of relations
$\srenv\renvone{}{}$ such that, for every environment $\renvone$, we have that
$\srenv\renvone{}{}\subseteq{\sem\renvone{}}\times\tyf\fset\renvone$.
The full semantics is in Figure \ref{fig:logicsem} \revision{where, in particular,
  with $\ind\ext$ we denote the composition of $\ind$ and $\ext$}.

The role of the three different forms of binary atomic formulas should be now clearer. 
Atomic propositions like $\tyf {\indp{\reone}{\retwo}}\renvone$ capture 
computational indistinguishability: 
${\tyf \rsone\renvone} \models^\renvone \tyf {\indp{\reone}{\retwo}}\renvone$ holds whenever
$\sem {\ternjudg\renvone\reone\tyone}{}({\tyf \rsone\renvone})$ is computationally
indistinguishable from
$\sem {\ternjudg\renvone\retwo\tyone}{}({\tyf \rsone\renvone})$.
The atomic formula $\tyf {\eq{\reone}{\retwo}}\renvone$ captures
a stronger form equivalence: it holds in a state $\rsone$ if and only if
$\sem {\ternjudg \renvone \reone \tyone}{}({\rsone})$ and
$\sem {\ternjudg \renvone \retwo \tyone}{}({\rsone})$ are equal ---
instead of being simply indistinguishable. 
Finally, $\tyf {\espl{\deone}{\detwo}}\renvone$ states
something even stronger:
$\srenv \renvone \rsone {\tyf{\espl{\deone}{\detwo}}\renvone}$
holds if and only if for every $\nat \in \NN$, every sample $\memone$
in the support of the distribution $\rsone_\nat \in \sem \renvone {}$
is such that $\parsem {\ternjudg \renvone \deone \tyone} \nat(\memone)=
\parsem {\ternjudg \renvone \detwo \tyone} \nat(\memone)$.
\revision{
The semantics of $\espl\deone \detwo$ differs from that of
$\eq\reone \retwo$ for two main reasons: first, the atomic proposition
$\espl\deone\detwo$ applies only to \emph{deterministic} expressions,
while $\eq\reone \retwo$ applies also to \emph{probabilistic} expressions.
Second, we have $\tyf \rsone \renvone \in \sem{\tyf{\espl\deone\detwo}\renvone}{}$
when the \emph{values} obtained by computing $\deone$ and $\detwo$ on every sample of
$\supp({\rsone_\nat})$ are identical independently of the security parameter $\nat$;
instead, $\tyf \rsone \renvone \in \tyf{\eq\reone\retwo}\renvone$ requires the \emph{families of distributions}
obtained by evaluating $\reone$ and $\retwo$ on the ensemble $\tyf \rsone \renvone$
are the same.

\begin{example}
  The distribution ensemble $\rsone \in \sem {\rvarone, \rvartwo:\bool}{}$
  that is defined as follows
  \[
    \rsone_\nat \defsym \left\{\{\rvarone \mapsto 0, \rvartwo \mapsto 1\}^{\frac 1 2},
      \{\rvarone \mapsto 1, \rvartwo \mapsto 0\}^{\frac 1 2}\right\}
  \]
  satisfies $\eq\rvarone \rvartwo$ because the marginal distributions of
  $\rvarone$ and $\rvartwo$ are identical, but it does not
  satisfy $\espl \rvarone \rvartwo$ because there is
  a value of the security parameter $\nat$
  and a sample $m\in \supp(\rsone_\nat)$
  where $m(\rvarone)\neq m (\rvartwo)$ (actually this happens for every sample, independently of $\nat$).\hfill\qed
\end{example}
}

The interpretation of the additive conjunction is the same as in 
\PSL except that it asks that the sub-formulas hold in specific parts of 
the stores. Formally, we have that $\srenv \renvone {\tyf \rsone\renvone} 
{\tyf{\tyf\fone\renvtwo
    \land \tyf\fone\renvthree}{\renvone}}$ whenever
there are $\rstwo,\rsthree \ext \rsone$ such that: 
$\srenv \renvtwo {\tyf \rstwo \renvtwo} {\tyf \fone \renvtwo}$
and $\srenv \renvthree {\tyf \rsthree \renvthree} {\tyf \ftwo \renvthree}$.

The crux of \CPSL is the way separating conjunction is interpreted, which
now refers to computational independence: we ask that $\srenv \renvone
{\tyf \rsone\renvone} {\tyf{\tyf\fone\renvtwo \sep \tyf\fone\renvthree}{\renvone}}$
holds whenever it is possible to find two ensembles $\rstwo \in \sem \renvtwo{}$
and $\rsthree\in \sem \renvthree{}$
such that \(
  \srenv \renvtwo {\tyf \rstwo \renvtwo} {\tyf \fone \renvtwo},
  \srenv \renvthree {\tyf \rsthree \renvthree} {\tyf \ftwo \renvthree}
    \), and
the projection of $\tyf \rsone\renvone$ on the
environment $\join \renvtwo \renvthree$ is indistinguishable
from the tensor product of ${\tyf \rstwo \renvtwo}$ and
$\tyf \rsthree \renvthree$.
As argued in \Cref{sec:compindep} 
below, the last condition is a characterization of the standard notion of
computational independence given by Fay~\cite{Fay14}.
\revision{
  So, our interpretation of the separating conjunction --- which is
  based on that of \BI\ --- exactly captures computational independence.
}


\paragraph*{Notational Conventions}
In the following, we often omit $\renvone$ from $\models^\renvone$
if it is evident from the context, from $\tyf \apone\renvone$
when $\apone$ mentions exactly the variables of $\dom(\renvone)$,
And from
$\tyf{\tyf \fone \renvtwo \odot \tyf \ftwo \renvthree}\renvone$
for $\odot \in \{\land,\sep\}$ when
$\renvone$ is exactly $\join \renvtwo \renvone$.
In order to simplify the presentation of the formulas,
we write $\renvone \restr S$ for the environment obtained by restricting 
the domain of $\renvone$ to $S$. We also omit $\renvone$ when it is
clear from the context. In particular, when
$\stone=\tyf\rsone \renvone$ and $\srenv \renvone{\stone}{\tyf \fone \renvone}$
holds, we write simply $\sr \stone \fone$.
Finally, we write
$
  \tyf \fone \renvone \models^\renvone \tyf \ftwo \renvone
$
as a shorthand for:
\[
  \forall \stone \in \sem\renvone{}. \left(\srenv \renvone \stone \tyf \fone \renvone\right) \Rightarrow \srenv \renvone \stone \tyf \ftwo \renvone.
\]

\subsubsection{Algebraic Properties of Efficiently Samplable Distribution Ensembles}

We model \CPSL on efficiently samplable distribution ensembles;
this set, in turn enjoys some remarkable algebraic properties with respect to
the tensor product $\tensprod$ and to the relations $\ind$
and $\ext$. In the following, we review some of them
in order to prepare the ground for showing, in the \Cref{sec:cpslvsbi},
that the structure $(\stset,\tensprod,\emst{},\ind\ext)$ is a \emph{partial 
Kripke resource monoid} (\Cref{def:pkm}). This result, in turn,
allows us to compare the semantics of our logic
with the standard interpretation the conjunctive fragment
of BI on \emph{partial Kripke resource monoids}
in \Cref{prop:cpsltobi,prop:bitocpsl}. 

\begin{definition}
  \label{def:pkm}
  A partial Kripke resource monoid $(\monone,\mulone,\id,\poone)$ 
  consists of a set $\monone$ together with a partial binary operation 
  $\mulone$, an element $\id\in\monone$ and a pre-order $\poone$ 
  satisfying the following:
  \begin{varenumerate}
  \item
    For every $\elone\in\monone$, both $\elone\mulone\id$ and 
    $\id\mulone\elone$ are defined, and equal to $\elone$.
  \item
    The operation $\mulone$ is associative, in the following sense:
    $(\elone\mulone\eltwo)\mulone\elthree$ and
    $\elone\mulone(\eltwo\mulone\elthree)$ are either both undefined or 
    both defined, and in the latter case, they are equal.
  \item
    The pre-order $\poone$ is compatible with $\mulone$, i.e., 
    if $\elone\poone\eltwo$ and 
    $\elone\mulone\elthree,\eltwo\mulone\elthree$ are both defined, 
    then $\elone\mulone\elthree\poone\eltwo\mulone\elthree$ (and 
    similarly for $\elthree\mulone\elone,\elthree\mulone\eltwo$).
    \hfill\qed
  \end{varenumerate}
\end{definition}

Showing that the relation $\ind\ext$, is a pre-order and is compatible with
respect to $\comp$, as required by \Cref{def:pkm}, requires appreciable effort.

\paragraph{Properties of projection.} We start with the projection operation on stores. We start by observing that the projection $\cdot_{\renvone\to\renvone}$ is the identity.

\begin{rem}[Identity]
  \label{rem:identity}
  For every $\rsone \in \sem \renvone{}$,  $\rsone_{\renvone\to\renvone}= \rsone$.
\end{rem}
\begin{proof}
  \[
    (\rsone_{\renvone\to\renvone})_\nat = \bind {\rsone_\nat}{ \memone \mapsto \unit{\memone\restr {\dom (\renvone)}}} = \bind {\rsone_\nat}{ \memone \mapsto \unit{\memone}} = \rsone.
  \]
\end{proof}

We also observe that projections are closed under composition, when this operation is defined.

\begin{lemma}[Closure Under Composition of Projections]
  \label{lemma:exttrans1}
  Whenever $\renvone_1\ext \renvone_2\ext \renvone_3$, for every $\rsone_3\in \sem{\renvone_3}{}$, it holds that:
  $$((\rsone_3)_{\renvone_3 \to \renvone_2})_{\renvone_2 \to \renvone_1}=(\rsone_3)_{\renvone_3 \to\renvone_1}$$
\end{lemma}
\begin{proof}
  $((\rsone_3)_{\renvone_3 \to \renvone_2})_{\renvone_2 \to \renvone_1}$
  equals
  \[
    \nat\mapsto \bind{\bind{{\rsone_3}_\nat}{\memone \mapsto \unit{\memone\restr{\dom(\renvone_2)}}}}{\memone \mapsto \unit{\memone\restr{\dom(\renvone_1)}}}
  \]
  which, in turn, equals:
  \[
    \nat\mapsto \bind{{\rsone_3}_\nat} {\memone \mapsto \unit{{\memone\restr{\dom(\renvone_2)}}\restr{\dom(\renvone_1)}}}.
  \]
  Observe that ${\memone\restr{\dom(\renvone_2)}}\restr{\dom(\renvone_1)}=\memone\restr{\dom(\renvone_1)}$. This concludes the proof.
\end{proof}

\paragraph{Properties of $\ext$} The extension relation $\ext$ is defined on top of the notion of projection; for this reason, \Cref{rem:identity} and \Cref{lemma:exttrans1} can be used to show that $\ext$ is a pre-order. 

\begin{lemma}[Pre-order]
  \label{lemma:extpo}
  The relation $\ext$ is a pre-order relation on $\stset$.
\end{lemma}

\begin{proof}~
  \begin{itemize}
  \item For reflexivity, we need to show that
    \[
      {\tyf \rsone \renvone}\ext{\tyf \rsone \renvone}.
    \]
    In particular, this requires to verify that $\renvone\ext \renvone$ and that
    $\rsone_{\renvone\to\renvone}=\rsone$, that is \Cref{rem:identity}.
    \item For transitivity, we need to show that:
    \begin{gather*}
      \tyf{\rsone_1}{\renvone_1}\ext\tyf{\rsone_2}{\renvone_2}\ext\tyf{\rsone_3}{\renvone_3} \Rightarrow\\
      \tyf{\rsone_1}{\renvone_1}\ext\tyf{\rsone_3}{\renvone_3}.
    \end{gather*}
    The conclusion on the environments is a consequence of the transitivity of $\subseteq$.
    Then we need to show that
    $(\rsone_3)_{\renvone_3 \to \renvone_1} = \rsone_1$. We also know that
    $(\rsone_3)_{\renvone_3 \to \renvone_2} = \rsone_2$ and that
    $(\rsone_2)_{\renvone_2 \to \renvone_1} = \rsone_1$; finally,
    due to $\renvone_1\ext\renvone_2$, we also know that $\left((\rsone_3)_{\renvone_3 \to \renvone_2}\right)_{\renvone_2 \to \renvone_1}=(\rsone_3)_{\renvone_3 \to \renvone_1}$. Therefore:
    \[
      \rsone_1=(\rsone_2)_{\renvone_2 \to \renvone_1}=((\rsone_3)_{\renvone_3 \to \renvone_2})_{\renvone_2 \to \renvone_1}=(\rsone_3)_{\renvone_1}.
    \]
    The last equivalence is a consequence of Lemma \ref{lemma:exttrans1}.
  \end{itemize}
\end{proof}

\paragraph{Properties of $\ind$.} We start by showing that $\ind$ is an equivalence relation (Lemma \ref{lemma:indeq}). This observation is crucial for showing that the composition of $\ind$ and $\ext$ is a pre-order. Then, we show that $\ind$ is compatible with the tensor product (\Cref{lemma:indcompcomp}), and that $\ind$ and $\ext$ enjoy a form of distributivity (\Cref{lemma:distributivity}). These properties are respectively used to show that the composition of $\ind$ and $\ext$ is compatible with the tensor product, and that it is transitive, and these are both crucial properties for showing that $(\stset,\tensprod,\emst{},\ind\ext)$ is a \emph{partial 
  Kripke resource monoid}.

\begin{lemma}
  \label{lemma:indeq}
  The relation $\ind$ is an equivalence on $\stset$.
\end{lemma}
\begin{proof}~
  \begin{itemize}
  \item Reflexivity is trivial.
  \item For symmetry we need to show that:
    \begin{gather*}
      {\tyf \rsone \renvone}\ind{\tyf \rstwo \renvone}\Rightarrow
      {\tyf \rstwo \renvone}\ind{\tyf \rsone \renvone},
    \end{gather*}
    which is a consequence of the fact that
    for every $\nat \in \NN$ and
    every polytime distinguisher
    $\distone:\detsem{\renvone}{\nat} \to \detsem\BB\nat$:
    \begin{gather*}
      \left|\Prob[\memone \leftarrow \rsone_\nat]\left[\distone(1^\nat, \memone)=1\right]-\Prob[\memone \leftarrow \rstwo_\nat]\left[\distone(1^\nat, \memone)=1\right]\right|=\\
      \left|\Prob[\memone \leftarrow \rstwo_\nat]\left[\distone(1^\nat, \memone)=1\right]-\Prob[\memone \leftarrow \rsone_\nat]\left[\distone(1^\nat, \memone)=1\right]\right|.
    \end{gather*}
  \item For transitivity we need to show that:
    \begin{gather*}
      \tyf{\rsone_1}\renvone\ind\tyf{\rsone_2}\renvone\land
      \tyf{\rsone_2}\renvone\ind\tyf{\rsone_3}\renvone\Rightarrow\\
      \tyf{\rsone_1} \renvone\ind \tyf{\rsone_3}\renvone.
    \end{gather*}
    It holds that:
    \begin{align*}
      &\left|\Prob[\memone \leftarrow {\rsone_1}_\nat]\left[\distone(1^\nat, \memone)=1\right]-\Prob[\memone \leftarrow {\rsone_3}_\nat]\left[\distone(\nat, \memone)=1\right]\right|=\\
      &\begin{multlined}\big|\Prob[\memone \leftarrow {\rsone_1}_\nat]\left[\distone(1^\nat, \memone)=1\right]-\Prob[\memone \leftarrow {\rsone_3}_\nat]\left[\distone(\nat, \memone)=1\right]+\\
        \Prob[\memone \leftarrow {\rsone_2}_\nat]\left[\distone(1^\nat, \memone)=1\right]-\Prob[\memone \leftarrow {\rsone_2}_\nat]\left[\distone(\nat, \memone)=1\right]\big|\le\end{multlined}\\
      &\begin{multlined}\left|\Prob[\memone \leftarrow {\rsone_1}_\nat]\left[\distone(1^\nat, \memone)=1\right]-\Prob[\memone \leftarrow {\rsone_2}_\nat]\left[\distone(\nat, \memone)=1\right]\right|+\\
        \left|\Prob[\memone \leftarrow {\rsone_3}_\nat]\left[\distone(1^\nat, \memone)=1\right]-\Prob[\memone \leftarrow {\rsone_2}_\nat]\left[\distone(\nat, \memone)=1\right]\right|\le\end{multlined}\\
      &\negl_1(\nat)+\negl_2(\nat)
    \end{align*}
    We know that $\negl_1$ and $\negl_2$ are two negligible functions because of the
    indistinguishability of $\tyf{\rsone_1}\renvone$ and $\tyf{\rsone_2}\renvone$
    and that of $\tyf{\rsone_2}\renvone$ and $\tyf{\rsone_3}\renvone$.
    The sum of two negligible functions is a negligible function. 
  \end{itemize}  
\end{proof}

\newcommand{\samplone}{\mathit R}

\begin{lemma}[Compatibility of $\ind$ and $\comp$]
  \label{lemma:indcompcomp}
  The relation $\ind$ and the function $\comp$ are compatible, i.e.
  \[
    \forall \stone,\sttwo, \stthree. \stone\ind\sttwo \Rightarrow (\stone\comp \stthree)\defined\Rightarrow(\sttwo\comp \stthree)\defined\Rightarrow   (\stone\comp \stthree)\ind(\sttwo\comp \stthree)
  \]
  and
  \[
    \forall \stone,\sttwo, \stthree. \stone\ind\sttwo \Rightarrow (\stthree\comp \stone)\defined\Rightarrow(\stthree\comp \sttwo)\defined\Rightarrow   (\stthree\comp \stone)\ind(\stthree\comp \sttwo)
  \]
\end{lemma}
\begin{proof}
  We only show the right compatibility. The proof of left compatibility is analogous. Assume the statement not to hold, i.e. there are three stores $\stone={\tyf {\rsone_1} {\renvone_1}},\sttwo = {\tyf {\rsone_2} {\renvone_2}}, \stthree = {\tyf {\rsone_3} {\renvone_3}}$ such that $\stone\ind\sttwo$, $(\stone\comp \stthree)$ is defined, $(\sttwo\comp \stthree)$ is defined and finally that there is a distinguisher $\distone$ which distinguishes $\stone\comp \stthree$ and $\stone \comp \sttwo$ with non-negligible advantage.
  We know that $\rsone_3$ is efficiently samplable for every $\nat\in \NN$, thus, we can call $\samplone$ the probabilistic polytime algorithm sampling $(\rsone_3)_\nat$. In particular, the distinguisher $\disttwo$ can be obtained by sampling from the poly-time algorithm for $\rsone_3$, namely $\samplone$ to obtain a sample $\memtwo$. Then, $\disttwo$ combines its input sample $\memone$ and $\memtwo$ in a unique sample and then invokes $\distone$ on this sample. By polynomiality of $\distone$ and of $\samplone$, $\disttwo$ is polynomial. Moreover, the advantage of $\disttwo$ in distinguishing $\rsone_1$ and $\rsone_2$ is the same of $\distone$ on $(\rsone_1\tensprod \rsone_3)$ and $(\rsone_2\tensprod \rsone_3)$.
\end{proof}

\begin{lemma}[Distributivity of $\ind$ and $\ext$]\label{lemma:distributivity}
  If $\stone\ext\sttwo\ind\stthree$, then there exists $\stfour$ such that
  $\stone\ind\stfour\ext\stthree$.
\end{lemma}
\begin{proof}
  The claim we are aiming to show is:
  for every $\stone,\sttwo,\stthree\in\stset$,
  $\stone\ext\sttwo\ind\stthree$, then there is $\stfour\in\stset$ such that
  $\stone\ind\stfour\ext\stthree$. This claim can be equivalently restated as follows:
  \begin{gather*}
    {\tyf {\rsone_1} {\renvone_1}}\ext{\tyf {\rsone_2} {\renvone_2}}\ind {\tyf {\rsone_3} {\renvone_3}} \Rightarrow\\
    \exists \renvone_4,\rsone_4.
    {\tyf {\rsone_1} {\renvone_1}}\ind{\tyf {\rsone_4} {\renvone_4}}\ext {\tyf {\rsone_3} {\renvone_3}}
  \end{gather*}
  Choosing $\renvone_4=\renvone_1$ and $\rsone_4=(\rsone_3)_{\renvone_3\to\renvone_1}$,
  we are able to show the claim.
  \begin{itemize}
  \item To show $\renvone_4\ext \renvone_3$, we can show that $\renvone_1\ext \renvone_3$. From our assumption, we know that $\renvone_1\ext \renvone_2 = \renvone_3$, so $\renvone_1\ext \renvone_3$.
  \item Similarly, to show $\rsone_4=(\rsone_3)_{\renvone_3\to\renvone_4}$ it suffices to observe that $\renvone_4=\renvone_1$.
  \end{itemize}
  This concludes the proof that ${\tyf {\rsone_4} {\renvone_4}}\ext {\tyf {\rsone_3} {\renvone_3}}$. It remains to show that:
  \begin{equation}
    {\tyf {\rsone_1} {\renvone_1}}\ind\tyf{(\rsone_3)_{\renvone_3\to\renvone_1}}{\renvone_1}
    \tag{$*$}
  \end{equation}
  Assume that $(\rsone_s)$ and $(\rsone_3)_{\renvone_3\to\renvone_1}$ are distinguishable by $\distone$, from Definition \ref{def:ensind}, this is equivalent to assuming that:
  \[
    \left|\Prob[\memone \leftarrow{\rsone_1}_\nat]\left[\distone(1^\nat, \memone)=1\right]-\Prob[\memone \leftarrow{(\rsone_3)_{\renvone_3\to\renvone_1}}_\nat]\left[\distone(1^\nat, \memone)=1\right]\right|
  \]
  is not a negligible function. Let the distinguisher $\disttwo$ which take its input $1^\nat$ and a sample $\memtwo\in \detsem {\renvone_3}\nat$ and executes $\distone$ on $(\memtwo\restr{\dom(\renvone_1)}, 1^\nat)$. $\disttwo$ is still a probabilistic polynomial time computable function. Due to  $\renvone_1\ext \renvone_3$, the probability of returning 1 of $\disttwo$ on samples taken from  $\rsone_3$ is exactly that of $\distone$ on samples taken from $(\rsone_3)_{\renvone_3\to\renvone_1}$, and similarly for $\rsone_2$, the output distribution is the same of $\distone$ on $(\rsone_2)_{\renvone_2\to\renvone_1}=\rsone_1$. 
  For this reason:
  \[
    \left|\Prob[\memone \leftarrow {(\rsone_3)}_\nat]\left[\disttwo(1^\nat, \memone)=1\right]-\Prob[\memone \leftarrow {(\rsone_2)}_\nat]\left[\disttwo(1^\nat, (\rsone_2)(\nat))=1\right]\right|
  \]
  is not a negligible function, but this is absurd since we were assuming that:
  \[
    {\tyf {\rsone_2} {\renvone_2}}\ind {\tyf {\rsone_3} {\renvone_3}},
  \]
  so $\distone$ cannot exist, and thus $(*)$ holds.
\end{proof}

\paragraph{Compatibility of $\ext$ and $\ind$.} The last property that we need to show is the compatibility of $\ext$ and $\ind$. With \Cref{lemma:indcompcomp}, this property is useful to show that even the relation $\ind\ext$ is compatible with $\tensprod$, and enabling us to show that  $(\stset,\tensprod,\emst{},\ind\ext)$ is a \emph{partial Kripke resource monoid}. The compatibility of $\ext$ and $\ind$ is shown in \Cref{lemma:extcompcomp} below, and it relies on showing that the projection is distributive with the tensor product (\Cref{lemma:prodrestrdistr}).

\begin{lemma}[Distributivity of Product over Restriction]
  \label{lemma:prodrestrdistr}
  Whenever $\rsone \in \sem \renvone{}$, $\rstwo \in \sem \renvtwo{}$, $\renvthree\ext \renvone, \renvfour\ext\renvtwo$, and $\join\renvtwo\renvone\defined$, it holds that:
    \[
      (\tyf \rstwo \renvtwo \tensprod \tyf \rsone \renvone)_{\join {\renvtwo}{\renvone}\to(\join {\renvthree}{\renvfour})} = ({\tyf\rstwo \renvtwo}_{{\renvtwo}\to {\renvthree}} \tensprod{\tyf \rsone\renvone}_{{\renvone}\to{\renvfour}}).
    \]  
  \end{lemma}
\begin{proof}
  Showing the identity of the final environments is trivial, so we will focus only on the distribution ensembles. We start by expanding the definition of $(\tyf \rstwo \renvtwo \tensprod \tyf \rsone \renvone)_{\join {\renvtwo}{\renvone}\to(\join {\renvthree}{\renvfour})}$ for a given $\nat$:
  \begin{align*}
      &= \bind{(\rsone \tensprod \rstwo)_\nat} {\memone \mapsto \unit{\memone\restr{\dom(\join {\renvthree}{\renvfour})}}}\\
      &= \memfour\mapsto\;\;\smashoperator[lr]{\sum_{\memone \in \detsem{\join {\renvtwo}{\renvone}}\nat}}\;\;{\rstwo_\nat(m\restr{\dom(\renvtwo)})\cdot\rsthree_\nat(m\restr{\dom(\renvone)})} \cdot{\unit{\memone\restr{\dom(\join {\renvthree}{\renvfour})}}}(\memfour)\\
      &= \memfour\mapsto\;\;\smashoperator[lr] {\sum_{
              \memtwo \in \detsem{{\renvtwo}}\nat,
              \memthree \in \detsem{{\renvone}}\nat
}}\;\;{\rstwo_\nat(\memtwo)\cdot\rsthree_\nat(\memthree)} \cdot{\unit{(\memtwo\cup\memthree)\restr{\dom(\join {\renvthree}{\renvfour})}}}(\memfour)\\
      &= \memfour\mapsto\;\;\smashoperator[lr] {\sum_{
              \memtwo \in \detsem{{\renvtwo}}\nat,
              \memthree \in \detsem{{\renvone}}\nat
}}\;\;{\rstwo_\nat(\memtwo)\cdot\rsthree_\nat(\memthree)} \cdot{\unit{(\memtwo\restr{\dom(\renvthree)}\cup\memthree\restr{\dom(\renvfour)})}}(\memfour)
  \end{align*}
  Observe that ${\unit{(\memtwo\restr{\dom(\renvthree)}\cup\memthree\restr{\dom(\renvfour)})}}(\memfour)$ is equal to:
  \[
    {\unit{\memtwo\restr{\dom(\renvthree)}}(\memfour\restr{\dom(\renvthree)})\cdot\unit{\memthree\restr{\dom(\renvfour)}}}(\memfour\restr{\dom(\renvfour)})
  \]
  From this observation, the conclusion comes easily, by associativity property and closing back the definition of projection.
\end{proof}

\begin{lemma}[Compatibility of $\ext$ and $\comp$]
  \label{lemma:extcompcomp}
  The relation $\ext$ and the function $\comp$ are compatible, i.e.
  \[
    \forall \stone,\sttwo, \stthree. \stone\ext\sttwo \Rightarrow (\stone\comp \stthree)\defined\Rightarrow(\sttwo\comp \stthree)\defined\Rightarrow (\stone\comp \stthree)\ext(\sttwo\comp \stthree)
  \]
  and
  \[
    \forall \stone,\sttwo, \stthree. \stone\ext\sttwo \Rightarrow (\stthree\comp \stone)\defined\Rightarrow(\stthree\comp \sttwo)\defined\Rightarrow (\stthree\comp \stone)\ext(\stthree\comp \sttwo)
  \]
\end{lemma}
\begin{proof}
  We only show the right compatibility, the proof of left compatibility is analogous. Assume that there are three stores $\stone={\tyf {\rsone_1} {\renvone_1}},\sttwo = {\tyf {\rsone_2} {\renvone_2}}, \stthree = {\tyf {\rsone_3} {\renvone_3}}$ such that $\stone\ext\sttwo$, $(\stone\comp \stthree)$ is defined, and $(\sttwo\comp \stthree)$ is defined. We want to show that:
  \begin{equation}
    \join {\renvone_1}{\renvone_3} \ext\join {\renvone_2}{\renvone_3},
    \tag{Claim 1}
  \end{equation}
  and that:
  \begin{equation}
     \left(\rsone_2\tensprod \rsone_3\right)_{\join {\renvone_2}{\renvone_3}\to\join {\renvone_1}{\renvone_3}}=\rsone_1\tensprod \rsone_3.
    \tag{Claim 2}
  \end{equation}
  Claim 1 is trivial. Claim 2 is a consequence of Lemma \ref{lemma:prodrestrdistr}, which allows us to show:
  \[
    \left(\rsone_2\tensprod \rsone_3\right)_{\join {\renvone_2}{\renvone_3}\to\join {\renvone_1}{\renvone_3}}={(\rsone_1)}_{\renvone_2\to\renvone_1}\tensprod (\rsone_3)_{\renvone_3\to\renvone_3}.
  \]
  By assumption we know that $\left(\rsone_2\right)_{{\renvone_2}\to{\renvone_1}}=\rsone_1$, from Lemma \ref{lemma:extpo} we conclude that
  \[
    (\rsone_3)_{\renvone_3\to\renvone_3}=\rsone_3.
  \]
\end{proof}

\subsection{Some Relevant Classes of Formulas}
\label{sec:cof}

Some properties of \CPSL we 
prove
in Section \ref{sec:propofcbi} below only hold for specific classes of formulas 
within $\fset$, which need to be appropriately defined.
We start by the notion of \emph{exact formula}:
intuitively, these formulas describe
\emph{statistical} properties of distributions
rather than \emph{computational} ones.

We write $\exfset$ for the set of exact formulas, i.e., the smallest
subset of $\fset$ that for every $\renvone$ contains
the formulas $\tyf \top\renvone$, 
$\tyf \bot \renvone$, $\tyf{\eq \reone \retwo} \renvone$,
$\tyf{\espl \deone \detwo} \renvone$ and
which is closed under additive conjunction. This means that
if $\tyf \fone \renvone$ and $\tyf \ftwo \renvtwo$
are exact,
$\renvone \ext \renvthree$ and $ \renvtwo \ext \renvthree$,
then $\tyf{\tyf \fone \renvone \land \tyf \ftwo \renvtwo}\renvthree$
is also exact.


Crucially, the semantics of exact formulas is not
closed with respect to $\ind$, as shown by the following example:

\begin{example}
  The stores $\rsone = \left \{ \left\{ (\rvarone \mapsto 0^\nat)^{1} \right\} \right\}_{\nat\in \NN}$
  and $\rstwo = \left \{ \left\{  (\rvarone \mapsto 0^\nat)^{1-2^{-\nat}}, (\rvarone \mapsto 1^{\nat})^{2^{-\nat}}\right\}\right\}_{\nat\in \NN}$
  are computationally indistinguishable, but they do not satisfy the same formulas.
  \revision{
  In particular, we have $\sr \rsone {\eq {\setzero_\nat}{\rvarone}}$, \emph{but not}
  $\sr \rstwo {\eq {\setzero_\nat}{\rvarone}}$. More precisely, the distribution
  corresponding to the expression ${\setzero_\nat}$ is the Dirac distribution
  centered on $0^\nat$, which is exactly the same distribution obtained by evaluating
  $\rvarone$ in $\rsone$, but not in $\rstwo$.
  }\hfill\qed
\end{example}

\revision{
  Also notice that, in the previous example,
  if we took ${\indp {\setzero_\nat}{\rvarone}}$ instead of ${\eq {\setzero_\nat}{\rvarone}}$,
  we would have $\sr \rsone {\indp {\setzero_\nat}{\rvarone}}$, \emph{and}
  $\sr \rstwo {\indp {\setzero_\nat}{\rvarone}}$. Specifically, this would hold because
  the advantage of any distinguisher for the family of distributions
  $\{\sem {\setzero_\nat} \nat(\rsone_\nat)\}_{\nat\in \NN}$
  and $\{\sem {\rvarone} \nat(\rsone)\}_{\nat \in \NN}$ is 0, and
  that of every adversary for the same two expressions evaluated in $\rstwo$
  is bounded by $2^{-\nat}$, which is a negligible function.}
Unlike \emph{exact formulas}, formulas like $\indp \cdot \cdot$
do not describe exact statistical properties
of distribution ensembles, but they express computational
properties. For this reason, we call these formulas
\emph{approximate}. 
The set of \emph{approximate formulas} is $\apfset$, and it is defined as the largest subset of $\fset$ where the atomic formulas 
$\eq \reone \retwo$ and
$\espl\deone\detwo$ do not occur as subformulas.
Due to their computational nature, the semantics of
\emph{approximate formulas} is closed with respect to $\ind$.
This is shown in \Cref{lemma:indcbi} below.

\subsection{Properties of \CPSL}
\label{sec:propofcbi}

In this section, we discuss the meta-theoretical
properties of \CPSL. In particular its semantics
enjoys some remarkable closure properties. The most important
is a form of monotonicity with respect to $\ext$. Intuitively,
this property states that when two states $\stone, \sttwo$
are such that $\stone \ext \sttwo$,
$\stone$ and $\sttwo$ satisfy \emph{related} formulas.
It is not true that they satisfy
the \emph{same} set of formulas, simply because
our semantics
relates states and formulas \emph{on the same environment},
so states that are defined on different domains
cannot be in relation with the same formulas.
However, we can say that whenever two states $\stone, \sttwo$
are such that $\stone \ext \sttwo$ and $\stone$
satisfies $\tyf \fone \renvone$, then $\sttwo$
satisfies a formula $\tyf \ftwo \renvtwo$
which is structurally identical to
$\tyf \fone \renvone$, but which carries different environments.
This relation between formulas is induced by the
store extension relation $\ext$ as follows:
\[
  \infer{\tyf\apone\renvone \ext \tyf \apone \renvtwo}{\renvone \ext \renvtwo}{}\quad
  \infer{\tyf{\tyf {\fone} {\renvone}\sep\tyf {\ftwo} {\renvtwo}}{\renvthree_1} \ext \tyf{\tyf {\fone} {\renvone}\sep\tyf {\ftwo} {\renvtwo}}{\renvthree_2}}
  {
    {\renvthree_1} \ext{\renvthree_2} 
  }
\]
\[
  \infer{\tyf{\tyf {\fone_1} {\renvone_1}\land\tyf {\ftwo_1} {\renvtwo_1}}{\renvthree_1} \ext \tyf{\tyf {\fone_2}{\renvone_2}\land\tyf {\ftwo_2} {\renvtwo_2}}{\renvthree_2}}
  {\tyf{\fone_1} {\renvone_1} \ext \tyf{\fone_2} {\renvone_2} &
    \tyf{\ftwo_1} {\renvtwo_1} \ext \tyf{\ftwo_2} {\renvtwo_2} &
    \renvthree_1\ext\renvthree_2
  }
\]
By leveraging this relation, we can state
the property we were discussing above as follows:

\begin{lemma}
  \label{lemma:prekm}
  For every pair of environments $\renvone, \renvtwo$, pair of formulas $\tyf \fone \renvone \ext \tyf\ftwo \renvtwo{}$ and pair of states $\stone\in \sem \renvone {}, \sttwo \in \sem \renvtwo{}$ such that $\stone \ext \sttwo$, it holds that $\srenv \renvone \stone {\tyf \fone \renvone} \Leftrightarrow \srenv \renvtwo \sttwo {\tyf \ftwo \renvtwo}$.
    \hfill\qed
\end{lemma}
\noindent
The following observations turn out to be useful in the proof of \Cref{lemma:prekm}:

\begin{rem}
  \label{rem:extftoexts}
  $\tyf \fone \renvone\ext \tyf \ftwo\renvtwo \Rightarrow \renvone \ext \renvtwo$.
\end{rem}

\begin{rem}
  \label{rem:suppext}
  If two distribution ensembles $\rsone\in \sem \renvone{},
  \rstwo \in \sem \renvtwo{}$ are such that $\rsone \ext \rstwo$,
  then for every $\nat \in \NN$,
  $\supp (\rsone_\nat) = \{ \memone \in \detsem \renvone \nat \mid \exists \memtwo \in \supp (\rstwo). \memtwo\restr{\dom(\renvone)}= \memone\}$.
\end{rem}
\begin{proof}
  From $\rsone \ext \rstwo$, we deduce that
  for every $\nat \in \NN$, $\rsone_\nat = \bind {\rstwo_\nat}{\memone \mapsto \memone\restr{\dom(\renvone)}}$,
  so the claim becomes
  \begin{equation*}
    \supp (\bind {\rstwo_\nat}{\memone \mapsto \unit{\memone\restr{\dom(\renvone)}}}) =\\
    \{ \memone \in \detsem \renvone \nat \mid \exists \memtwo \in \supp (\rstwo). \memone\restr{\dom(\renvone)}= m\},
  \end{equation*}
  by definition of $\supp(\cdot)$, this is equivalent to:
  \begin{equation*}
    \left\{\memone \in \detsem \renvone \nat \mid \;\;\smashoperator[lr]{\sum_{\memtwo \in \supp(\rstwo)}} \;\;\rstwo_\nat(\memtwo)\cdot \unit{\memtwo\restr{\dom(\renvone)}}(\memone)> 0\right\} =\\
    \{ \memone \in \detsem \renvone \nat \mid \exists \memtwo \in \supp (\rstwo). \memtwo\restr{\dom(\renvone)}= \memone\}.
  \end{equation*}
  In turn, this rewrites as follows:
      \begin{equation*}
    \left\{\memone \in \detsem \renvone \nat \mid \;\;\;\;\smashoperator[lr] {\sum_{\memtwo \in \supp(\rstwo)\land \memtwo\restr{\dom(\renvone)}=\memone}} \;\;\;\;\rstwo_\nat(\memtwo)> 0\right\} =\\
    \{ \memone \in \detsem \renvone \nat \mid \exists \memtwo \in \supp (\rstwo). \memtwo\restr{\dom(\renvone)}= \memone\}.
  \end{equation*}
  The equality of these expressions is trivial. 
\end{proof}

\begin{proof}[Proof of \Cref{lemma:prekm}]
  By induction on the proof that $\tyf \fone \renvone \ext \tyf\ftwo \renvtwo{}$.
  \begin{itemize}
  \item If these are atomic formulas, we go by cases on the atomic proposition. We start by assuming that the proposition is $\ud\reone$, and we start from the left-to-right implication. We assume $\srenv \renvone \stone {\tyf{\ud\reone} \renvone}$, which means that $\sem {\ternjudg \renvone\reone\tyone}{}(\stone) {} \ind \unif{\tyone}$. From $\stone \ext \sttwo$ and \Cref{lemma:projexpsem}, we deduce that $\sem {\ternjudg \renvtwo\reone\tyone}{}(\sttwo) {} =\sem {\ternjudg \renvone\reone\tyone}{}(\stone) {}\ind \unif{\tyone}$, which concludes the proof. For the opposite direction, if we assume $\sem {\ternjudg \renvtwo\reone\tyone}{}(\sttwo) {} \ind \unif{\tyone}$, we can apply \Cref{lemma:projexpsem} and conclude that $\sem {\ternjudg \renvone\reone\tyone}{}(\stone) {}=\sem {\ternjudg \renvtwo\reone\tyone}{}(\sttwo) {}\ind \unif{\tyone}$.
    If the proposition is $\eq \reone \retwo$, from $\stone \ext \sttwo$ and \Cref{lemma:projexpsem}, we deduce that $\sem {\ternjudg \renvtwo\reone\tyone}{}(\sttwo) {} =\sem {\ternjudg \renvone\reone\tyone}{}(\stone) {}$ and that $\sem {\ternjudg \renvtwo\retwo\tyone}{}(\sttwo) {} =\sem {\ternjudg \renvone\retwo\tyone}{}(\stone) {}$, so in both the directions the conclusion is a consequence of the transitivity of identity.
    The case where the proposition is $\indp \reone \retwo$ is analogous, but applying the transitivity of $\ind$ (\Cref{lemma:indeq}).
    Finally, we assume that the proposition is $\espl \deone \detwo$. For the left-to-right implication, we assume that $\rsone \in \sem \renvone{}$ and $\rstwo \in \sem \renvtwo{}$ with $\sr \rsone \tyf {\espl \deone \detwo}\renvone$ and $\rsone \ext \rstwo$. The goal is to show $\sr \rstwo \tyf {\espl \deone \detwo}\renvtwo$. Assume this not to hold. This means that there are $\overline \nat$ and  $\overline \mem\in \supp(\rstwo_{\overline \nat})$ such that $\parsem\deone{\overline \nat}(\overline \mem)\neq\parsem\detwo{\overline \nat}(\overline \mem)$. Then, we apply \Cref{rem:suppext} to observe that $\overline \mem\restr{\dom(\renvone)} \in \supp (\rsone)$. From \Cref{lemma:parsemproj}, we deduce $\parsem\deone{\overline \nat}(\overline \mem\restr{\dom(\renvone)})\neq\parsem\detwo{\overline \nat}(\overline \mem\restr{\dom(\renvone)})$, which contradicts the premise $\sr \rsone \tyf {\espl \deone \detwo}\renvone$. For the opposite direction, we assume that $\sr \rstwo \tyf {\espl \deone \detwo}\renvtwo$. To this aim, we fix $\nat$, and we let $\overline \mem$ be an element of $\supp(\rsone_\nat)$. The goal is to show that we have $\parsem\deone{\nat}(\overline \mem)=\parsem\detwo{\nat}(\overline \mem)$. From \Cref{rem:suppext}, we deduce that there is $\mem\in \supp(\rstwo_\nat)$ such that $\overline \mem = \mem\restr{\dom(\renvone)}$. From the assumption, we deduce that $\parsem\deone{\nat}(\mem)=\parsem\detwo{\nat}(\mem)$, so with \Cref{lemma:parsemproj}, we conclude that $\parsem\deone{\nat}(\mem\restr{\dom(\renvone)})=\parsem\detwo{\nat}(\mem\restr{\dom(\renvone)})$, which is equivalent to the claim because we assumed that $\overline \mem = \mem\restr{\dom(\renvone)}$.
    \item If these are $\tyf \ftrue\renvone $ or $\tyf \ffalse\renvone$, the conclusion is trivial.
        { 
        \item If the two formulas are additive conjunctions, then we start from the left to right implication. Assume $\srenv\renvone \stone {\tyf{\tyf {\fone_1} {\renvone_1}\land\tyf {\ftwo_1} {\renvtwo_1}}{\renvone}}$, so we deduce that $\srenv{\renvone_1} {\stone_1} {\tyf {\fone_1} {\renvone_1}}$ and $\srenv{\renvtwo_1} {\stone_2} {\tyf {\ftwo_1} {\renvtwo_1}}$ with  $\stone_1=\stone_{\renvone \to \renvone_1}=(\sttwo_{\renvtwo\to\renvone})_{\renvone \to \renvone_1}=\sttwo_{\renvtwo \to \renvone_1}$, and $\stone_2=\stone_{\renvone \to \renvtwo_1}=(\sttwo_{\renvtwo\to\renvone})_{\renvone \to \renvtwo_1}=\sttwo_{\renvtwo \to \renvtwo_1}$ by definition of $\ext$ and \Cref{lemma:exttrans1}. From $\tyf {\fone_1} {\renvone_1}\ext \tyf {\fone_2} {\renvone_2}$ and $\tyf {\ftwo_1} {\renvtwo_1}\ext \tyf {\ftwo_2} {\renvtwo_2}$, we can conclude that $\renvtwo_1\ext \renvtwo_2$ and that $\renvone_1 \ext\renvone_2$ using \Cref{rem:extftoexts}. For this reason, we can apply the IHs to the rule's premises, to conclude $\srenv{\renvone_2} {\sttwo_{\renvtwo\to\renvone_2}} {\tyf {\fone_2} {\renvone_2}}$ and $\srenv{\renvtwo_2} {\sttwo_{\renvtwo\to\renvtwo_2}} {\tyf {\ftwo_2} {\renvtwo_2}}$, which concludes the proof.
          For the opposite direction, we assume that $\srenv\renvtwo {\sttwo}{\tyf{\tyf {\fone_2} {\renvone_2}\land\tyf {\ftwo_2} {\renvtwo_2}}{\renvtwo}}$, By definition of the semantic relation and by definition of $\ext$, we deduce that: $\srenv{\renvone_2} {\sttwo_{\renvtwo\to\renvone_2}} {\tyf {\fone_2} {\renvone_2}}$ and $\srenv{\renvtwo_2} {\sttwo_{\renvtwo\to\renvtwo_2}} {\tyf {\ftwo_2} {\renvtwo_2}}$. From ${\tyf{\tyf {\fone_1} {\renvone_1}\land\tyf {\ftwo_1} {\renvtwo_1}}{\renvone}}\ext{\tyf{\tyf {\fone_2} {\renvone_2}\land\tyf {\ftwo_2} {\renvtwo_2}}{\renvtwo}}$, we deduce that $\tyf {\fone_1} {\renvone_1}\ext \tyf {\fone_2} {\renvone_2}$ and $\tyf {\ftwo_1} {\renvtwo_1}\ext \tyf {\ftwo_2} {\renvtwo_2}$, so we can conclude that $\renvtwo_1\ext \renvtwo_2$ and that $\renvone_1 \ext\renvone_2$, \Cref{rem:extftoexts}. These observations allow us to apply the IHs on the pair of states $(\sttwo_{\renvtwo\to\renvone_2}, \sttwo_{\renvtwo\to\renvone_1})$, and $(\sttwo_{\renvtwo\to\renvtwo_2}, \sttwo_{\renvtwo\to\renvtwo_1})$, to show that $\srenv{\renvone_1} {\sttwo_{\renvtwo\to\renvone_1}} {\tyf {\fone_1} {\renvone_1}}$ and $\srenv{\renvtwo_1} {\sttwo_{\renvtwo\to\renvtwo_1}} {\tyf {\ftwo_1} {\renvtwo_1}}$. This also shows that $\srenv{\renvone_1} {\stone_{\renvone\to \renvone_1}} {\tyf {\fone_1} {\renvone_1}}$ and $\srenv{\renvtwo_1} {\stone_{\renvone\to \renvtwo_1}} {\tyf {\ftwo_1} {\renvtwo_1}}$, because $\sttwo_{\renvtwo\to\renvone_1}=(\sttwo_{\renvtwo\to\renvone})_{\renvone\to\renvone_1}=\stone_{\renvone\to\renvone_1}$ and $\sttwo_{\renvtwo\to\renvtwo_1}=(\sttwo_{\renvtwo\to\renvone})_{\renvone\to\renvtwo_1}=\stone_{\renvone\to\renvtwo_1}$, for \Cref{lemma:exttrans1}. This concludes the proof.
}
  \item If the two formulas are two separating conjunctions, then we start from the left to right implication. Assume $\srenv\renvone \stone {\tyf{\tyf {\fone_1} {\renvthree}\sep\tyf {\ftwo_1} \renvfour}\renvone}$, so we deduce that $\srenv\renvthree {\stone_1} {\tyf {\fone_1} {\renvthree}}$ and $\srenv\renvfour {\stone_2} {\tyf {\ftwo_1} {\renvfour}}$ with $\stone_1\comp\stone_2 \ind\ext \stone\ext\sttwo$. So, we conclude $\srenv\renvtwo \sttwo {\tyf{\tyf {\fone} {\renvthree}\sep \tyf {\ftwo} {\renvfour}} \renvtwo}$. For the opposite direction, we assume that there are $\srenv\renvthree {\sttwo_1} {\tyf {\fone_1} {\renvthree}}$ and $\srenv\renvfour {\sttwo_2} {\tyf {\ftwo_1} {\renvfour}}$ with
    $\sttwo_1\comp\sttwo_2 \ind\ext \sttwo$.
    From the last observation and the definitions of $\ind$ and $\ext$, we can deduce that
    $\sttwo_{\renvtwo\to \join \renvthree\renvfour} \ind \sttwo_1\comp\sttwo_2$. We also know that ${\tyf{\tyf {\fone_1} {\renvthree}\sep\tyf {\ftwo_1} \renvfour}\renvone}$ is a formula, which means that $\join \renvthree \renvfour \ext \renvone$. So, by applying Lemma \ref{lemma:exttrans1} we can conclude that $\sttwo_1\comp\sttwo_2\ind\sttwo_{\renvtwo\to \join \renvthree\renvfour}=(\sttwo_{\renvtwo\to\renvone})_{\renvone\to \join \renvthree\renvfour}=(\stone)_{\renvone\to \join \renvthree\renvfour}$ which concludes the proof.
  \end{itemize}
\end{proof}

The approximate fragment of \CPSL\ also enjoys good properties with respect to 
$\ind$.
In particular, if two states $\stone$ and $\sttwo$ are cryptographically indistinguishable,
and one of these satisfies a formula $\tyf \fone \renvone \in \apfset$,
then also the other satisfies $\tyf \fone \renvone$. This is stated in \Cref{lemma:indcbi}.

\begin{lemma}
  \label{lemma:indcbi}
  For every environment $\renvone$, and every \emph{approximate formula} $\tyf \fone \renvone\in \apfset$ and pair of states $\stone, \sttwo\in \sem \renvone {}$ such that $\stone \ind \sttwo$, it holds that $\srenv \renvone \stone {\tyf \fone \renvone} \Rightarrow \srenv \renvone \sttwo {\tyf \fone \renvone}$.
    \hfill\qed
\end{lemma}

\begin{proof}[Proof of \Cref{lemma:indcbi}]
  By induction on the proof that $\tyf \fone \renvone$ is a formula.
  \begin{itemize}
  \item If $\tyf \fone \renvone$ is an atomic formula, we go by cases on the atomic proposition. Now assume that the formula is $\tyf {\ud\reone}\renvone$. We assume $\srenv \renvone \stone {\tyf{\ud\reone} \renvone}$, which means that $\sem {\ternjudg \renvone\reone\tyone}{}(\stone) {} \ind \unif{\tyone}$. From $\stone \ind \sttwo$ and \Cref{lemma:exprind}, we deduce that $\sem {\ternjudg \renvone\reone\tyone}{}(\sttwo) {} \ind\sem {\ternjudg \renvone\reone\tyone}{}(\stone) {}\ind \unif{\tyone}$, which concludes the proof.
    \item If $\tyf \fone \renvone$ is $\tyf \ftrue\renvone $ or $\tyf \ffalse\renvone$, the conclusion is trivial.
        { 
    \item  If the formula is an additive conjunction, we assume $\srenv\renvone \stone {\tyf{\tyf {\fone_1} {\renvthree}\land\tyf {\ftwo_1} \renvfour}\renvone}$, so we deduce that $\srenv\renvthree {\stone_1} {\tyf {\fone_1} {\renvthree}}$ and $\srenv\renvfour {\stone_2} {\tyf {\ftwo_1} {\renvfour}}$ with $\stone_1\ext \stone$ and $\stone_2\ext \stone$. From this result, we deduce that $\stone_{\renvone\to\renvthree}= \stone_1$ and $\stone_{\renvone\to\renvfour}= \stone_2$ By applying \Cref{lemma:exprind}, we conclude $\sttwo_{\renvone\to\renvthree}\ind \stone_1$ and $\sttwo_{\renvone\to\renvfour}\ind \stone_2$. By applying the IHs to $\stone_1$ and $\sttwo_{\renvone\to\renvthree}$, and to $\stone_2$ and $\sttwo_{\renvone\to\renvfour}$, we are able to show that $\srenv\renvthree {\sttwo_{\renvone\to\renvthree}} {\tyf {\fone_1} {\renvthree}}$ and $\srenv\renvfour {\sttwo_{\renvone\to\renvthree}} {\tyf {\ftwo_1} {\renvfour}}$, which concludes the proof.
    }

  \item If the formula is a separating conjunction, we assume $\srenv\renvone \stone {\tyf{\tyf {\fone_1} {\renvthree}\sep\tyf {\ftwo_1} \renvfour}\renvone}$, so we deduce that $\srenv\renvthree {\stone_1} {\tyf {\fone_1} {\renvthree}}$ and $\srenv\renvfour {\stone_2} {\tyf {\ftwo_1} {\renvfour}}$ with $\stone_1\comp\stone_2 \ind\ext \stone\ind\sttwo$. By applying \Cref{lemma:distributivity}, we obtain $\stone_1\comp\stone_2 \ind\ext\sttwo$, which concludes the proof.
  \end{itemize}
\end{proof}

These two results, taken together, can be combined
in a third result which is analogous to the restriction
property of~\cite{PSL}:

\begin{cor}
  \label{cor:restriction}
  For every pair of environments $\renvone, \renvtwo$, every pair of approximate formulas $\tyf \fone \renvone$ and $\tyf\ftwo \renvtwo$ such that $\tyf \fone \renvone \ext \tyf\ftwo \renvtwo{}$ and every pair of states $\stone\in \sem \renvone {}, \sttwo \in \sem \renvtwo{}$ such that $\stone \ind\ext \sttwo$, it holds that $\srenv \renvone \stone {\tyf \fone \renvone} \Leftrightarrow \srenv \renvtwo \sttwo {\tyf \ftwo \renvtwo}$.
\end{cor}

We also introduce a sound although incomplete
set of axioms about atomic formulas. Despite their simplicity,
they are expressive enough to reason about cryptographic primitives.

\renewcommand{\tyf}[2]{#1}
\begin{lemma}
  \label{lemma:apax}
  The following axiom schemas are valid:
  \begin{align*}
    &\sr{} \tyf{\indp \reone \reone}\renvone
      \tag{S0}\\
   \tyf{\indp\reone \retwo}\renvone &\models \tyf{ \indp \retwo \reone}\renvone
      \tag{S1}\\
    \tyf{\indp\reone \retwo \land \indp \retwo \rethree}\renvone  &\models \tyf{\indp \reone \rethree}\renvone
      \tag{S2}\\
    &\models \tyf{\eq \reone \reone}\renvone
      \tag{T0}\\
   \tyf{\eq\reone \retwo}\renvone &\models \tyf{\eq \retwo \reone}\renvone
      \tag{T1}\\
    \tyf{\tyf{\eq\reone \retwo} \renvone \land\tyf{ \eq \retwo \rethree}\renvone} \renvone &\models \tyf{\eq \reone \rethree}\renvone
      \tag{T2}\\
    \tyf{\eq\reone \retwo}\renvone &\models\tyf {\indp \reone \retwo}\renvone
      \tag{W1}\\
    \tyf{\espl\deone \detwo}\renvone  &\models \tyf{\eq \deone \detwo}\renvone
      \tag{W2}\\
    \tyf{\tyf{\indp\reone \retwo}\renvone \land\tyf{ \ud \reone}\renvone}\renvone &\models \tyf{\ud \retwo}\renvone
      \tag{U1}
  \end{align*}
  \hfill\qed
\end{lemma}
\renewcommand{\tyf}[2]{{(#1)^{#2}}}

\begin{proof}[Proof of Lemma \ref{lemma:apax}]
  \begin{proofcases}
  \proofcase[S0] Follows from reflexivity of $\ind$.
  \proofcase[S1] Follows from symmetry of $\ind$.
  \proofcase[S2] Follows from transitivity of $\ind$.
  \proofcase[T0-2] The axioms on $\eq\cdot\cdot$ are trivial consequences of the fact that identity is an equation relation. Their proof would be similar to the correspondent axiom for $\indp\cdot\cdot$.
  \proofcase[W1] Assume $\sem {\ternjudg \renvone \reone\tyone}{} (\rsone)=\sem {\ternjudg \renvone \retwo\tyone} {}(\rsone)$, then:
    \begin{equation*}
      \Pr_{t \leftarrow \sem \reone\nat(\rsone)}[\distone(1^n, t)=1]=
      \Pr_{t \leftarrow \sem \retwo\nat(\rsone)}[\distone(1^n, t)=1],
    \end{equation*}
    so the advantage of any adversary is exactly $0$.
    \proofcase[W2] Assume $\forall \mem \in \supp(\rsone). \parsem {\ternjudg \renvone \reone\tyone}{} (\mem)=\parsem {\ternjudg \renvone \retwo\tyone} {}(\mem)$, then the goal is to show that $\sem {\ternjudg \renvone \deone\tyone}{} (\rsone)=\sem {\ternjudg \renvone \detwo\tyone} {}(\rsone)$.
    We fix $\nat\in \NN$. The proof proceeds by expanding $\bind{\rsone_\nat}{\mem\mapsto\detsem\deone{}(\mem)}(t)$ as follows:
    \begin{align*}
      \smashoperator[r]{\sum_{\mem \in\supp(\rsone_\nat)}}\;\;\rsone_\nat(\mem)\cdot \detsem \deone\nat(\mem)(t) & = \smashoperator[r]{\sum_{\mem \in\supp(\rsone_\nat)}}\;\;\rsone_\nat(\mem)\cdot \unit {\parsem \deone\nat(\mem)}(t)\\
      & = \smashoperator[r]{\sum_{\mem \in\supp(\rsone_\nat)}}\;\;\rsone_\nat(\mem)\cdot \unit {\parsem \detwo\nat(\mem)}(t)=
      \smashoperator[lr]{\sum_{\mem \in\supp(\rsone_\nat)}}\;\;\rsone_\nat(\mem)\cdot \detsem \detwo\nat(\mem)(t).
    \end{align*}
  \proofcase[U1] Let $\tyf \rsone \renvone$ be any store such that $\sr {\tyf \rsone \renvone} {{\indp x y}\land \ud y}$. The claim is $\sr {\tyf \rsone \renvone} \ud x$. It holds that $\sem {\ternjudg\renvone\reone\tyone}{} (\rsone)\ind \sem {\ternjudg\renvone\retwo\tyone} {} (\rsone){}\ind\unif\tyone$. The claim follows from transitivity of $\ind$.
  \end{proofcases}  
\end{proof}

\begin{figure*}[t]
  \centering
  \revision{  \[
      \infer[\text{AP}]{\prenv \renvone {\tyf \apone \renvone} {\tyf \apone \renvone}}{}\quad
      \infer[\top_i]{\prenv \renvone {} {\tyf\top\renvone}}{}\quad
      \infer[\bot_e]{\prenv \renvone {\tyf\bot\renvone} {\tyf\fone\renvone}}{}
    \]
    \vspace{-2mm}
    \[
      \infer[\land_i]{\prenv \renvone {\tyf \fthree \renvone} {\tyf{\tyf \fone \renvtwo \land \tyf \ftwo \renvthree}}\renvone}
      {
        \prenv \renvone {\tyf \fthree \renvone} {\tyf \fone \renvone} &
        \prenv \renvone {\tyf \fthree \renvone} {\tyf \ftwo \renvone} &
        \tyf \fone \renvtwo \ext \tyf \fone \renvone &
        \tyf \fthree \renvthree \ext \tyf \fthree \renvone
      }
      \quad
      \infer[\land_e]
      {
        \prenv \renvone {\tyf \ftwo \renvone} {\tyf {\fone_i} \renvone}
      }
      {
        \prenv \renvone {\tyf \ftwo \renvone} \tyf {{\tyf {\fone_1} \renvtwo}\land{\tyf {\fone_2} \renvthree}} \renvone
      }
    \]
    \vspace{-2mm}
    \[
      \infer[\sep_i]{\prenv \renvone {\tyf{{\tyf \fone \renvtwo} \sep {\tyf \ftwo \renvthree}}\renvone}
        {\tyf{{\tyf \fthree \renvtwo} \sep {\tyf \ffour \renvthree}}\renvone}}
      {
        \prenv \renvtwo {\tyf \fone \renvtwo} {\tyf \fthree \renvtwo} &
        \prenv \renvthree {\tyf \ftwo \renvthree} {\tyf \ffour \renvthree}
      }
      \quad
      \infer[\sep_{c}]{\prenv \renvone {\tyf{{\tyf \fone \renvtwo} \sep {\tyf \ftwo \renvthree}} \renvone}
        {\tyf{  {\tyf \ftwo \renvthree} \sep {\tyf \fone \renvtwo}} \renvone}}
      {
      }
    \]
    \vspace{-2mm}
    \[
      \infer[\sep_{a_1}]{\prenv \renvone {\tyf{({\tyf \fone \renvtwo} \sep {\tyf \ftwo \renvthree}) \sep \tyf \fthree \renvfour} \renvone}
        {\tyf {\tyf \fone \renvtwo \sep (\tyf \ftwo \renvthree \sep \tyf \fthree \renvfour)} \renvone}
      }
      {
      }
    \]
    \[      \infer[\sep_{a_2}]{\prenv \renvone 
        {\tyf {\tyf \fone \renvtwo \sep (\tyf \ftwo \renvthree \sep \tyf \fthree \renvfour)} \renvone}
        {\tyf{({\tyf \fone \renvtwo} \sep {\tyf \ftwo \renvthree}) \sep \tyf \fthree \renvfour} \renvone}
      }
      {
      }
    \]
  }
  \caption{A Hilbert-style proof system for \CPSL.}
  \label{fig:proofcpsl}
\end{figure*}

The two conjunctions of \CPSL are commutative
and enjoy a form of associativity, in particular, the
separating conjunction is associative when the environments
associated to a formula are minimal.

\begin{lemma}
  \label{rem:conjasscomm}
  For every $\odot \in \{\land ,\sep\}$,
  if $\tyf {\tyf {\tyf \fone \renvone \odot \tyf \ftwo \renvtwo} {\join \renvone \renvtwo}\odot \tyf \fthree \renvthree}{\join {\join \renvone \renvtwo}\renvthree}$ is a formula, then 
  $\tyf { \tyf \ftwo \renvtwo \odot \tyf \fone \renvone} {\join \renvone \renvtwo}$
  and  $\tyf {{\tyf \fone \renvone} \odot \tyf  {{\tyf \ftwo \renvtwo} \odot {\tyf \fthree \renvthree}}{\join \renvtwo \renvthree}}{\join {\join \renvone \renvtwo}\renvthree}$  are also
  formulas, and we have:
  \begin{align*}
    {\tyf {\tyf \fone \renvone \odot \tyf \ftwo \renvtwo} {\join \renvone \renvtwo}}
    &\models {\tyf { \tyf \ftwo \renvtwo \odot \tyf \fone \renvone} {\join \renvone \renvtwo}}&&\text{(Commutativity)}\\
    \tyf {\tyf {\tyf \fone \renvone \odot \tyf \ftwo \renvtwo} {\join \renvone \renvtwo}\odot \tyf \fthree \renvthree}{\join {\join \renvone \renvtwo}\renvthree} &\models
    ({\tyf \fone \renvone} \odot ({\tyf \ftwo \renvtwo} \odot {\tyf \fthree \renvthree})^{\join \renvtwo \renvthree})^{\join {\join \renvone \renvtwo}\renvthree} &&\text{(Associativity)}
  \end{align*}
\end{lemma}

\begin{proof}
  The commutative property is trivial for both the conjunctions.
  For the associative property we start by showing that
  $\tyf {{\tyf \fone \renvone} \odot \tyf  {{\tyf \ftwo \renvtwo} \odot {\tyf \fthree \renvthree}}{\join \renvtwo \renvthree}}{\join {\join \renvone \renvtwo}\renvthree}$
  is a formula. On the one hand, if $\odot =\land$,
  we know that $\tyf \fone \renvone$,
  $\tyf \ftwo \renvtwo$ and $\tyf \fthree \renvthree$ are all formulas.
  From this observation, we can conclude that
  $\tyf  {{\tyf \ftwo \renvtwo} \land {\tyf \fthree \renvthree}}{\join \renvtwo \renvthree}$
  is a formula because $\renvtwo, \renvthree \ext \join \renvtwo \renvthree$
  and with a similar argument we show that the formula on the right is well-formed.
  On the other hand, if $\odot =\sep$, we start by observing that
  it must be the case where the domains of $\renvone, \renvtwo$ and
  $\renvthree$ are disjoint, and as above
  we know that $\tyf \fone \renvone$,
  $\tyf \ftwo \renvtwo$ and $\tyf \fthree \renvthree$ are formulas.
  This means that 
  $\tyf  {{\tyf \ftwo \renvtwo} \sep {\tyf \fthree \renvthree}}{\join \renvtwo \renvthree}$
  is a formula, and with an argument, we have shown that the formula on the right is well-formed
  also when the $\odot=\sep$.
  Now we show that
  \begin{equation*}
    \tyf {\tyf {\tyf \fone \renvone \odot \tyf \ftwo \renvtwo} {\join \renvone \renvtwo}\odot \tyf \fthree \renvthree}{\join {\join \renvone \renvtwo}\renvthree} \models
    \tyf {{\tyf \fone \renvone} \odot \tyf  {{\tyf \ftwo \renvtwo} \odot {\tyf \fthree \renvthree}}{\join \renvtwo \renvthree}}{\join {\join \renvone \renvtwo}\renvthree}
  \end{equation*}
  If $\odot =\land$, we assume that there is $\stone$ which satisfies the premise.
  This means that:
  \begin{proofcases}
    \proofcase[H1] $\sr {\stone_{\join {\join \renvone \renvtwo}\renvthree\to \join \renvone \renvtwo}}{\tyf {\tyf \fone \renvone \land \tyf \ftwo \renvtwo} {\join \renvone \renvtwo}}$, which in turn means that:
    \begin{proofcases}
      \proofcase[H1A] $\sr {\stone_{\join {\join \renvone \renvtwo}\renvthree\to  \renvone}}{\tyf { \fone} { \renvone}}$.
      \proofcase[H1B]  $\sr {\stone_{\join {\join \renvone \renvtwo}\renvthree\to \renvtwo}}{\tyf {\ftwo } { \renvtwo}}$.
    \end{proofcases}
    \proofcase[H2]  $\sr {\stone_{\join {\join \renvone \renvtwo}\renvthree\to \renvthree}}{\tyf {\fthree } { \renvthree}}$.
  \end{proofcases}
  The claim asks us to show that:
  \begin{proofcases}
    \proofcase[C1] $\sr {\stone_{\join {\join \renvone \renvtwo}\renvthree\to  \renvone}}{\tyf { \fone} { \renvone}}$.
    \proofcase[C2] $\sr {\stone_{\join {\join \renvone \renvtwo}\renvthree\to \join \renvtwo \renvthree}} {\tyf  {{\tyf \ftwo \renvtwo} \land {\tyf \fthree \renvthree}}{\join \renvtwo \renvthree}}$, which can be reduced to showing:
    \begin{proofcases}
      \proofcase[C2A]  $\sr {\stone_{\join {\join \renvone \renvtwo}\renvthree\to \renvtwo}}{\tyf {\ftwo } { \renvtwo}}$.
      \proofcase[C2B] $\sr {\stone_{\join {\join \renvone \renvtwo}\renvthree\to  \renvthree}}{\tyf { \fthree} { \renvthree}}$.
    \end{proofcases}
    by definition of the semantics of $\land$
  \end{proofcases}
  The conclusion is trivial, because the premises coincide with the claims.
  On the other hand, $\odot =\sep$, the proof is more articulated.
  We start by assuming that there is a state $\stone$ which satisfies the premises.
  This means that there are two states $\sttwo \in \sem {\join \renvone\renvtwo}{}$, $\stthree\in \sem {\renvthree}{}$ such that:
  \begin{proofcases}
    \proofcase[H1] $\sr \sttwo {\tyf {\tyf \fone \renvone \sep \tyf \ftwo \renvtwo} {\join \renvone \renvtwo}}$, which in turn means that there are two states $\sttwo_1\in \sem \renvone {}$ and  $\sttwo_2\in \sem \renvtwo{}$ such that:
    \begin{proofcases}
      \proofcase[H1A] $\sr {\sttwo_1}{\tyf { \fone} { \renvone}}$.
      \proofcase[H1B]  $\sr {\sttwo_2}{\tyf {\ftwo } { \renvtwo}}$.
      \proofcase[H1C]  $\sttwo_1 \comp \sttwo_2 \ind \sttwo$.
    \end{proofcases}
    \proofcase[H2]  $\sr {\stthree}{\tyf {\fthree } { \renvthree}}$.
    \proofcase[H3]  $\sttwo \comp \stthree \ind \stone$.
  \end{proofcases}
  The claim asks us to show that there are two states $\stfour_1\in \sem \renvone {}$
  and  $\stfour_2 \in \sem {\join \renvtwo \renvthree} {}$ such that :
  \begin{proofcases}
    \proofcase[C1] $\sr {\stfour_1}{\tyf { \fone} { \renvone}}$.
    \proofcase[C2] $\sr {\stfour_2} {\tyf  {{\tyf \ftwo \renvtwo} \sep {\tyf \fthree \renvthree}}{\join \renvtwo \renvthree}}$.
    \proofcase[C3] ${\stfour_1}\comp \stfour_2 \ind \stone$.
  \end{proofcases}
  We choose $\stfour_1 = \sttwo_1$ and $\stfour_2 = \sttwo_2  \comp \stthree$. With this
  choice (C1) is trivial, (C2) requires showing that there are $\stfour_{2a} \in \sem \renvtwo {}$
  and $\stfour_{2b} \in \sem \renvthree {}$ such that:
  \begin{proofcases}
    \proofcase[C2A]  $\sr {\stfour_{2a}}{\tyf {\ftwo } { \renvtwo}}$.
    \proofcase[C2B]  $\sr {\stfour_{2b}}{\tyf {\fthree } { \renvthree}}$.
    \proofcase[C2C]  $\stfour_{2a} \comp \stfour_{2b} \ind \stfour_2$.
  \end{proofcases}
  In particular, we can choose $\stfour_{2a}= \sttwo_2$ and $\stfour_{2b}= \stthree$.
  With this choice (C2A) corresponds with (H1B),
  (C2B) corresponds with (H2).
  And (C2C) is trivial, because we can show:
  \[
    \stfour_{2a} \comp \stfour_{2b} = \stfour_2.
  \]
  In order to show (C3), we are required to show:
  \[
    {\stfour_1}\comp \stfour_2 \ind \stone,
  \]
  which means showing:
  \[
    {\sttwo_1}\comp (\sttwo_2\comp \stthree) \ind \stone
  \]
  By associativity of $\comp$, we can rewrite the claim as
  \[
    ({\sttwo_1}\comp \sttwo_2)\comp \stthree \ind \stone.
  \]
  For compatibility of $\ind$ and $\comp$ (\Cref{lemma:indcompcomp}),
  we can use (H1C) to simplify the claim even more to: 
  \[
    \sttwo\comp \stthree \ind \stone.
  \]
  and conclude with (H3).
\end{proof}
This result justifies the notation $\bigodot_{i =0}^k \tyf{\fone_i}{\renvone_i}$ as a shorthand for
\[
  \tyf{\tyf{\tyf{\fone_0}{\renvone_0}\odot \tyf{\fone_1}{\renvone_1}}{\join {\renvone_0}{\renvone_1}} \ldots \odot \tyf{\fone_k}{\renvone_k}}{\join{\join{\join {\renvone_0}{\renvone_1}} \ldots}{\renvone_k}}
\]
with $\odot\in \{\sep, \land\}$.
Thanks to this result,
we can also give a Hilbert-style proof system for \CPSL,
based on the standard one for \BI from~\cite{Pym04};
see \Cref{fig:proofcpsl}.

\begin{lemma}[Soundness of $\prenv \renvone \cdot \cdot$]
  \label{lemma:cpslsound}
  For every pair of formulas $\tyf \fone \renvone, \tyf \ftwo \renvone \in \fset \renvone$,
  we have that:
  \[
    \left(\prenv \renvone {\tyf \fone \renvone} {\tyf \ftwo \renvone} \right)\Rightarrow
    \left(\srenv \renvone {\tyf \fone \renvone} {\tyf \ftwo \renvone}\right).
  \]
\end{lemma}

\begin{proof}
  By induction on the proof relation. In particular, $\land_i$ is a consequence of
  \Cref{lemma:prekm}. Commutativity and associativity of $\sep$ follow from \Cref{rem:conjasscomm}.
\end{proof}

\subsection{\CPSL and Other Separation Logics}

\subsubsection{\CPSL vs. \BI}
\label{subsubsec:cpslvsbi}
From a syntactic point of view,
\CPSL is the conjunctive fragment of \BI where
every formula $\fone$ is tagged with an environment $\renvone$
such that $\binjudg \renvone \fone$. In contrast with
what happens with \BI, by employing this syntax, we
do not need to use the notion of \emph{footprint} for programs and
formulas~\cite{PSL,Reynolds08} to denote the minimal set of variables that
these objects depends on, because this set is now explicit
in the formula. This choice comes with almost no additional cost,
as environments are anyway needed to enforce the polytime constraint on our language.

From a semantic point of view, we model \CPSL on
stores, i.e. pairs of an \emph{effectively samplable} distribution ensemble
and an environment. These objects,
in turn, form a \emph{partial Kripke resource monoid} with the pre-order
$\ind\ext$ and the tensor product as composition.
\begin{proposition}
  \label{prop:krmsec1}
  The structure
  $(\stset,\tensprod,\emst{},\ind\ext)$ is a partial 
  Kripke resource monoid.
  \hfill\qed
\end{proposition}

\begin{proof}
  The associativity of $\comp$ is a consequence of the associativity
  of environments' union and of multiplication.
  Now we show that $\emst{}$ is the identity for $\comp$:
  \begin{align*}
    (\{\{\varepsilon^1\}\}_{n \in \NN}^\varepsilon  \comp{\tyf \rsone \renvone})_\nat(\mem) &= \tyf{(\{\varepsilon^1\}\tensprod\rsone_\nat)(\mem)}{\join \varepsilon \renvone}\\
    &=\tyf{\{\varepsilon^1\}( \mem\restr \emptyset) \cdot\rsone_\nat(\mem\restr {\dom(\renvone)})}{ \renvone}\\
    &=\tyf{1\cdot\rsone_\nat(\mem\restr {\dom(\renvone)})}{ \renvone}\\
    &=\tyf{\rsone_\nat(\mem)}{ \renvone}
  \end{align*}
  \normalsize
  The proof of the right identity of $\emst{}$ is analogous. Now, we show that $\ind\ext$ is a pre-order. Reflexivity is trivial since both the relations are reflexive. Transitivity can be shown as follows:
  \begin{align*}
    \stone \ind \sttwo \ext \stthree \ind \stfour \ext \stfive &\models
    \big(\exists\stthree'.\stone \ind \sttwo \ind \stthree' \ext \stfour \ext \stfive \big)\\ &\models
    \big(\exists\stthree'.\stone \ind \stthree' \ext \stfive \big) \vdash \stone \ind \ext \stfive
  \end{align*}
  Where the second step is a consequence of Lemma \ref{lemma:distributivity}, and the third is a consequence of the transitivity of the two relations $\ind, \ext$. It remains us to show that $\ind\ext$ is compatible with $\comp$.
  In particular, this requires to show that
  \[
    \forall \stone,\sttwo, \stthree. \stone\ind\ext\sttwo \Rightarrow (\stone\comp \stthree)\defined\Rightarrow(\sttwo\comp \stthree)\defined\Rightarrow (\stone\comp \stthree)\ind\ext(\sttwo\comp \stthree)
  \]
  and
  \[
    \forall \stone,\sttwo, \stthree. \stone\ind\ext\sttwo \Rightarrow (\stthree\comp \stone)\defined\Rightarrow(\stthree\comp \sttwo)\defined\Rightarrow (\stthree\comp \stone)\ind\ext(\stthree\comp \sttwo)
  \]
  We show only the first claim, the proof of the second is analogous. Assume that $\stone\ind\ext\sttwo$, $(\stone\comp \stthree)$ is defined, and that $(\sttwo\comp \stthree)$ is defined. From $\stone\ind\ext\sttwo$ we deduce that there is a state $\stfour$ such that $\stone\ind\stfour\ext\sttwo$.  Let $\stthree\in\stset$ be a store. From Lemma \ref{lemma:indcompcomp}, we obtain $\stone\comp \stthree\ind\stfour\comp\stthree$ and from Lemma \ref{lemma:extcompcomp}, we obtain $\stfour\comp \stthree\ext\sttwo\comp \stthree$. This concludes the proof of the compatibility of $\ind\ext$ and $\comp$.
\end{proof}

\newcommand{\unenv}[1]{\left[ #1\right]}

In order to ease the comparison between these two logics, in
the following we call $\fset_{\BI}$ the set of conjunctive \BI's formulas
and, given a formula $\fone\in \fset_{\BI}$, we write $\fone =\unenv{\tyf \fone \renvone}$
if $\fone$ can be obtained by removing all the environment notations from $\tyf \fone \renvone$. 
Thanks to \cref{prop:krmsec1}, we can model $\fset_{\BI}$ on \emph{partial Kripke resource monoids}.
This will allow us, to stipulate a semantic correspondence between these two logics.
The semantics of $\fset_{\BI}$ is in \Cref{fig:bisem}. Notice that we use $\Vdash$ instead
of $\models$, as we do for \CPSL's formulas.

\begin{figure*}[t]
  \begin{center}
    \begin{align*}
      \sem{{\ud{\reone}}}{}\defsym\{{\tyf \rsone \renvone \in \stset}\mid&\;\sem{\ternjudg \renvone\reone\tyone}{}(\rsone)\ind\unif{\tyone}\}\\
      \sem{{\indp{\reone}{\retwo}}}{}\defsym\{{\tyf \rsone \renvone \in \stset}\mid&\;\sem{\ternjudg \renvone\reone\tyone}{}(\rsone)\ind\sem{\ternjudg \renvone\retwo\tyone}{}(\rsone)\}\\
      \sem{{\eq{\reone}{\retwo}}}{}\defsym\{{\tyf \rsone \renvone \in \stset}\mid&\;\sem{\ternjudg \renvone\reone\tyone}{}(\rsone)=\sem{\ternjudg \renvone\retwo\tyone}{}(\rsone)\}\\
      \sem{{\espl{\deone}{\detwo}}}{}\defsym\{{\tyf \rsone \renvone\in \stset}\mid&\;\forall \nat \in \NN.\forall \mem \in \supp (\rsone_\nat). \parsem{\ternjudg \renvone\deone\tyone}{\nat}(\mem)=\parsem{\ternjudg \renvone\detwo\tyone}{\nat}(\mem)\}
    \end{align*}
    
    \begin{tabular}{ll}
      ${\stone}\Vdash{ \apone}$&iff $\stone\in\sem{\apone}{}$ \\
      ${\stone}\Vdash{ \ftrue}$&always\\
      ${\stone}\Vdash{ \ffalse}$&never \\
      {${\stone}\Vdash\fone \land \ftwo$} &{iff ${\stone}\Vdash\fone$ and ${\stone}\Vdash\ftwo$}\\
      ${\stone}\Vdash{ \fone \sep \ftwo}$&iff there are $\sttwo, \stthree$ such that $\sttwo\comp \stthree\defined$, $\sttwo \comp \stthree \ind\ext \stone$, ${\sttwo}\Vdash{ {\fone}}$ and 
                                     ${\stthree}\Vdash{ {\ftwo}}$\\
    \end{tabular}
  \end{center}  
  \caption{Semantics of \BI}
  \label{fig:bisem}
\end{figure*}

The semantics in \Cref{fig:bisem} is very close to \CPSL's one, with the exception that
we do not impose any restriction on the domain of the stores. For this reason, the semantics
of $\fset_{\BI}$ formulas is less strict that the one of \CPSL's formula. 
This is why from the validity of a \CPSL's formula
$\tyf \fone \renvone$ in a store $\stone$,
we can deduce the validity of the structurally identical
formula (but without environment annotations) in the same store,
using the \BI interpretation.

\begin{proposition}
  \label{prop:cpsltobi}
  For every environment $\renvone$, \CPSL formula
  $\tyf \fone \renvone \in \tyf \fset \renvone$, and
  store $\stone \in \sem \renvone {}$, we have:
  \[
    \srenv \renvone \stone \tyf \fone \renvone \Rightarrow  \stone \Vdash \unenv{\tyf \fone \renvone},
  \]
  where $\unenv{\tyf \fone \renvone}$ is the \BI formula that is obtained by removing
  the environment tags from $\tyf \fone \renvone$.
  \hfill\qed
\end{proposition}
\begin{proof}
  Direct consequence of \Cref{lemma:cpsltobiaux} below.
\end{proof}

\begin{lemma}
  \label{lemma:cpsltobiaux}
  \begin{equation*}
    \forall \fone \in \fset_{\BI}.\forall \renvone. \forall \tyf \fone \renvone \in \tyf \fset \renvone.
    \fone =\unenv{ \tyf \fone \renvone} \Rightarrow \forall \stone \in \sem \renvone {}. \srenv \renvone \stone \tyf \fone \renvone \Rightarrow  \stone \Vdash \fone.
  \end{equation*}
\end{lemma}
\begin{proof}
  The proof goes by induction on $\fone$, we treat the case of atomic propositions uniformly:
  \begin{proofcases}
    \proofcase[$\apone$] Let $\renvone$ be an environment that types all the expressions in $\apone$ (with a unique type if there are multiple expressions). Fix $\stone \in \sem \renvone{}$, and assume $\srenv \renvone \stone \tyf \apone \renvone$, which is $\stone \in \sem {\tyf \apone \renvone}{}$,
    observe that for every environment $\renvone$, we have $\sem {\tyf \apone \renvone}{}\subseteq\sem { \apone}{}$. This observation shows the claim.
    \proofcase[$\top, \bot$] Trivial.
    \proofcase[$\fone = \ftwo \land \fthree$] Fix an environment $\renvone$, and assume that $\tyf \fone \renvone$ is structurally identical to $\fone$ and that it belongs to $\tyf \fset \renvone$. This means that there are two \CPSL formulas $\tyf \ftwo \renvtwo \in \tyf \fset \renvtwo$ and $\tyf \fthree \renvthree \in \tyf \fset \renvthree$ (with $\renvtwo \ext \renvone$ and $\renvthree\ext \renvone$) such that $\tyf \fone \renvone = \tyf {\tyf \ftwo \renvtwo \land \tyf \fthree \renvthree}\renvone$ and notably:
    \[
      \ftwo =\unenv {\tyf \ftwo \renvtwo} \quad \fthree =\unenv{ \tyf \fthree \renvthree}.
      \tag{H1, H2}
    \]
    Now, we fix $\stone \in \sem \renvone {}$, and we assume that
    \[
      \srenv \renvone \stone \tyf \fone \renvone,
    \]
    which means that there are $\sttwo \in \sem \renvtwo {}$ and $\stthree \in \sem \renvthree {}$
    such that $\sttwo \ext \stone$ and $\stthree \ext \stone$, and notably:
    \[
      \srenv \renvtwo \sttwo \tyf \ftwo \renvtwo \quad \srenv \renvthree \stthree \tyf \fthree \renvthree
      \tag{HA, HB}
    \]
    Let $\tyf \ftwo \renvone$ and $\tyf \fthree \renvone$ be the formulas that are obtained by exchanging the most external environments of the formulas $\tyf \ftwo \renvtwo$ and $\tyf \fthree \renvthree$ (i.e. $\renvtwo$ and $\renvthree$) with $\renvone$. Observe that $\tyf \ftwo \renvtwo \ext \tyf \ftwo \renvone$ and  $\tyf \fthree \renvthree \ext \tyf \fthree \renvone$. For this reason, we can apply \Cref{lemma:prekm} to (HA)  and (HB), to show that 
    \[
      \srenv \renvone \stone \tyf \ftwo \renvone \quad \srenv \renvone \stone \tyf \fthree \renvone
      \tag{HA', HB'}
    \]
    From (H1) and (H2), we also deduce:
    \[
      \ftwo = \unenv{\tyf \ftwo \renvone} \quad \fthree =\unenv{\tyf \fthree \renvone}.
      \tag{H1', H2'}
    \]
    The claim is a consequence of the IHs on $\fone$ and $\ftwo$ and of the hypotheses (H1'), (HA'), (H2'), (HB').
    \proofcase[$\fone = \ftwo \sep \fthree$] Fix an environment $\renvone$, and assume that $\tyf \fone \renvone$ is structurally identical to $\fone$ and belongs to $\tyf \fset \renvone$. This means that there are two \CPSL formulas $\tyf \ftwo \renvtwo \in \tyf \fset \renvtwo$ and $\tyf \fthree \renvthree \in \tyf \fset \renvthree$ (with $\renvtwo \ext \renvone$ and $\renvthree \ext \renvone$) such that $\tyf \fone \renvone = \tyf {\tyf \ftwo \renvtwo \sep \tyf \fthree \renvthree}\renvone$ and notably:
    \[
      \ftwo =\unenv{\tyf \ftwo \renvtwo} \quad \fthree = \unenv{\tyf \fthree \renvthree}.
      \tag{H1, H2}
    \]
    Now, we fix $\stone \in \sem \renvone {}$, and we assume that
    \[
      \srenv \renvone \stone \tyf \fone \renvone,
    \]
    which means that there are $\sttwo \in \sem \renvtwo {}$ and $\stthree \in \sem \renvthree {}$
    such that $\sttwo \ind\ext \stone$ and $\stthree \ind\ext \stone$, and notably:
    \[
      \srenv \renvtwo \sttwo \tyf \ftwo \renvtwo \quad \srenv \renvthree \stthree \tyf \fthree \renvthree,
      \tag{HA, HB}
    \]
    and finally:
    \[
      \sttwo \tensprod \stthree \ind \ext \stone.
      \tag{HC}
    \]
    By applying the IHs on $\fone$ and $\ftwo$ and using the hypotheses (H1), (HA), (H2), (HB),
    we obtain that
    \[
      \sttwo \Vdash \ftwo \quad \stthree \Vdash \fthree,
    \]
    Then, we finally observe that this is sufficient to establish
    \[
      \stone \Vdash \ftwo \sep \fthree
    \]
    because of (HC).
  \end{proofcases}
\end{proof}

For the converse direction, we need to be more careful:
it is not true that if a conjunctive \BI formula $\fone$ is valid in a state $\stone$,
then the same formula is valid in the same state using the \CPSL semantics
\emph{independently} of the environments that we use to tag that formula.

\begin{example}
  Let $\renvone = {\rvarone, \rvartwo, t :\bool}$ and
  $\tyf \rsone \renvone \in \sem\renvone{}$ be the store that
  where $\rvarone$ and $\rvartwo$ are uniform and independent, and where
  for every value of the security parameter $\nat$,
  and for every $\mem \in \supp(\rsone_\nat)$
  $\mem(t) =\mem(\rvartwo)$. Observe that $\tyf \rsone \renvone$ satisfies
  $\indp\rvarone\rvarone \sep \indp\rvartwo\rvartwo$ according to the \BI semantics,
  but it does not satisfy every formula $\tyf \fone \renvone$ that is
  obtained by labeling the subformulas of $\fone$ with environments.
  In particular, this does not hold for the formula:
  \[
    \tyf{\tyf {\indp \rvarone \rvarone} {\rvarone, t:\bool} \sep \tyf {\indp \rvartwo \rvartwo} {\rvartwo:\bool}}\renvone
  \]
  because $\rvartwo$ and $t$ are not computationally independent.
  \hfill\qed
\end{example}

This shows that, given the \BI-validity of a conjunctive formula $\fone$
in a state $\stone\in \sem \renvone{}$, the \CPSL-validity
of a structurally identical formula $\tyf \fone \renvone$ on the same state,
depends on the underlying sets of independent variables.  
This happens because, due to the annotation of formulas with environments,
the semantics of the separating conjunction in \CPSL become less
ambiguous than in \BI, and the translation from \BI to \CPSL
requires this ambiguity to be resolved. However, notice that since we work with
conjunctive operators only, the tags of the two subformulas
of a separating conjunction can be syntactically \emph{under-approximated}
by taking the restriction of the global environment to the set of variables
that appear in these formulas. 

\begin{proposition}
  \label{prop:bitocpsl}
  For every environment $\renvone$, store $\stone \in \sem \renvone {}$,
  and conjunctive \BI formula $\fone$, if
  \(
  \stone \Vdash \fone,
  \)
  there is a formula $\tyf \fone \renvone$ of 
  \CPSL that is structurally identical to $\fone$ such that:
  \[
    \srenv \renvone \stone \tyf \fone \renvone.
  \]
  \hfill\qed
\end{proposition}

\begin{proof}
  The claim that we want to show is:
    \begin{equation*}
    \forall \fone \in \fset_{\BI}. \forall \renvone, \forall \stone \in \sem \renvone{}. \stone \Vdash \fone \Rightarrow\exists \tyf \fone \renvone \in \tyf \fset \renvone. \fone =\unenv{ \tyf \fone \renvone} \land  \srenv \renvone \stone \tyf \fone \renvone
  \end{equation*}
  \begin{proofcases}
    \proofcase[$\apone$] Trivial.
    \proofcase[$\top, \bot$] Trivial.
    \proofcase[$\fone = \ftwo \land \fthree$] We fix $\renvone$ and $\stone \in \sem \renvone{}$,
    and from the definition of the semantics, we observe:
    \[
      \stone \Vdash \ftwo \quad \stone \Vdash \fthree.
    \]
    Then, we apply the two IHs on the two subformulas, $\stone$ and $\renvone$, this shows the existence of two formulas $\tyf \ftwo \renvone$ and $\tyf \fthree\renvone$ such that $\ftwo =\unenv{\tyf \ftwo \renvone}$, $\fthree =\unenv{ \tyf \fthree \renvone}$, and notably:
    \[
      \srenv \renvone \stone \tyf \ftwo \renvone \quad \srenv \renvone \stone \tyf\fthree \renvone.
    \]
    Since $\stone \ext \stone$, we can conclude that
    \[
      \srenv \renvone \stone \tyf {\tyf \ftwo \renvone \land \tyf \fthree \renvone} \renvone.
    \]
    Finally, observe that $\ftwo \land \fthree \unenv {\tyf {\tyf \ftwo \renvone \land \tyf \fthree \renvone} \renvone}$.
    \proofcase[$\fone = \ftwo \sep \fthree$] we fix $\renvone$ and $\stone \in \sem \renvone{}$,
    and from the definition of the semantics, we observe that there are $\sttwo, \stthree \in \stset$ such that 
    \[
      \sttwo \tensprod \stthree \ind \ext \stone \quad \sttwo \Vdash \ftwo \quad \sttwo \Vdash \fthree.
      \tag{H1, H2, H3}
    \]
    From (H1), we can also conclude that $\sttwo \in \sem \renvtwo {}$ for some $\renvtwo \ext \renvone$
    and that $\stthree \in \sem \renvthree {}$ for some $\renvthree \ext \renvone$.
    Then, we apply the two IHs on the two subformulas, $\sttwo, \stthree$ and $\renvtwo, \renvthree$, using the hypotheses (H1) and (H2). This shows the existence of two formulas $\tyf \ftwo \renvtwo$ and $\tyf \fthree\renvthree$ such that $\ftwo =\unenv {\tyf \ftwo \renvtwo}$, $\fthree =\unenv {\tyf \fthree \renvthree}$, and notably:
    \[
      \srenv \renvtwo \sttwo \tyf \ftwo \renvtwo \quad \srenv \renvthree \stthree \tyf\fthree \renvthree.
    \]
    From (H1), we can conclude that:
    \[
      \srenv \renvone \stone \tyf {\tyf \ftwo \renvtwo \land \tyf \fthree \renvthree} \renvthree.
    \]
    Finally, observe that $\ftwo \sep \fthree =\unenv{ \tyf {\tyf \ftwo \renvtwo \sep \tyf \fthree \renvthree} \renvone}$.
  \end{proofcases}
\end{proof}

\subsubsection{\CPSL vs. \PSL}
\label{sec:cpslvsbi}

A comparison between \CPSL and Barthe et al.'s \PSL is now in order.
First of all, the semantics of the separating conjunction is more restrictive 
than that in \PSL: in addition to knowing that the substates which satisfy the 
sub-formulas are not statistically but \emph{computationally} independent, we 
also know the domains of their environments.
For instance, if:
\begin{equation}
  \label{eq:nodf}
\srenv \renvone \stone {\tyf{\tyf{\top}{m: \stng\nat} \sep \tyf{\ud{k}}{k: \stng\nat}}}{\renvone},
\end{equation}
we know
that there are two computationally independent substates of $\stone$,
namely $\sttwo, \stthree$,
such that $\sr\sttwo{\tyf{\top}{m: \stng\nat}}$,
$\sr\stthree{\tyf{\ud{k}}{k: \stng\nat}}$.
In addition to this, we also know that $\sttwo$ and $\stthree$ store
only $m$ and $k$, respectively.

Also observe that in \CPSL it is possible to state the
computational independence of (sets of) variables naturally:
with the condition $\sr\stone{\tyf{\top}{m: \stng\nat}} \sep {\tyf{\ud k}{k: \stng\nat}}$, we are
stating that the projection of $\stone$ on $m$ is computationally independent from
its projection on $k$, without imposing further assumptions on the value of $m$.
This semantics can be achieved in
\PSL by replacing $\tyf \top {m:\stng \nat}$ with a
tautological atomic proposition $\apone$ such that $\fv \apone =\{m\}$.

In contrast to what happens for \PSL, \CPSL's equations are always reflexive.
This is not always the case for \PSL: for instance,
$\indp\reone \reone$ is not valid,
as it does not hold, for instance, in the empty store.

\section{Inference Rules}\label{sec:inference}
\begin{figure*}[t]
  \centering
  $$
  \infer[\RSKIP]{\hotj{\fone}{\binjudg \renvone\pskip}{\fone}}{}\qquad
  \infer[\RSEQ]{\hotj{\fone}{\binjudg \renvone\seq{\prgone}{\prgtwo}}{\fthree}}
  {\hotj{\fone}{\binjudg \renvone\prgone}{\ftwo} & \hotj{\ftwo}{\binjudg \renvone\prgtwo}{\fthree}}\qquad
  $$
  $$
  \infer[\RASS]{\hotj{\ftrue}{\binjudg \renvone{\ass{\rvarone}{\reone}}}{\eq{\rvarone}{\reone}}}
  { \binjudg \renvone{\ass \rvarone \reone} &
    \rvarone\not\in\fv{\reone}}
  \qquad
    \infer[\DASS]{\hotj{\ftrue}{\binjudg \renvone{\ass{\rvarone}{\deone}}}{\espl{\rvarone}{\deone}}}
  { \binjudg \renvone{\ass \rvarone \deone} &
    \rvarone\not\in\fv{\deone}}
  $$
  $$  
  \infer[\SRASS]{\hot{\tyf{\tyf \fone {\renvtwo}\sep\tyf \ftwo \renvthree}\renvone} {\binjudg \renvone \ass \rvarone \reone}{\tyf{\tyf{\tyf \fone \renvtwo \land \tyf {\eq \rvarone \reone} {\renvtwo\cup \{\rvarone :\tyone\}}} {\renvtwo\cup \{\rvarone :\tyone\}}\sep\tyf \ftwo {\renvthree\setminus\{\rvarone:\tyone\}}}\renvone}}{
    \ternjudg\renvtwo \reone \tyone & \rvarone \notin \fv \reone &\tyf \ftwo {\renvthree\setminus\{\rvarone:\tyone\}} \in  \fset & \rvarone \notin \dom(\renvtwo) }
  $$
  $$
  \infer[\SDASS]{\hot{\tyf{\tyf \fone {\renvtwo}\sep\tyf \ftwo \renvthree}\renvone} {\binjudg \renvone \ass \rvarone \deone}{\tyf{\tyf{\tyf \fone \renvtwo \land \tyf {\espl \rvarone \deone} {\renvtwo\cup \{\rvarone :\tyone\}}} {\renvtwo\cup \{\rvarone :\tyone\}}\sep\tyf \ftwo {\renvthree\setminus\{\rvarone:\tyone\}}}\renvone}}{
    \ternjudg\renvtwo \deone \tyone & \rvarone \notin \fv \deone &\tyf \ftwo {\renvthree\setminus\{\rvarone:\tyone\}} \in  \fset & \rvarone \notin \dom(\renvtwo) }
  $$
  $$
  \infer[\RCONDCM]{\hotj{\df \rvarone}{\binjudg \renvone{\ifr\rvarone \prgone 
  \prgtwo}}{\ftwo}}{
  	\binjudg \renvone{\ifr \rvarone \prgone \prgtwo} &
  	\hotj{\espl r 1}{\binjudg \renvone\prgone}{\ftwo} &
  	\hotj{\espl r 0 }{\binjudg \renvone\prgtwo}{\ftwo} &
  	\tyf \ftwo\renvone \in \exfset
  }
  $$
  \caption{Computational Rules}
  \label{fig:comprules}
\end{figure*}

In this section, we show how \CPSL can be turned into a Hoare logic for 
reasoning about computational independence and pseudorandomness. We first 
define Hoare triples, then we provide a set of rules for reasoning about 
probabilistic programs. Finally, we prove the soundness of our system, and we
show how it can be used to reason about the security of some classic 
cryptographic primitives: in \Cref{sec:potp} we use this system of rules
to show the computational secrecy of
a variation of the OTP, where the key \emph{is not} drawn from a 
\emph{statistically} uniform distribution, but rather from a \emph{pseudorandom} one.
In \Cref{subsec:stretching}, we show that our rule for $\mathtt{if}$
statements can be used to show the correctness of a program that computes
the exclusive or.
Finally, in \Cref{subsec:stretching} we use this logic to show the security of 
the standard construction in which a pseudorandom generator with low expansion 
factor is composed with itself to obtain another pseudorandom generator 
with higher expansion factor.

\subsection{Hoare Triples}\label{subsec:rules}
A \emph{Hoare triple} is an expression in the form
\begin{equation}\label{eq:hot}
\hot{\tyf \fone \renvone }{\binjudg{\renvone}{\prgone}}{\tyf \ftwo \renvone}
\end{equation}
For the sake of simplicity, we write $\hot{\fone}{\binjudg \renvone \prgone}{\ftwo}$
or $\hot{\fone}{\prgone}{\ftwo}$
instead of 
(\ref{eq:hot}) when this does not create any ambiguity. The Hoare triple
(\ref{eq:hot}) is said to be \emph{valid} if
$\binjudg \renvone \prgone$ and for every store $\rsone \in \sem \renvone{}$ such 
that $\sr{\tyf \rsone \renvone}{\tyf \fone \renvone }$, it holds that 
$\sr{\sem{\binjudg \renvone\prgone}{}(\tyf \rsone \renvone)}{\tyf \ftwo \renvone}$. In that case, 
we write
$$
\models\hot{\fone}{\binjudg{\renvone}{\prgone}}{\ftwo}.
$$
Fortunately, Hoare triples do not need to be proved valid by hand, and we can 
define a certain number of inference rules, which can then be proved to be 
validity-preserving. In the following subsection, we are going to introduce two 
categories of inference rules, namely computational and structural rules. Each 
rule in the former class matches a specific instruction and does not require 
any other triple as a premise. As such, computational rules are meant to
draw conclusions about how the state evolves when the corresponding instruction
is executed. On the other hand, structural rules are used to obtain
a new triple about a given program $\prgone$ from another one about $\prgone$ 
itself, this way altering preconditions and postconditions.

\begin{figure*}[t]
\centering
\begin{align*}
  \mv{\pskip} & \defsym \emptyset \\
  \mv{\ass \rvarone \reone} & \defsym \{\rvarone\}\\
  \mv{\seq\prgone \prgtwo} & \defsym \mv \prgone \cup \mv \prgtwo \\
\mv{\ifr \reone \prgone \prgtwo} & \defsym \mv \prgone \cup \mv \prgtwo
 \end{align*}    
\caption{Definition of $\mv \cdot$.}
\label{fig:rvmvwvdef}
\end{figure*}
%
%
\begin{figure*}[t]
  \centering
  $$
  \infer[\RWEAK]{\hotj{\tyf \fone \renvone}{\binjudg \renvone \prgone}{\tyf \ftwo \renvone}}{\hotj{\tyf \fthree \renvone}{\binjudg \renvone \prgone}{\tyf \ffour \renvone}
    &
    \tyf \fone \renvone \models\tyf \fthree\renvone  & \tyf \ffour \renvone \models \tyf \ftwo\renvone }
  $$
  $$
  \infer[\RCONST]{\hotj{\tyf{\tyf \fone\renvtwo \wedge \tyf \fthree\renvthree}\renvone}{\binjudg \renvone \prgone}{\tyf{\tyf \ftwo\renvtwo \wedge\tyf \fthree\renvthree}\renvone}}
  {\hotj{\tyf \fone\renvtwo }{\binjudg \renvtwo \prgone}{\tyf \ftwo\renvtwo} & \dom(\renvthree)\cap\mv{\prgone}=\emptyset}
  $$
  \vspace{-1mm}
  $$
  \infer[\RFRAME]{\hotj{\tyf{\tyf \fone \renvtwo \sep\tyf \fthree\renvthree}\renvone}{\binjudg \renvone \prgone}{\tyf{\tyf \ftwo \renvtwo \sep\tyf \fthree \renvthree}\renvone}}
  {\hotj{\tyf \fone\renvtwo}{\binjudg \renvtwo \prgone}{\tyf \ftwo\renvtwo}
  }
  $$ 
  \caption{Structural Rules}
  \label{fig:structrules}
\end{figure*}

\subsection{Computational Rules}

\CPSL's computational rules can be found in Figure \ref{fig:comprules}.
The rules $\RSKIP$ and $\RSEQ$ are self explanatory, while the rules for 
assignments and conditionals require some more discussion.
First of all, we note the presence of four different rules for assignments. two 
of them, namely $\RASS$ and $\SRASS$ are concerned with arbitrary 
expressions and the $\eq{\cdot}{\cdot}$ predicate. The other two, namely 
$\DASS$ and $\SDASS$, are instead concerned with deterministic expressions 
and the $\espl{\cdot}{\cdot}$ predicate.

$\RASS$ is arguably the simplest of the four rules, and is present in Barthe 
et al.'s \PSL. The $\DASS$ rule, instead can be applied only to deterministic 
assignments and yields a very precise conclusion: a state $\rsone$ satisfies 
the conclusion $\espl{\rvarone}{\deone}$
if and only if, for every $\nat \in \NN$ and every sample $\memone$
of $\rsone_\nat$, the semantics of the deterministic
expressions $\deone$ evaluated on $\memone$ is equal to the
deterministic semantics of $\rvarone$ evaluated on $\memone$.
In particular, this rule is useful for indistinguishability proofs.
An example of fruitful use of this rule can be found in
\Cref{sec:potp} below.
\revision{
  The \emph{separating assignment rules} $\SRASS$ and $\SDASS$ are used
  to produce triples with separating conjunctions in the pre- and post-conditions.
  Intuitively, these rules are meant to describe how independence propagates 
  through assignments: if there are two independent
  portions of a store, and an expression that can be
  evaluated in one of these is assigned to a variable $\rvarone$, then
  $\rvarone$ is independent from the other portion of the store.
  In addition, the assignment does not invalidate any other
  formula $\tyf \ftwo\renvthree$ that is already valid in that portion of
  the store, provided that $\ftwo$ does not mention $\rvarone$. 
 }
Finally, notice that the standard assignment rules 
are not instances of the separating ones. For example,
if we show $\tyf {\tyf{\eq \reone \retwo}\renvtwo \sep \tyf\fone \renvthree}\renvone$,
We cannot deduce $\tyf{\eq\reone \retwo}\renvone$ from it, because \Cref{lemma:indcbi}
holds only for approximate rules.

Let us now switch our attention to the rule $\RCONDCM$. This rule has a very
specific formulation. This is due to the peculiar semantic underlying \CPSL, 
which is based on \emph{efficiently samplable} distribution ensembles.
In particular, as far as the authors of this
paper know, it is not known whether efficiently samplable
distribution ensembles are closed under conditioning.
\revision{
  For this reason, when showing the soundness of this rule, we cannot
  condition the initial state on the value of the guard, 
  as it happens in~\cite{PSL}. As a consequence, it is challenging
  to state this rule with a more general pre-condition.
}
In \Cref{sec:examples} we show how this rule, although 
having a limited expressive power, can nevertheless be employed on a small 
example dealing with the exclusive or.
\commentout{
  This is done by proving the triple:
\[
  \hotj{\tyf\ftrue\renvone}{\binjudg\renvone\XOR}{\tyf{\eq c {\exor(k, m)}}\renvone}
\]
for $\renvone ={k,c,m:\bool}$ and the following program:
\[
  \XOR \defsym \ifr{k = 1}{\ass c \lnot m}{\ass c m}.
\]
}

\begin{proposition}
\label{lemma:compsound}
The computational rules are sound.
\hfill\qed
\end{proposition}

\begin{proof}
  We do not show the validity of the typing judgment associated to a triple, because it is trivial. For the logical soundness, we want to show that every triple $\hot {\tyf \fone \renvone} {\binjudg \renvone \prgone}{\tyf \ftwo \renvone}$
  that we can deduce using the computational rules is sound, namely:
  \[
    \forall \rsone \in \sem \renvone.\left(\sr{\tyf \rsone \renvone}{\tyf \fone \renvone}\right) \Rightarrow\left( \sr{\sem {\prgone} {}({\tyf \rsone \renvone})}{{\tyf \fthree \renvone}}\right)
  \]
  We go by induction on the proof of $\hotj{{\tyf \fone \renvone}}{\binjudg \renvone\prgone}{{\tyf \ftwo \renvone}}$:
  \begin{proofcases}
  \proofcase[$\RSKIP$] The claim is
    \begin{equation*}
        \forall \rsone \in \sem \renvone. \left(\sr{{\tyf \rsone \renvone}}{\tyf \fone \renvone}\right) \Rightarrow
        \left( \sr{\sem{\binjudg \renvone \pskip}{}({{\tyf \rsone \renvone}})}{{\tyf \fone \renvone}}\right).
    \end{equation*}
    \noindent
    This holds because $\sem {\binjudg \renvone \pskip}{}({{\tyf \rsone \renvone}})={{\tyf \rsone \renvone}}$.
    \proofcase[$\RSEQ$] We are asked to show:
      \begin{equation*}
        \forall \rsone \in \sem \renvone{}.\left(\sr{{\tyf \rsone \renvone}}{\tyf \fone \renvone}\right) \Rightarrow
        \left( \sr{\sem {\binjudg\renvone {\seq\prgone\prgtwo} } {}({{\tyf \rsone \renvone}})}{{\tyf \fthree \renvone}}\right)
      \end{equation*}
      under the rule's preconditions.
      Since we know that the rule $\RSEQ$ was applied, we also know that
      there is ${\tyf \ftwo \renvone}$ such that $\hotj{{\tyf \fone \renvone}}{\binjudg \renvone\prgone}{{\tyf \ftwo \renvone}}$ and $\hotj{{\tyf \ftwo \renvone}}{\binjudg \renvone\prgtwo}{{\tyf \fthree \renvone}}$,
      the two IHs on these sub-derivation state:
      \begin{equation*}
      \forall \rsone \in \sem \renvone{}.\left(\sr{{\tyf \rsone \renvone}}{\tyf \fone \renvone}\right) \Rightarrow
        \left( \sr{\sem {\binjudg \renvone\prgone} {}({{\tyf \rsone \renvone}})}{{\tyf \ftwo \renvone}}\right)  
      \end{equation*}
      and
      \begin{equation*}
        \forall \rsone \in \sem \renvone{}.\left(\sr{{\tyf \rsone \renvone}}{\tyf \ftwo \renvone}\right) \Rightarrow
        \left( \sr{\sem {\binjudg \renvone\prgtwo} {}({{\tyf \rsone \renvone}})}{{\tyf \fthree \renvone}}\right)
      \end{equation*}
      Fix ${{\tyf \rsone \renvone}}$ and assume $\left(\sr{{\tyf \rsone \renvone}}{\tyf \fone \renvone}\right)$, the goal of the sub-derivation becomes:
      \[
        \sr{\sem {\binjudg \renvone{\seq\prgone\prgtwo}} {}({{\tyf \rsone \renvone}})}{{\tyf \fthree \renvone}}
      \]
      which is
      \[
         \sr{\sem {\binjudg \renvone\prgtwo}{}\left(\sem {\binjudg \renvone\prgone}{}({{\tyf \rsone \renvone}})\right)}{{\tyf \fthree \renvone}}.
      \]
      From the first IH we obtain $\sr{\sem {\prgone}{}({{\tyf \rsone \renvone}})}{{\tyf \ftwo \renvone}}$, using this hypothesis as premise of the second IH, we get the claim.
\proofcase[$\RASS$]
  In this sub-derivation, the goal is to show:
  \begin{equation*}
    \forall \rsone \in \sem \renvone{}.\left(\sr{{\tyf \rsone \renvone}} \top\right) \Rightarrow\left( \sr{\sem {\binjudg \renvone{\ass \rvarone \reone}} {}({{\tyf \rsone \renvone}})}{\eq \rvarone \reone}\right)
  \end{equation*}
  under the rule's preconditions ($\rvarone \notin \fv \reone$).
  This claim can be restated as:
  \[
    \forall \rsone \in \sem \renvone{}.\left( \sr{\sem {\binjudg \renvone{\ass \rvarone \reone}} {}({{\tyf \rsone \renvone}})}{{\eq \rvarone \reone}}\right)
  \]
  From now on, due to the lack of ambiguity we are going to omit the judging statement from the notation of the triples. The claim can be expanded as:
  \begin{align*}
    \sem \rvarone {}\left(\sem {\ass \rvarone \reone} {}(\rsone)\right) = \sem \reone {}\left(\sem {\ass \rvarone \reone} {}(\rsone)\right).
  \end{align*}
  We will pass through the following equivalence and concluding for transitivity:
  \begin{align*}
    \sem \rvarone {}\left(\sem {\ass \rvarone \reone} {}(\rsone)\right) = \sem {\reone} {}(\rsone) = \sem \reone {}\left(\sem {\ass \rvarone \reone} {}(\rsone)\right).
  \end{align*}
  Which is the claim of \Cref{lemma:exprcompsem}.
\proofcase[$\RCONDCM$] We want to show that $\models \hot{\top}{\binjudg \renvone \prgone}{\ftwo}$, which is equivalent to:
  \[
    \begin{multlined}
    \forall \rsone\in \sem \renvone{}. \sr{\sem \prgone{}({\tyf \rsone \renvone})}{\tyf\ftwo\renvone},  
    \end{multlined}
  \]
  with the following assumptions:
  \begin{proofcases}
  \proofcase[H1] $\hotj{\espl r 1}{\binjudg \renvone \prgone}{\ftwo}$,
  \proofcase[H2] $\hotj{\espl r 0 }{\binjudg \renvone \prgtwo}{\ftwo}$,
  \proofcase[H3] $\tyf\ftwo\renvone \in \exfset$.
  \end{proofcases}
  We assume that $\sr {{\tyf \rsone \renvone}} {\df \rvarone}$.
  Assume that for some $b$:
  \[
    \smashoperator[r]{\sum_{{\memone: \memone \in \supp(\rsone_\nat) \land \memone(\rvarone)=b}}}\;\;\rsone_\nat(\memone) = 0,
    \tag{$*$}
  \]
  In this case, it is possible to show that the
  semantics of the statement is equal to $\sem \prgone{}(\stone)$
  or to $\sem \prgtwo {} (\stone)$ and to conclude with one of the IHs.
  In the remaining part of the proof, we can
  assume without loss of generality that $(*)$ does not hold for any $b \in \BB$.
  In this case, we expand the semantics of the statement obtaining that,
  for some fixed $\nat \in \NN$,
  it is equal to:
  \[
    \gbind{\rsone_\nat}{\memone \mapsto
      \begin{cases}
        \sem \prgone\nat(\unit\memone) & \text{if }\memone(\rvarone)=1\\
        \sem \prgtwo\nat(\unit\memone) & \text{if }\memone(\rvarone)=0
      \end{cases}
                                         },
  \]
  for every $\memone$ and $b$,
  let $\rstwo^{\nat',\memone,b}\in \sem\renvone{}$ be a distribution ensemble
  such that $\rstwo^{\nat',\memone,b}_{\nat}(m)=1$ if $\nat'=\nat$,
  and that for every $\nat\neq \nat'$,
  and every $\overline \mem \in \supp(\rstwo^{\nat'\memone,b}_\nat)$
  $\parsem\rvarone{\overline\nat'} (\overline \mem)=b$.
  Notice that every $\rstwo^{\nat',\memone,b}$ is a family of efficiently samplable
  distributions such that $\sr{\tyf{\rstwo^{\nat',\memone,b}}\renvone}\tyf{\espl\rvarone b}\renvone$ if
  $\memone(\rvarone)=b$.
  The semantics
  of the conditional can
  be restated as follows: 
  \[
    \gbind{\rsone_\nat}{\memone \mapsto
      \begin{cases}
        \sem \prgone\nat(\rstwo^{\nat,\memone,1}_\nat) & \text{if }\memone(\rvarone)=1\\
        \sem \prgtwo\nat(\rstwo^{\nat,\memone,0}_\nat) & \text{if }\memone(\rvarone)=0
      \end{cases}
    },
  \]
  which in turn is equal to:
  \[
    \gbind{\rsone_\nat}{\memone \mapsto
    \begin{cases}
      \left(\sem \prgone{}{(\rstwo^{\nat,\memone, 1})}\right)_\nat & \text{if }\memone(\rvarone)=1\\
      \left(\sem \prgtwo{}{(\rstwo^{\nat,\memone, 0})}\right)_\nat & \text{if }\memone(\rvarone)=0
    \end{cases}
    }.
  \]
  We can assume without loss of generality that $\ftwo$
  is a conjunction of $\tyf{\apone({\reone_i}, {\retwo_i})}{\renvtwo_i}$ for $0\le i\le k$,
  where $\apone$ can either be $\eq\cdot \cdot$ or $\espl \cdot\cdot$.
  Indeed, we know that $\tyf\fone\renvone$ is a conjunction of smaller formulas
  by definition of exact formulas.
  If any of these smaller formulas or propositions is $\tyf\bot{\renvtwo_j}$,
  we conclude ad absurdum, because no state satisfies such formula,
  conversely, if any of these smaller formulas is $\tyf \top {\renvtwo_j}$,
  we can show that the final state satisfies the formula $\tyf{\ftwo'}\renvone$
  obtained removing this formula from the conjunction, which is equivalent to $\ftwo$.
  No other cases are possible.
  If $\tyf \ftwo \renvone$ has this shape, our claim is that there are two finite set of triples
  \begin{itemize}
  \item $E=\{(\tyf {\reone_i}{\renvtwo_i}, \tyf {\retwo_i}{\renvtwo_i}, \tyone_i)\mid  0\le i\le k\}$ 
  \item $S=\{(\tyf {\deone_i}{\renvtwo_i}, \tyf {\detwo_i}{\renvtwo_i}, \tyone_i)\mid  0\le i\le l\}$ 
  \end{itemize}
  such that:
  \begin{itemize}
  \item for every triple $(\tyf {\reone_i}{\renvtwo_i}, \tyf {\retwo_i}{\renvtwo_i}, \tyone_i)\in E$,
    it holds that
    $$
    \sem{\ternjudg{\renvtwo_i}{\reone_i}{\tyone_i}}{}({\sem{\ifr\rvarone\prgone\prgtwo}{}}(\rsone)_{\renvone\to\renvtwo_i})=\sem{\ternjudg{\renvtwo_i}{\retwo_i}{\tyone_i}}{}({\sem{\ifr\rvarone\prgone\prgtwo}{}}(\rsone)_{\renvone\to\renvtwo_i}).$$
  \item for every triple $(\tyf {\deone_i}{\renvtwo_i}, \tyf {\detwo_i}{\renvtwo_i}, \tyone_i)\in S$,
    every $\nat \in \NN$ and
    $$
    \mem \in \supp({\sem{\ifr\rvarone\prgone\prgtwo}{\nat}}(\rsone_\nat)_{\renvone\to\renvtwo_i}),
    $$
    we have:
    $$
    \parsem{\ternjudg{\renvtwo_i}{\deone_i}{\tyone_i}}\nat(\mem)=\parsem{\ternjudg{\renvtwo_i}{\detwo_i}{\tyone_i}}\nat(\mem).
    $$
  \end{itemize}
  We start by showing the first conclusion.
  For \Cref{lemma:projexpsem}, we can show an equivalent result:
  \begin{equation*}
    \sem{\ternjudg{\renvone}{\reone_i}{\tyone_i}}{}({\sem{\ifr\rvarone\prgone\prgtwo}{}}(\rsone))=\\
      \sem{\ternjudg{\renvone}{\retwo_i}{\tyone_i}}{}({\sem{\ifr\rvarone\prgone\prgtwo}{}}(\rsone))
  \end{equation*}
  We fix some $i$, and we show that this result holds for $\reone_i$, $\retwo_i$.
  This means that for every $n$, we require:
  \begin{equation*}
    \sem{\ternjudg{\renvone}{\reone_i} {\tyone_i}}{\nat}({\sem{\ifr\rvarone\prgone\prgtwo}{\nat}}(\rsone))= \sem{\ternjudg{\renvone}{\retwo_i}{\tyone_i}}{\nat}({\sem{\ifr\rvarone\prgone\prgtwo}{\nat}}(\rsone))
        \tag{C}
  \end{equation*}
  Fix $\nat$, it holds that the term on the left is equivalent to 
  \begin{equation*}
    \bind{\bind{\rsone_\nat}{h}}{\detsem {\reone_i} \nat}=\bind{\rsone_\nat} {\memone\mapsto\bind{h(\memone)} {\detsem {\reone_i}\nat}}.
  \end{equation*}
  where $h(\memone)$ is
  \[
    \begin{cases}
      \left(\sem \prgone{}{(\rstwo^{\nat,\memone, 1})}\right)_\nat & \text{if }\memone(\rvarone)=1\\
      \left(\sem \prgtwo{}{(\rstwo^{\nat,\memone, 0})}\right)_\nat & \text{if }\memone(\rvarone)=0.
    \end{cases}
  \]
  It is possible to express in these terms also the expression on the right of (C).
  For this reason, it suffices to conclude the proof by showing that for every
  $\memone \in \detsem \renvone\nat$:
  \[
    \bind{h(\memone)} {\detsem {\reone_i}\nat}=\bind{h(\memone)} {\detsem {\retwo_i}\nat}.
  \]
  Thus, we fix $\memone$, and we assume without loss of generality
  that $\memone(\rvarone)=1$. The claim reduces to:
  \begin{equation*}
    \bind{\left(\sem \prgone{}{(\rstwo^{\nat,\memone, 1})}\right)_\nat}{\detsem{\reone_i}\nat}=
    \bind{\left(\sem \prgone{}{(\rstwo^{\nat,\memone, 1})}\right)_\nat}{\detsem{\retwo_i}\nat},
  \end{equation*}
  which is:
  \[
    \sem {\reone_i}\nat({\sem \prgone{}{(\rstwo^{\nat,\memone, 1})}}_\nat)=
    \sem {\retwo_i}\nat({\sem \prgone{}{(\rstwo^{\nat,\memone, 1})}}_\nat).
  \]
  
  For the IH, we conclude that $\srenv\renvone{{\sem \prgone{}{(\rstwo^{\nat,\memone, 1})}}}{\tyf\ftwo\renvone}$. So, in particular, we have that:
  \[
    \sem{\ternjudg{\renvtwo_i} {\reone_i}{\tyone_i}}{}({\sem \prgone{}{(\rstwo^{\nat,\memone, 1})}}_{\renvone\to\renvtwo_i})=\sem{\ternjudg{\renvtwo_i} {\retwo_i}{\tyone_i}}{}({\sem \prgone{}{(\rstwo^{\nat,\memone, 1})}}_{\renvone\to\renvtwo_i}).
  \]
  By applying  \Cref{lemma:projexpsem},
  we obtain:
  \[
    \sem{{\reone_i}}{}({\sem \prgone{}{(\rstwo^{\nat,\memone, 1})}})=\sem{{\retwo_i}}{}({\sem \prgone{}{(\rstwo^{\nat,\memone, 1})}}),
  \]
  in particular:
  \[
    \sem{{\reone_i}}{\nat}({\sem \prgone{}{(\rstwo^{\nat,\memone, 1})}}_\nat)=\sem{{\retwo_i}}{\nat}({\sem \prgone{}{(\rstwo^{\nat,\memone, 1})}}_\nat),
  \]
  which is our claim.
  Now we go to the second result; in particular we fix some $i, \nat$, and we show a stronger result: for every $\mem \in \supp(({\sem{\ifr\rvarone\prgone\prgtwo}{}}(\rsone)_{\renvone\to \renvtwo_i})_\nat)$, we have that $\parsem{\ternjudg{\renvtwo_i}{\deone_i}{\tyone_i}}\nat(\mem)=\parsem{\ternjudg{\renvtwo_i}{\detwo_i}{\tyone_i}}\nat(\mem)$.
  From $\memone \in \supp(({\sem{\ifr\rvarone\prgone\prgtwo}{}}(\rsone)_{\renvone\to \renvtwo_i})_\nat)$, and \Cref{rem:suppext}, we deduce that there is $\memone'\in \supp(({\sem{\ifr\rvarone\prgone\prgtwo}{}}(\rsone))_\nat)$ such that $\memone'_{\dom(\renvtwo_i)}= \memone$.
  From \Cref{lemma:suppbackcomp}, we deduce that there is $\overline \memone \in \supp(\rsone_\nat)$ such that $\memone' \in \supp(\bind{h(\overline \mem)}{\memthree \mapsto \memthree\restr{\dom(\renvtwo_i)}})$ where $h(\memone)$ is
  \[
    \begin{cases}
      \left(\sem \prgone{}{(\rstwo^{\nat,\memone, 1})}\right)_\nat & \text{if }\memone(\rvarone)=1\\
      \left(\sem \prgtwo{}{(\rstwo^{\nat,\memone, 0})}\right)_\nat & \text{if }\memone(\rvarone)=0.
    \end{cases}
  \]
  We assume without lack of generality that \(\overline \mem(\rvarone)=1\). In particular, this means that $\memone'\in \supp(\bind{\left(\sem \prgone{}{(\rstwo^{\nat,\memone, 1})}\right)_\nat}{\memthree \to \memthree\restr{\dom(\renvtwo_i)}})=\supp({((\sem \prgone{}{(\rstwo^{\nat,\memone, 1})})_{\renvone \to \renvtwo_i})_\nat})$. For IH, we conclude that $\srenv\renvone{{\sem \prgone{}{(\rstwo^{\nat,\memone, 1})}}}{\tyf\ftwo\renvone}$. So, in particular, we know that for every $\nat \in \NN$ and $\memthree \in \supp(((\sem \prgone{}{(\rstwo^{\nat,\memone, 1})})_{\renvone\to \renvtwo_i})_\nat)$, we have that $\parsem{\ternjudg{\renvtwo_i}{\deone_i}{\tyone_i}}\nat(\memthree)=\parsem{\ternjudg{\renvtwo_i}{\detwo_i}{\tyone_i}}\nat(\memthree)$. In particular, this holds for $\memone'$, but it also holds for $\memone$ because of \Cref{lemma:parsemproj}.
  \proofcase[$\DASS$]
  We start by assuming fixing $\renvone$, a state $\stone =\rsone^\renvone$, and by assuming that $\binjudg \renvone \ass \rvarone \deone$ and that $\rvarone \notin \fv \deone$. Observe that $\srenv \renvone \stone \top$. For this reason, the goal is to show that we have $\srenv \renvone \stone {\espl \rvarone \deone}$, which is equivalent to showing that for every $\memtwo \in \supp( \sem {\ass \rvarone \deone}{}(\stone))$, we have $\forall \nat\in \NN. \parsem \rvarone {\nat}(\memtwo)= \parsem \deone{\nat}(\memtwo)$. This is exactly the claim of \Cref{lemma:detexprcompsem}.
  \proofcase[$\SRASS$]   We start by assuming that there is a state $\stone \in \sem \renvone{}$ that satisfies the premise, this means that there are $\stone_1$, $\stone_2$ such that:
  \begin{proofcases}
    \proofcase[H1] $\srenv \renvtwo {\stone_1} \tyf \fone \renvtwo$. 
    \proofcase[H2] $\srenv \renvthree {\stone_2} \tyf \ftwo \renvthree$. 
    \proofcase[H3] $\stone_1 \comp {\stone_2} \ind \ext \stone$. 
  \end{proofcases}
  The goal is to show that there are $\sttwo_1 \in \sem {\renvtwo\cup\{\rvarone:\tyone\}}{}$, $\sttwo_2 \in \sem {\renvthree\setminus\{\rvarone:\tyone\}}{}$
  such that 
  \begin{proofcases}
    \proofcase[C1] $\srenv {\renvtwo\cup\{\rvarone:\tyone\}} {\sttwo_1} \tyf \fone \renvtwo \land \tyf {\eq \rvarone \reone} {\renvtwo\cup\{\rvarone:\tyone\}}$. 
    \proofcase[C2] $\srenv {\renvthree\setminus\{\rvarone:\tyone\}} {\sttwo_2} \tyf \ftwo {\renvthree\setminus\{\rvarone:\tyone\}}$. 
    \proofcase[C3] $\sttwo_1 \comp {\sttwo_2} \ind \ext \sem {\ass \rvarone \reone}{}(\stone)$. 
  \end{proofcases}
  We define $\renvfour$ as $\{\rvarone: \tyone\}$
  and we call $\stfour \in \sem \renvfour {}$  the store that for every $\nat \in \NN$ is the Dirac distribution on a memory $\overline \mem_\nat \in \detsem \renvfour{\nat}$.
  The store $\sttwo_1$, is  $\sem{\binjudg {\join \renvtwo \renvfour}{\ass\rvarone  \reone}}{}(\stone_1\comp \stfour)$, and the store $\sttwo_2$ is ${\stone_2}_{\renvthree\to\renvthree\setminus\{\rvarone:\tyone\}}$.
  From \Cref{lemma:mvproj}, we deduce that $\sem{\binjudg {\join \renvtwo \renvfour}{\ass\rvarone  \reone}}{}(\stone_1\comp \stfour)_{\join \renvtwo \renvfour\to \renvtwo} = {(\stone_1\comp \stfour)}_{\join \renvtwo \renvfour\to \renvtwo} = \stone_1$. This shows that $\sr{(\sttwo_1)_{\join \renvtwo \renvfour\to \renvtwo}}{\tyf \fone \renvtwo}$; in order to show the claim (C1), it remains to show that $\sr{\sttwo_1}{\eq \rvarone \reone}$ holds as well, but this is a consequence of \Cref{lemma:exprcompsem}.
  The proof of (C2) relies on \Cref{lemma:prekm} and (H2). Then we go to (C3): by definition, we have that  $\sttwo_1\comp\sttwo_2 = \sem{\binjudg {\renvtwo \cup \renvfour}{\ass\rvarone  \reone}}{}(\stone_1\comp \stfour) \comp {\stone_2}_{\renvthree\to\renvthree\setminus\{\rvarone:\tyone\}}$  from \Cref{lemma:sempartcomm}, we conclude that:
  $\sem{\binjudg {\renvtwo \cup \renvfour}{\ass\rvarone  \reone}}{}(\stone_1\comp \stfour) \comp {\stone_2}_{\renvthree\to\renvthree\setminus\{\rvarone:\tyone\}} =    \sem{\binjudg {\join{\renvtwo \cup \renvfour}{(\renvthree\setminus\{\rvarone:\tyone\} )}}  {\ass\rvarone  \reone}}{}(\stone_1\comp \stfour  \comp {\stone_2}_{\renvthree\to\renvthree\setminus\{\rvarone:\tyone\}})$. From compatibility of $\ext$ with $\comp$ \Cref{lemma:extcompcomp}, and from reflexivity of $\ext$, we deduce that
  $\stone_1\comp {\stone_2}_{\renvthree\to\renvthree\setminus\{\rvarone:\tyone\}} \ind\ext \stone_1\comp\stone_2$, so from (H3), transitivity of $\ind\ext$ and this observation, we conclude $\stone_1\comp {\stone_2}_{\renvthree\to\renvthree\setminus\{\rvarone:\tyone\}} \ind\ext \stone$. From the definition of $\ind\ext$, we deduce that $\stone_{\renvone\to \join\renvtwo{\renvthree\setminus \{\rvarone:\tyone\}}}\ind  \stone_1\comp {\stone_2}_{\renvthree\to\renvthree\setminus\{\rvarone:\tyone\}}$. From this observation and compatibility of $\ind$ with $\comp$, \Cref{lemma:indcompcomp}, we deduce that: $\stone_{\renvone\to \join\renvtwo{\renvthree\setminus \{\rvarone:\tyone\}}} \comp \stfour \ind  \stone_1\comp {\stone_2}_{\renvthree\to\renvthree\setminus\{\rvarone:\tyone\}}\comp \stfour$, and we conclude with \Cref{lemma:prgind}, by showing that: $\sem{\binjudg {\join{\renvtwo \cup \renvfour}{(\renvthree\setminus\{\rvarone:\tyone\} )}}  {\ass\rvarone  \reone}}{}(\stone_1\comp {\stone_2}_{\renvthree\to\renvthree\setminus\{\rvarone:\tyone\}}\comp \stfour) \ind \sem{\binjudg {\join{\renvtwo \cup \renvfour}{(\renvthree\setminus\{\rvarone:\tyone\} )}}  {\ass\rvarone  \reone}}{}(\stone_{\renvone\to \join\renvtwo{\renvthree\setminus \{\rvarone:\tyone\}}}\comp \stfour)$. From \Cref{lemma:sepassntech}, we also deduce that: $\sem{\binjudg {\join{\renvtwo \cup \renvfour}{(\renvthree\setminus\{\rvarone:\tyone\} )}}  {\ass\rvarone  \reone}}{}(\stone_{\renvone\to \join\renvtwo{\renvthree\setminus \{\rvarone:\tyone\}}}\comp \stfour) = \sem{\binjudg {\join{\renvtwo \cup \renvfour}{(\renvthree\setminus\{\rvarone:\tyone\} )}}  {\ass\rvarone  \reone}}{} (\stone_{\renvone\to \join{\join\renvtwo\renvthree}\renvfour})$. From \Cref{lemma:semextproj}, we conclude $\sem{\binjudg {\join{\renvtwo \cup \renvfour}{(\renvthree\setminus\{\rvarone:\tyone\} )}}  {\ass\rvarone  \reone}}{} (\stone_{\renvone\to {\join{\renvtwo \cup \renvfour}{(\renvthree\setminus\{\rvarone:\tyone\} )}}})\ext \sem{\binjudg {\renvone}{\ass\rvarone  \reone}}{} (\stone)_{\renvone\to {\join{\renvtwo \cup \renvfour}{(\renvthree\setminus\{\rvarone:\tyone\} )}}}$. This shows (C3), and concludes the soundness proof for the first rule.
  \proofcase[$\SDASS$] The soundness proof of this rule is completely analogous to the one we have just shown, the main difference is that (C1) would require to show that $\srenv {\sttwo_1} \renvtwo \espl \reone \deone$ instead of $\srenv {\sttwo_1} \renvtwo \indp \reone \reone$. This requires the application of \Cref{lemma:detexprcompsem} instead of \Cref{lemma:exprcompsem}. The rest of the proof is identical. 
  \end{proofcases}
\end{proof}


\subsection{Structural Rules}

Structural rules are defined in Figure \ref{fig:structrules}
and are roughly inspired from those from~\cite{PSL}, at the same time being 
significantly simpler. 

The rule $\RWEAK$ is standard, while the other two rules are more interesting. 
$\RCONST$ is based on the ordinary rule
for constancy of Hoare logic: it states that is always possible to add 
any formula $\fthree$ as a conjunct of the precondition 
and the postcondition if this formula does not mention
variables which are modified by the program. The set of modified variables can be over approximated 
by the function  $\mv \cdot$ defined in Figure~\ref{fig:rvmvwvdef}.
The intuition behind the soundness of this rule is that
the projection of the initial state where $\tyf \fthree\renvthree$
is evaluated is not affected by the execution of $\prgone$.

The $\RFRAME$ rule is somehow similar to the $\RCONST$ rule,
but introduces a separating conjunction: by means of this rule, it is possible
to show that some operations
preserve computational independence.
The rule is sound because, by the definition of formulas,
none of the variables where $\tyf \fthree \renvthree$ is interpreted
appears in $\renvtwo$, because  $\join \renvtwo\renvthree$ is
implicitly proved to be well-defined. Moreover, from the first premise of the 
rule, we deduce that
$\binjudg \renvtwo \prgone$, so the variables in the domain of
$\renvthree$ are not influenced by the execution of 
$\prgone$. This is a sufficient condition to ensure that
computational indistinguishability is preserved during the
execution of $\prgone$.
\revision{
  Notice that the way in which the store is split
  is induced by the annotation in the formula.
  This choice results in a less general
  rule compared to that of \PSL~\cite{PSL} because,
  in particular, $\fone$ and $\ftwo$ are interpreted
  in the same environment $\renvtwo$, which means that this
  rule cannot be applied in all situations in which $\prgone$ modifies the
  statistical dependencies between the variables.
  This explains the need for the two
  non-derivable \emph{separating assignment rules} $\SRASS$ and $\SDASS$,
  which allow the movement of variables from one side of the
  separating conjunction to the other.
}
Even though more fine-grained formulations
of the $\RFRAME$ rule are valid,
we decided to adopt this one given the good trade-off between
simplicity and expressiveness it provides.

\begin{proposition}
\label{lemma:structsound}
The structural rules are sound.
\hfill\qed
\end{proposition}

\begin{proof}
  As we did for \Cref{lemma:compsound}, we do not show $\binjudg \renvone \prgone$; this means that we will only show the logical correctness of this system. To do so, we go by induction on the proof of the triple.
  \begin{proofcases}
  \proofcase[$\RWEAK$] In this case, we are asked to show that $\models \hot{\fone}{\binjudg \renvone \prgone}{\ftwo}$ knowing that:
    \begin{proofcases}
    \proofcase[H1] $\hotj{\tyf \fthree \renvone}{\binjudg \renvone \prgone}{\tyf \ffour\renvone}$,
    \proofcase[H2] $\tyf \fone \renvone\models \tyf  \fthree \renvone$,
    \proofcase[H3] $\tyf  \ffour \renvone \models \tyf \ftwo \renvone$.
    \end{proofcases}
    Applying the IH to H1, we obtain $\models \hot{\tyf \fthree \renvone}{\binjudg \renvone \prgone}{\tyf \ffour\renvone}$ (H4). The claim is:
    \[
      \forall \rsone \in \sem \renvone {}.\left(\sr{{\tyf \rsone \renvone}} \fone \right) \Rightarrow\left( \sr{\sem {\binjudg \renvone \prgone} {}({{\tyf \rsone \renvone}})}{\fthree}\right)
    \]
    Fix $\rsone$, from (H2) we obtain that ${{\tyf \rsone \renvone}}$  satisfies $\tyf \fthree\renvone$ as well. Applying (H4), we obtain $\left( \sr{\sem {\prgone} {}({{\tyf \rsone \renvone}})}{\tyf \ffour\renvone }\right)$, so we conclude with (H3).
    \proofcase[$\RCONST$] We want to show that $\models\hot{\tyf{\tyf \fone \renvtwo \land\tyf \fthree\renvthree}\renvone}{\binjudg \renvone \prgone}{\tyf{\tyf \ftwo \renvtwo \land\tyf \fthree \renvthree}\renvone}$ under the assumptions:
    \begin{proofcases}
    \proofcase[H1] ${\hotj{\tyf \fone \renvtwo}{\binjudg \renvtwo \prgone}{\tyf \ftwo \renvtwo}}$,
    \proofcase[H2] $\dom (\renvthree) \cap \mv \prgone = \emptyset$.
    \end{proofcases}
    Applying the IH to (H1), we conclude $\models \hot{\tyf \fone \renvtwo}{\binjudg \renvtwo \prgone}{\tyf \ftwo \renvtwo}$ (H3). The claim is:
    \begin{align*}
      \forall \rsone \in \sem \renvone {}.\left(\sr{{\tyf \rsone \renvone}} {\tyf{\tyf \fone \renvtwo \land\tyf \fthree\renvthree}\renvone} \right)  \Rightarrow\left( \sr{\sem {\binjudg \renvone \prgone} {}({{\tyf \rsone \renvone}})}{\tyf{\tyf \ftwo \renvtwo \land\tyf \fthree\renvthree}\renvone}\right)
    \end{align*}
    Fix $\rsone$ and assume $\sr{{\tyf \rsone \renvone}} {\tyf{\tyf \fone \renvtwo \land\tyf \fthree\renvthree}\renvone}$. From this assumption, and the definition of $\ext$, we deduce that:
    \[
      \sr{{\tyf {\rsone_{\renvone\to\renvtwo}} \renvtwo}}{\tyf \fone \renvtwo}
    \]
    and that
    \[
      \sr{{\tyf {\rsone_{\renvone\to\renvthree}} \renvthree}}{\tyf \fthree \renvthree}
      \tag{C1}
    \]
    simply by inlining the definition of $\sr \cdot \cdot$.
    From (H3), we conclude $\sr{\sem \prgone {} ({\tyf {\rsone_{\renvone\to\renvtwo}} \renvtwo})}{\tyf \fone \renvtwo}$ (C2). Observe that this object is defined because for validity of the premise (H3), we
    have $\binjudg\renvtwo \prgone$. \Cref{lemma:semextproj} shows that
    \[
      \sem \prgone {} (\tyf{\rsone_{\renvone\to\renvtwo}}\renvtwo) \ext \sem \prgone {} ({\tyf \rsone \renvone}).
      \tag{C3}
    \] 
    From this observation. Finally, with an application of Lemma \ref{lemma:mvproj}, we observe that:
    \[
      \sem {\prgone} {}({{\tyf \rsone \renvone}})_{\renvone \to \renvthree} = {{\tyf \rsone \renvone}_{\renvone \to \renvthree}}{}.
      \tag{C4}
    \]
    The conclusion is given by (C1-4). 
  \proofcase[$\RFRAME$] We are asked to show that $\models\hot{\tyf{\tyf \fone \renvtwo \sep\tyf \fthree\renvthree}\renvone}{\binjudg \renvone \prgone}{\tyf{\tyf \ftwo \renvtwo \sep\tyf \fthree \renvthree}\renvone}$ under the assumptions:
    \begin{proofcases}
    \proofcase[H1] ${\hotj{\tyf \fone \renvtwo}{\binjudg \renvtwo \prgone}{\tyf \ftwo \renvtwo}}$,
    \end{proofcases}
    Applying the IH to (H1), we conclude $\models \hot{\tyf \fone \renvtwo}{\binjudg \renvtwo \prgone}{\tyf \ftwo \renvtwo}$ (H2). The claim is:
    \begin{equation*}
      \forall \rsone \in \sem \renvone {}.\left(\sr{{\tyf \rsone \renvone}} {\tyf{\tyf \fone \renvtwo \sep\tyf \fthree\renvthree}\renvone} \right)  \Rightarrow\left( \sr{\sem {\binjudg \renvone \prgone} {}({{\tyf \rsone \renvone}})}{\tyf{\tyf \ftwo \renvtwo \sep\tyf \fthree\renvthree}\renvone}\right)
    \end{equation*}
    Fix $\rsone$ and assume $\sr{{\tyf \rsone \renvone}} {\tyf{\tyf \fone \renvtwo \sep\tyf \fthree\renvthree}\renvone}$. From this assumption, we deduce that there are:
    \[
      \sr{{\tyf \rstwo \renvtwo}}{\tyf \fone \renvtwo}
    \]
    and that
    \[
      \sr{{\tyf \rsthree \renvthree}}{\tyf \fthree \renvthree}
      \tag{C1}
    \]
    such that ${\tyf{\rstwo \tensprod \rsthree} {\join \renvtwo\renvthree}}\ind\ext{\tyf \rsone \renvone}$, which is equivalent to:
    $$
    {\tyf{\rstwo \tensprod \rsthree} {\join \renvtwo\renvthree}}\ind{\tyf \rsone \renvone}_{\renvone\to \join \renvtwo \renvthree}.
    $$
    Observe that for (H2) $\binjudg \renvtwo \prgone$. From these observations, we know that $\sem {\binjudg {\join \renvtwo \renvthree} \prgone}{}$ is defined. From \Cref{lemma:prgind}, we deduce that:
    $$
    \sem \prgone {}{({\tyf{\rstwo \tensprod \rsthree} {\join \renvtwo\renvthree}})}\ind\sem \prgone{}{({\tyf \rsone \renvone}_{\renvone\to \join \renvtwo \renvthree})}.
    $$
    From this result and \Cref{lemma:semextproj}, we conclude that
    $$
    \sem \prgone {}{({\tyf{\rstwo \tensprod \rsthree} {\join \renvtwo\renvthree}})}\ind\sem \prgone{}{({\tyf \rsone \renvone})_{\renvone\to \join \renvtwo \renvthree}}, 
    $$
    which means
    $$\sem \prgone {}{({\tyf{\rstwo \tensprod \rsthree} {\join \renvtwo\renvthree}})}\ind \ext\sem \prgone{}{({\tyf \rsone \renvone})}.
    $$
    By applying \Cref{lemma:sempartcomm}, we observe that $\sem \prgone {}{({\tyf{\rstwo \tensprod \rsthree} {\join \renvtwo\renvthree}})}=\sem \prgone {}{({\tyf \rstwo \renvtwo})}\comp {\tyf \rsthree \renvthree}$, so:
    \[
    \sem \prgone {}{({\tyf \rstwo \renvtwo})}\comp {\tyf \rsthree \renvthree}\ind \ext\sem \prgone{}{({\tyf \rsone \renvone})}.
    \tag{C2}
    \]
    With (C1-2), and by applying (H2) to ${\tyf \rstwo \renvtwo}$, we conclude the proof.
    \proofcase[$\RRESTR$] We want to show that $\models\hot{\tyf{ \fone}\renvone}{\binjudg \renvone \prgone}{\tyf{ \ftwo}\renvone}$ under the assumptions:
    \begin{proofcases}
      \proofcase[H1] ${\hotj{\tyf \fone \renvtwo}{\binjudg \renvtwo \prgone}{\tyf \ftwo \renvtwo}}$,
      \proofcase[H2] $\tyf \fone \renvtwo \ext \tyf \fone  \renvone$.
      \proofcase[H3] $\tyf \ftwo \renvtwo \ext \tyf \ftwo  \renvone$.
    \end{proofcases}
    For IH and (H1), we know that $\models {\hot{\tyf \fone \renvtwo}{\binjudg \renvtwo \prgone}{\tyf \ftwo \renvtwo}}$ holds. Les $\stone$ be a state such that:
    \[
      \srenv  \renvone \stone{\tyf  \fone \renvone}.
    \]
    From (H2) or (H3) and \Cref{rem:extftoexts} we can deduce $\renvtwo \ext \renvone$.
    This means that $\stone_{\renvone \to \renvtwo}$ is defined, and for
    \Cref{lemma:prekm} and (H2), we have:
    \[
      \srenv  \renvtwo {\stone_{\renvone\to \renvtwo}}{\tyf  \fone \renvtwo}.
    \]
    From the IH, we conclude $\srenv \renvtwo {\sem{\binjudg \renvtwo \prgone}{}(\stone_{\renvone \to \renvtwo})} {\tyf \ftwo \renvtwo}$. From \Cref{lemma:semextproj},
    we conclude
    $$
    \srenv \renvtwo {{\sem{\binjudg \renvone \prgone}{}(\stone)}_{\renvone \to \renvtwo}} {\tyf \ftwo \renvtwo}.
    $$
    We show the claim with another application of \Cref{lemma:prekm}, which is justified by (H3).
  \end{proofcases}    
\end{proof}

\section{Some Examples}
\label{sec:examples}

This section is devoted to giving two simple well-known cryptographic 
constructions which can be proved to have the desired properties by way of 
\CPSL.

\subsection{The Pseudo One-Time Pad}
\label{sec:potp}

\begin{figure*}[t]
\revision{  \begin{center}
    \resizebox{\textwidth}{!}{\(
    \infer[\RSEQ]{\hotj{\top}{\binjudg\renvone{\seq{\ass k \rand()}  {\ass{c}{\exor(m,g(k))}}}}{\tyf{\df m}{\renvone_m} \sep \tyf{\ud c}{\renvone_c }}}
    {
      \juone_\POTP
      &
      \infer[\RWEAK]{\hotj{{
            {\ud {k}}
            \sep \tyf{\df m}{\renvone_{m}}}}{\binjudg\renvone{\ass c {\exor(m, g(k))}}}{\tyf{\df m}{\renvone_m} \sep
          {\ud c}
        }}
      {
        \infer[\RWEAK]{\hotj{{
              {\ud {g(k)}}
              \sep \tyf{\df m}{\renvone_{m}}}}{\binjudg\renvone{\ass c {\exor(m, g(k))}}}{\tyf{\df m}{\renvone_m} \sep
            {\ud c}
          }}
        {
          \infer[\RCONST]{\hotj{\tyf \ftrue{\renvone} \land \tyf{
                {\ud k}
                \sep \tyf{\df m}{\renvone_{m}}}{\renvone_{k, m}}}{\binjudg\renvone{\ass c {\exor(m, g(k))}}} {\tyf{\espl c {\exor(m, g(k))}}{{\renvone}} \land\tyf
              {
                {\ud {g(k)}}
                \sep \tyf{\df m}{\renvone_{m}}}{\renvone_{k, m}}}}
          {
            \infer[\DASS]{\hotj{\ftrue}{\binjudg{{\renvone}}{\ass c {\exor(m, g(k))}}}{\espl c {\exor(m, g(k))}}}
            {
              c \notin \fv {\exor(m, g(k))}
            }
            &
            \{m,  k\}\cap \mv {\ass c \exor(m, g(k))}=\emptyset
          }
        &
        \text{\Cref{lemma:prekm}, Axiom \eqref{ax:auxpotp2}}
        }
        &
        \eqref{eq:axpotp}
      }
    }
    \)}    
  \end{center}
Where $\juone_\POTP$ is:
\begin{center}
        \resizebox{0.6\textwidth}{!}{\(
      \infer[\RWEAK]{\hotj{\top}{\binjudg\renvone{\ass k \rand()}}
        {
          {\ud k}
          \sep \tyf{\df m}{\renvone_{m}}}}
      {
        \infer[\SRASS]{\hotj{\tyf{\top}{\{\}} \sep \tyf{\df m}{\renvone}}{\binjudg\renvone{\ass k {\rand()}}}
          {{
              \tyf
              {\top \land\eq k {\rand()}}
              {\renvone_{k}}
              \sep \tyf{\top}{\renvone_{m, c}}}}
        }
        {
          \binjudg{\varepsilon}{\rand:\stng \nat} &
          \rvarone \notin \fv \rand &
          \tyf \top {\renvone_{c, m}} \in \fset &
          \rvarone \notin \dom(\varepsilon)
        }
        &
        \text{Axiom \eqref{ax:auxpotp1}}
      }
      \)}
  \end{center}  
}  \caption{Proof tree of $\hotj{\fone_\POTP}{\POTP}{\ftwo_\POTP}$.}
  \label{fig:potpmain}
\end{figure*}

\newcommand{\Gen}{\mathsf{Gen}}
\newcommand{\Enc}{\mathsf{Enc}}
\newcommand{\Dec}{\mathsf{Dec}}

The pseudo-one-time pad ($\POTP$) is a symmetric key encryption
scheme that enjoys \emph{computational secrecy}~\cite{BlumMicali,KatzLindell}.
It differs from the classic one-time pad because exclusive or is not applied to the key
itself but rather to the output of a
\emph{pseudorandom generator} when applied to the actual key. This solves the 
main drawback of \emph{perfectly secret} encryption schemes: with $\POTP$, the 
cardinality of the key space can be strictly smaller than that of the message space. 

%

It is possible to use our logic to show that the $\POTP$
scheme enjoys computational secrecy.
Analogously to what we did in \Cref{sec:viewpoint} we express the encryption 
procedure in our language as the following program, that we call
$\POTP$:
\begin{align*}
  &\ass{k}{\rand()}\semic\\
  &\ass{c}{\exor(m,g(k))}
\end{align*}
After having generated the key $k$ by way of $\rand()$, the latter is passed to 
a pseudorandom generator $g(\cdot)$, and only at that point the latter is xored 
with the message $m$. In the proof, the fact that $g$ is a pseudorandom 
generator is captured by the following axiom:
\begin{align*}
  \label{eq:axpotp}
  \tyf{\ud{x}}\renvone\models\tyf{\ud{g(x)}}\renvone &&&  \text{if }\renvone(x)=\stng \nat 
  \tag{$\text{Ax}_\POTP$}
\end{align*}
Although apparently stronger than the textbook definition for 
pseudorandom generators where $\ud{x}$ is replaced by
$\eq x {\rand()}$, $\text{Ax}_\POTP$ is actually equivalent to it.\footnote{\revision{
  This happens because if $\text{Ax}_\POTP$ did not follow from
  the standard version of the axiom, then there would be a store
  $\tyf \rsone\renvone \in \sem{\tyf{\ud x}\renvone}{}\setminus
  \sem{\tyf{\eq x{\rand()}}\renvone}{}$
  such that $\sem {\tyf{g(x)}\renvone}{}(\rsone)$ is \emph{not} pseudorandom.
  In particular, there would be an adversary $\distone$ with
  non-negligible advantage in distinguishing
  $g(x)$ from the uniform distribution.
  But this would contradict $\tyf \rsone\renvone \in \sem{\tyf{\ud x}\renvone}{}$, because
  the distinguisher $\lambda y, z.\distone(y, \parsem{g(z)}{})$ would have
  non-negligible advantage on $\sem {x}{}(\rsone)$ and $\unif{\stng \nat}$.
}
}

The next step is to find a way to capture the desired property as a Hoare triple.
The program we want to prove secure is certainly $\POTP$, but how about pre- 
and post-conditions? Traditionally, computational secrecy is captured by a form 
of indistinguishability, which requires the introduction of an experiment. As 
suggested by Yao~\cite{Yao}, however, we can stay closer to Shannon's 
perfect security, and rather adopt a weaker notion of independence. More 
specifically, we pick a post-condition stating that the message and the 
ciphertext are computationally independent. Formally,
we can take  $\fone_\POTP$ as precondition and $\ftwo_\POTP$ as postcondition, 
as follows:
%
\begin{align*}
  \fone_\POTP&=\tyf{\top}{\renvone}\\
  \ftwo_\POTP&=\tyf{\tyf{\df{m}}{m: \stng{\polyone(\nat)}}\sep{\ud{c}}}{\renvone},
\end{align*}
\normalsize
with $\renvone ={{k}:\stng\nat,{m},{c},
  \rvarone: \stng{\polyone(\nat)}}$.
Observe that
 $\ftwo_\POTP$ requires $\ud{c}$, which means that
$c$ is not only computational independent from $m$, but also
indistinguishable from the uniform distribution, while in principle it
would be sufficient to show a weaker statement instead,
namely: $\tyf{\tyf{\df{m}}{m: \stng{\polyone(\nat)}}\sep\tyf{\df{c}}{c: 
\stng{\polyone(\nat)}}}{{k}:\stng{\nat},{m},{c}: \stng{\polyone(\nat)}}$.
Proving that this is equivalent to indistinguishability for our notion of 
computational independence is not trivial, and~\Cref{sec:compindep} is devoted 
to showing it.

Proving the validity
of the triple 
$
\models\hot{\fone_\POTP}{\POTP}{\ftwo_\POTP}
$
by hand is cumbersome and error-prone. Conversely, it is possible
to give the proof in a more structured and concise way by employing
the inference rules we defined in \Cref{sec:inference}. It is remarkable that the 
proof above is not based on any form of reduction, being solely based on a form 
of program analysis, which proceeds compositionally without even
mentioning the adversary.
The complete proof tree is in \Cref{fig:potpmain}; it starts with
an application of the $\RSEQ$ rule to synchronize the
triple obtained for the first command (in the subtree $\pi_\POTP$)
to the one obtained for the second one. In $\pi_\POTP$, the rule
$\SRASS$ plays the important role of moving $k$ from the environment
to the right of the separating conjunction to that on the left.
This way, the post-condition states that $k$ is uniform
and that it is computationally independent from all the other
variables in the environment. Then, by means of the $\RWEAK$ rule,
we rewrite the pre- and the post-conditions: on the left-hand side,
we are rewriting a tautology to another tautology,
while the rewriting on the right is valid of the following axiom:
\begin{equation}
  \label{ax:auxpotp1}
  \srenv \renvone
  {
    \tyf
    {\top \land\eq k {\rand()}}
    {\renvone_{k}}
    \sep \tyf{\top}{\renvone_{m, c}}}
  {\ud k}
  \sep \tyf{\df m}{\renvone_{m}}    
\end{equation}
For the second command, we start by applying the $\DASS$ rule,
which shows that in every sample in the support of the
final state, the value of $c$ is identical to that of $\exor(m, g(k))$.
Since the variables $m$ and $k$ are not modified by the second assignment,
we can apply the $\RCONST$ rule to add to both the sides
any formula $\tyf \fone{\renvone_{m, k}}$, where
$\renvone_{m, k}$ denotes the restriction to the domain $\{m, k\}$.
In particular, we add $\tyf{
  {\ud k}
  \sep \tyf{\df m}{\renvone_{m}}}{\renvone_{k, m}}
$, which is very close to post-condition of the triple
we showed in $\pi_\POTP$. We chose this formula, because one of our goals
is to obtain the post-condition of $\pi_\POTP$ as a pre-condition
in order to apply the $\RSEQ$ rule.
With the next application of the $\RWEAK$ rule, we rewrite the
pre-condition with a stronger formula --- the modification to the environment of
$\tyf{
  {\ud k}
  \sep \tyf{\df m}{\renvone_{m}}}{\renvone_{k, m}}
$ is justified by \Cref{lemma:prekm}. On the post-condition, we apply
the following axiom schema:
\begin{equation}
  \label{ax:auxpotp2}
  \begin{gathered}
    {\tyf{\tyf{\espl c {\exor(m, \deone)}}{{\join \renvtwo{\renvone_{c, m}}}} \land\tyf{\tyf{\ud \deone}{\renvtwo} \sep \tyf{\df m}{\renvone_{m}}}{\renvone}}\renvone}
    \\
    \hfill\models^\renvone  {\tyf{\df m}{\renvone_m} \sep \tyf{\ud c}{\renvone_c}},
  \end{gathered}
\end{equation}
that allows to conclude that $c$ is presudorandom
under the assumption that the expression $\deone$ is
also pseudorandom.
In the very last application of the $\RWEAK$ rule, we apply the
\eqref{eq:axpotp} to the pre-condition in order to deduce $\ud{g(k)}$ from
$\ud{k}$; this way our final pre-condition matches exactly the post condition of
$\pi_\POTP$, and we can conclude the proof by applying the $\RSEQ$ rule.
The validity of axioms \eqref{ax:auxpotp1} and \eqref{ax:auxpotp2} is shown
in \Cref{lemma:ax2prg,rem:potp1} below.

\begin{lemma}
  \label{lemma:ax2prg}
  For every, $\renvone$, such that $\renvone(\{m, c\})=\{\stng{\polyone(n)}\}$,
  $\renvone(k)=\{\stng{n}\}$, every $\renvtwo$ such that
  $\join \renvtwo {\renvone_{m, c}}\defined$ and 
  $\join \renvtwo {\renvone_{m, c}} \ext \renvone$,
  and every $\stone \in \sem \renvone{}$, if $\stone$ satisfies
  \small
  \[
    \tyf{\tyf{\espl c {\exor(m, \deone)}}{{\join \renvtwo{\renvone_{c, m}}}} \land\tyf{\tyf{\ud \deone}{\renvtwo} \sep \tyf{\df m}{\renvone_{m}}}{\renvone}}\renvone
  \]
  \normalsize
  then it satisfies:
  $$
  \tyf  {\tyf{\df m}{\renvone_m} \sep \tyf{\ud c}{\renvone_c}}\renvone.
  $$
\end{lemma}
\begin{proof}
  From the assumption, we deduce:
  \begin{proofcases}
  \proofcase[H1] $\sr{\stone_{\renvone\to {\join \renvtwo {\renvone_{c, m}}}}}{\espl c {\exor(m,\deone)}}$
  \proofcase[{H2}] $\sr{\stone}{\tyf{\tyf{\ud \deone}{\renvtwo} \sep \tyf{\df m}{\renvone_{m}}}{\renvone}}$
  \end{proofcases}
  The assumption (H2) allows us to deduce
  \begin{proofcases}
  \proofcase[H2A] $\sr{\stthree_1}{\tyf{\ud \deone}{\renvtwo}}$
  \proofcase[H2B] $\sr{\stthree_2}{\tyf{\df m}{\renvone_{m}}}$
  \proofcase[H2C] ${\stthree_1\comp {\stthree_2}}\ind\ext{\stone}$
  \end{proofcases}
  The goal is to find $\stfour_1, \stfour_2$ such that:
  \begin{proofcases}
  \proofcase[C1] $\sr{\stfour_1} \tyf{\df m}{\renvone_m}$
  \proofcase[C2] $\sr{\stfour_2} \tyf{\ud c}{\renvone_c}$
  \proofcase[C3] $\stfour_1\comp \stfour_2 \ind\stone_{\renvone\to\renvone_{m,c}}$.
  \end{proofcases}
  Where, for stating (C3), we have expanded the definition of $\ind\ext$.
  We set $\stfour_1= \stthree_2$ and $\stfour_2= \stone_{\renvone\to\renvone_c}$.
  With these choices (C1) is trivial. We go for (C2).
  This claim requires us to show that:
  \[
    \sem c{}(\stone_{\renvone\to\renvone_c})\ind \unif{\stng {\polyone(\nat)}}
  \]
  From (H1) and (W2), we know that 
  \[
    \sem c{}(\stone_{\renvone\to\join \renvtwo {\renvone_{m, c}}})=\sem {\exor (m, \deone)}{}(\stone_{\renvone\to\join \renvtwo {\renvone_{m, c}}})
  \]
  From Lemma \ref{lemma:projexpsem}, we deduce that:
  \[
    \sem c{}(\stone_{\renvone\to\renvone_c})=\sem c{}(\stone_{\renvone\to\join \renvtwo {\renvone_{m, c}}}),
  \]
  so, for this observation, we can reduce (C2) to:
  \[
    \sem {\exor (m, \deone)}{}(\stone_{\renvone\to\join \renvtwo {\renvone_{m,c}}})\ind\unif{\stng {\polyone(\nat)}},
  \]
  And by applying Lemma \ref{lemma:projexpsem} once again, to:
  \[
    \sem {\exor (m, \deone)}{}(\stone_{\renvone\to\join \renvtwo {\renvone_{m}}})\ind\unif{\stng {\polyone(\nat)}}.
  \]
  Moreover, from (H2C), we observe that $\stone_{\renvone\to\join \renvtwo {\renvone_{m}}}\ind \stthree_1\comp\stthree_2$ so, with an application of Lemma \ref{lemma:exprind},
  (C2) can be restated as follows:
  \[
    \sem {\exor (m, \deone)}{}(\stthree_1\comp\stthree_2)\ind\unif{\stng{\polyone(\nat)}}.
  \]
  Now, we assume without lack of generality that the function symbols also contain a
  primitive $\rand(\tyone)$ for every type $\tyone$ that generates a
  binary string of length $\polytwo(\nat)$ uniformly at random whenever
  $\tyone=\stng{\polytwo(\nat)}$. Observe that: 
  \[
    \sem {\exor (m, \rand(\polyone(\nat)))}{}(\stthree_1\comp\stthree_2)=\unif{\stng{\polyone(\nat)}}.
  \]
  This can be shown with the same argument used to
  show the correctness of the OTP encryption scheme.
  For this reason, (C2) can be reduced to:
  \[
    \sem {\exor (m, \deone)}{}(\stthree_1\comp\stthree_2)\ind\sem {\exor (m, \rand(\polyone(\nat)))}{}(\stthree_1\comp\stthree_2).
  \]
  assume that this claim does not hold. This means that there is a distinguisher
  $\distone$ with non-negligible advantage for these two distribution ensembles.
  Also observe that $\stthree_1\comp\stthree_2$ is efficiently samplable,
  that $g$ is poly-time and that also $\rand(\polyone(\nat))$ is efficiently samplable.
  It is possible to define a distinguisher $\disttwo$ for $\sem{\deone}{}(\stthree_1)$
  and the family of
  uniform distributions $\{\unif\polyone(\nat)\}_{\nat\in \NN}$.
  This distinguisher receives in input a sample
  string $x$, it samples $\stthree_1 \comp \stthree_2$ to obtain a map $\memtwo$,
  it computes $\memtwo(m)$, and returns the result of the $\exor$ between this value and $x$
  to $\distone$. Observe that if the input is sampled according to $\sem{\deone}{}(\stthree_1)$, the distribution
  of the inputs for $\distone$ will be that of $\sem {\exor (m, \deone)}{}(\stthree_1\comp\stthree_2)$,
  conversely, if it is sampled according to the uniform distribution, it
  will be that of $\sem {\exor (m, \rand(\polyone(\nat)))}{}(\stthree_1\comp\stthree_2)$.
  For this reason the advantage of $\disttwo$ is the same of $\distone$, in particular,
  it is non-negligible. This is absurd because we are assuming that
  $\stthree_1\models \ud \deone$ (H2A).
  It remains just to show (C3). This claim can be rewritten
  by inlining
  the definitions of $\stfour_1$ and $\stfour_2$ for more
  clarity:
  \[
    \stthree_2\comp\stone_{\renvone\to\renvone_c} \ind \stone_{\renvone\to\renvone_{m, c}}.
  \]
  For compatibility of $\ind$ and $\comp$ (\Cref{lemma:indcompcomp}), we show:
  \[
    \stthree_2\comp\stfour \ind \stone_{\renvone\to\renvone_{m, c}}.
  \]
  where $\stfour$ is the distribution defined on $c$ as the family of uniform distribution
  of strings of length $\polyone(\nat)$.
  Assume that this does not hold. This means that there is a distinguisher $\distone$
  with non-negligible advantage for these two distributions.
  Out of it, it is possible to define another distinguisher $\disttwo_1$ for
  the following efficiently samplable distribution ensemble on pairs:
  \[
    \bind{\stthree_2 \comp \stfour}{\memtwo \mapsto \unit{(\memtwo(m), \exor(\memtwo(m), \memtwo(c)))}}
    \tag{D1}
  \]
  and
  \[
    \bind{\stone_{m, c}}{\memtwo \mapsto \unit{(\memtwo(m), \exor(\memtwo(m), \memtwo(c)))}}.
    \tag{D2}
  \]
  Indeed, $\disttwo_1$ can simply take its input sample $x$,
  compute the bitwise exclusive or on the first component ($x_1$)
  and the second component of $x$ ($x_2$) and call that value $y$,
  then, $\disttwo_1$ calls $\distone$ on the sample where $m$ is associated to
  $x_1$ and $c$ is associated to $y$. It is easy to see that the
  advantage of $\distone$ is the same of $\disttwo_1$.
  In particular, observe that (D1) is equal to:
  $$
  \bind{\stthree_2}{\memthree \mapsto \bind {\stfour} {\memtwo\mapsto\unit{(\memthree(m), \exor(\memthree(m), \memtwo(c)))}}},
  $$
  which, in turn, is equal to: 
  $$
  \bind{\stthree_2}{\memthree \mapsto \bind {\stfour} {\memtwo\mapsto\unit{(\memthree(m), \memtwo(c))}}}.
  $$
  Because $l(c)$ is statistically uniform.
  From (H1), we deduce that (D2) is equal to: 
  $$
  \bind{\stone}{\memtwo \mapsto \unit{(\memtwo(m), \parsem \deone {} (\memtwo))}},
  $$
  and thus to:
  $$
  \bind{\stone_{\renvone\to \join \renvtwo {\renvone_{m}}}}{\memtwo \mapsto \unit{(\memtwo(m), \parsem \deone {} (\memtwo))}},
  $$
  for \Cref{rem:detsemmemupdate}. Due to (H2C), we know that $\disttwo_1$ has negligible advantage also
  in distinguishing (D2), from:
  \[
    \bind{\stthree_1\comp\stthree_2}{\memtwo \mapsto \unit{(\memtwo(m), \parsem \deone {} (\memtwo))}},
    \tag{D3}
  \]
  This is shown by contraposition: any distinguisher $\overline \distone$
  that succeeds in distinguishing (D2) and (D3) with
  non-negligible advantage can be used to define a distinguisher $\overline \disttwo$
  for $\stone_{\renvone\to \join \renvtwo {\renvone_{m}}}$ and,
  $\stthree_1\comp\stthree_2$. The distinguisher $\overline \disttwo$
  simply computes $(\memtwo(m), \parsem \deone \nat (\memtwo))$
  on its input sample and passes it to $\overline \distone$.
  The distinguishers $\overline \distone$ and $\overline \disttwo$
  have the same advantage, and this contradicts (H2C).
  This shows that (D3) and (D2) are computationally indistinguishable.
  So, for the triangular inequality, and by observing that the sum of a negligible
  function with another negligible function is itself a negligible function,
  it must be the case where the advantage of $\disttwo_1$ in distinguishing (D1)
  and (D3) is non-negligible.

  Now, we observe that (D1) is equal to: 
  $$
  \bind{\stthree_2}{\memthree \mapsto \bind {\stfour} {\memtwo\mapsto\unit{(\memthree(m), \memtwo(c))}}},
  $$
  and to:
  $$
  \bind{\stthree_1\comp \stthree_2}{\memthree \mapsto \bind {\stfour} {\memtwo\mapsto\unit{(\memthree(m), \memtwo(c))}}}.
  $$
  This is a contradiction: it is possible to define a new adversary $\disttwo_2$ out of
  $\disttwo_1$ with non-negligible advantage on $\sem{\deone}{}(\stthree_1)$ and
  the uniform distribution of binary strings with length $\polyone(\nat)$.
  This adversary receives in input its sample $x$, then it samples from
  $\stthree_1 \comp \stthree_2$ a sample $\memtwo$ and then gives the
  pair $(\memtwo(m), \exor ({\memtwo(m)}, x))$ in input to $\disttwo_1$.
  Observe that, if the sample is
  taken from $\sem{\deone}{}(\stthree_1)$, then $\disttwo_2$
  gives the same output $\disttwo_1$ gave for (D3), while if
  the input sample $x$ was distributed uniformly, $\disttwo_2$
  gives the same output that $\disttwo_1$ gave for (D1).
  The advantage of $\disttwo_2$ is non-negligible,
  but this contradicts (H2A).
\end{proof}

\begin{rem}
  \label{rem:potp1}
  $\tyf{\tyf{\eq k {\rand()}} {\renvone_k}\sep \tyf \ftrue {\renvone_{m,c}}}\renvone\models\tyf{\tyf{\ud k}{\renvone_k} \sep \tyf \ftrue {\renvone_{m}}}\renvone$
\end{rem}
\begin{proof}
  Assume that a state $\sttwo$ satisfies the formula on the left. There are $\sttwo, \stthree$ such that $\sttwo \comp \stthree\defined$ and in particular $\sr \sttwo {\eq k {\rand()}}$, so the claim follows trivially from (W1) and (U1). Observe that the second sub-formula of consequence is not identical to that of the premise, but this is sound because of \Cref{lemma:prekm}.  
\end{proof}

\subsection{Bitwise Exclusive Or}
\label{sec:appxor}

\begin{figure*}[t]
  \centering
  \resizebox{\textwidth}{!}{
    $\begin{aligned}
    \infer[\RCONDCM]{\hotj {\tyf \top \renvone} {\binjudg \renvone {\ifr{k = 1}{\ass c \lnot m}{\ass c m}}} {\tyf{\espl c {\exor(k,m)}}\renvone}}
    {
      \infer[\RWEAK]{\hotj {\tyf{\espl k 1}\renvone} {\binjudg \renvone{\ass c {\lnot k}}} {\tyf{\espl c {\exor(k,m)}}\renvone}}
      {
        \infer[\RCONST]{\hotj{\tyf{\tyf \top \renvone \land \tyf{\espl k 1}{\renvone_k}}\renvone} {\binjudg\renvone{\ass c {\lnot k}}} {\tyf{\tyf{\espl c {\lnot m}}\renvone\land \tyf{\espl k 1}{\renvone_k}}\renvone}}
        {
          \infer[\DASS]{\hotj{\tyf \top \renvone} {\binjudg \renvone{\ass c {\lnot k}}} {\tyf{\espl c {\lnot m}}\renvone}}
          {
            c \notin \fv {\lnot k}
          }&
          \{k\} \cap \mv {\ass c {\lnot k}}=\emptyset
        }&
        \text{Remarks \ref{rem:pi1} and \ref{rem:pi2}}
      }&
      \jutwo &
      \espl c {\exor(k,m)} \in \exfset
    }
  \end{aligned}
  $}
\caption{Example of employment of the $\RCONDCM$ rule for the $\XOR$ program.} 
\label{fig:xor}
\end{figure*}

In this section we show how our program logic can be employed to show the
correctness of the following program computing the exclusive or of two bits.
\[
  \XOR \defsym \ifr{k = 1}{\ass c \lnot m}{\ass c m}.
\]
We do so by showing the correctness of the triple:
\[
  \hotj{\tyf\ftrue\renvone}{\binjudg\renvone\XOR}{\tyf{\eq c {\exor(k, m)}}\renvone},
\]
where $\renvone=c, k, m:\bool$.
The full tree is in \Cref{fig:xor}. We omit the subtree $\jutwo$ because it is analogous to the one on its left. Actually, the claim we are showing is stronger that the one in the triple, but this last one can be obtained by an application of the $\RWEAK$ rule with the (W2) axiom.  

\begin{rem}
  \label{rem:pi1}
  $\tyf{\espl k 1}\renvone\models\tyf{\tyf \top \renvone \land \tyf{\espl k 1}{\renvone_k}}\renvone$
\end{rem}
\begin{proof}
  With \Cref{lemma:parsemproj}, we show the validity of $\tyf{\eq k 1}{\renvone_k}$; that of $\tyf \top \renvone$ is trivial. 
\end{proof}

\begin{rem}
  \label{rem:pi2}
  $\tyf{\tyf{\eq c {\lnot m}}\renvone\land \tyf{\eq k 1}{\renvone_k}}\renvone\models\tyf{\eq c {\exor(k,m)}}\renvone$
\end{rem}
\begin{proof}
  Observe that:
  \[
    \srenv\renvone{}{\tyf{\eq {\lnot m} {\exor (m, 1)}}\renvone}
  \]
  The claim follows from two applications of (T2).
\end{proof}

\subsection{Pseudorandom Key Stretching }
\label{subsec:stretching}

\newcommand{\EXP}[1][]{\mathtt{EXP}_{#1}}

In this section we employ our program logic to show that
if pseudorandom generators with expansion factor $\nat+1$ exist,
then pseudo-random generators with expansion factor $n+h+1\ge n+1$  exist as well 
for any constant $h\in \NN$.
Even though this result by itself is not particularly insightful and is 
superseded by classic results about pseudorandom generators~\cite{KatzLindell},
we believe it is anyway interesting, given the style of the resulting proof.

To this end, we define the family of programs $\{\EXP[i]\}_{i\in \NN}$,
where each program $\EXP[h]$ with $h\in \NN$
is a key stretching algorithm with expansion factor $\nat+h+1$.
The construction of these algorithms is based on that of the canonical
pseudo-random generator with expansion factor $\nat+\polyone(\nat)$
starting from a pseudo-random generator with expansion factor
$\nat+1$, see~\cite{KatzLindell}.
We define the pseudorandom generator $\{\EXP[i]\}_{i\in \NN}$ as follows:
\begin{align*}
  \EXP[h]\defsym& \seq{\seq{\prgone_{0}}\ldots}{\prgone_{h}}\semic
                    \ass {s_0} {k}\semic
                    \seq{\seq{\prgtwo_0}\ldots}{\prgtwo_{h}}
\end{align*}
where, for every $i \in \NN$, we have:
\begin{align*}
  \prgone_i\defsym& \seq{\seq{\ass {\rvarone_i} {g(k)}}{\ass{b_i}{\head ({\rvarone_i})}}}{\ass k {\tail ({\rvarone_i})}}\semic\\
  \prgtwo_i\defsym& \ass {s_{i+1}} {\concat_{\nat+i, 1}({s_i},{b_i})}
\end{align*}

The claim we are interested in showing states that, if the
distribution of $k$ in the initial store
is pseudo-random, then the distribution of $s_{h+1}$ in the final
store is also pseudo-random.
This can be expressed by the following triple:
\[
  \hot {\tyf{\ud k}{k:\stng \nat}\sep\top^{\renvtwo_0}} {\binjudg {\renvone_h} {\EXP[h]}}{\ud {s_{h+1}}},
\]
where, for every $h \in \nat$, the environment $\renvone_h$ is defined
the environment $\renvone$ is defined as follows:
\begin{gather*}
  \begin{aligned}
    \renvone_h(k)&= \stng\nat &\forall 0\le i \le h. \renvone_h(\rvarone_i)&=\stng{\nat+1}\\
    \forall 0\le i \le h. \renvone_h (b_i) &= \bool &  \renvone_h({s_{i}})(k)&= \stng{\nat +i}.
  \end{aligned}\\
  \renvone_h(\rvartwo_{h+1}) = \stng{\nat+h+1}
\end{gather*}
and for every $i \in \NN$, the environment $\renvtwo_i$
is equal to $\renvone_h$, but all the $\rvarone_j$ and $b_j$
for $0\le j <i$ are undefined. 

%

The proof relies on the axiom \eqref{eq:axpotp} and on two other
auxiliary axioms \eqref{eq:axsplit} and \eqref{eq:axmerge}. The
former states that if in a store $\stone$ the variable $\rvarone$
is distributed uniformly at random, the variable $b$ stores the first bit
of $\rvarone$, and the variable $\rvartwo$ stores the remaining part of the
string, then the marginal distributions of $b$ and $\rvartwo$
are pseudorandom and independent.
The axiom \eqref{eq:axmerge} somehow states the converse:
by concatenating a bit and a string that are both pseudorandom and
conditionally independent, we obtain a string that is pseudorandom.
These axioms can be given in our language as follows: 
\begin{equation}
  \label{eq:axsplit}
  \begin{multlined}
    {\tyf{\tyf{\ud \rvarone}\renvtwo\land \tyf{\espl b {\head (\rvarone)}}\renvtwo\land \tyf{\espl \rvartwo {\tail (\rvarone)}}\renvtwo}\renvtwo}
    \models^\renvtwo {\tyf{\ud b \sep \ud \rvartwo}\renvtwo}    
  \end{multlined}
  \tag{$\text{Ax}_{\mathtt{SPL}}$}
\end{equation}
\begin{equation}
  \label{eq:axmerge}
  \srenv \renvtwo{(\ud \rvarone \sep \ud b) \land \espl \rvartwo{\concat (\rvarone, b)}}
  {\ud \rvartwo}
  \tag{$\text{Ax}_{\mathtt{MRG}}$}
\end{equation}
\noindent
Where $\renvtwo(b) = \bool,  \renvtwo(s) = \stng{p(\nat)}, \renvtwo(r) = \stng{p(\nat)+1}$.
\noindent
\begin{lemma}
  \label{lemma:axsplitmerge}
  The axioms \eqref{eq:axsplit} and \eqref{eq:axmerge} are sound.
    \hfill\qed
\end{lemma}
\begin{proof}
  We start by showing the \eqref{eq:axsplit} axiom:
  \small
  \[
    \srenv\renvtwo{\tyf{\ud \rvarone\land \espl b {\head (\rvarone)}\land \espl \rvartwo {\tail (\rvarone)}}\renvtwo}
    {\tyf{\ud b \sep \ud \rvartwo}\renvtwo}    
  \]
  \normalsize
  The goal is to show:
  \begin{proofcases}
    \proofcase[C1] $\srenv {\renvtwo_b} {\stthree_b}  {\ud b}$
    \proofcase[C2] $\srenv {\renvtwo_\rvartwo} {\stthree_\rvartwo}  {\ud \rvartwo}$
    \proofcase[C3] $\stthree_b \comp \stthree_\rvartwo \ind \stthree_{b, \rvartwo}$
  \end{proofcases}
  Under the assumption:
  \[
    \srenv \renvtwo \stone {\ud \rvarone \land \espl b {\head(r)} \land \espl s {\tail(r)}},
  \]
  i.e.:
  \begin{proofcases}
    \proofcase[H1] $\srenv {\renvtwo} {\stthree}  {\ud \rvarone}$
    \proofcase[H2] $\srenv {\renvtwo} {\stthree}  {\espl b {\head (\rvarone)}}$
    \proofcase[H3] $\srenv {\renvtwo} {\stthree}  {\espl \rvartwo {\tail (\rvarone)}}$
  \end{proofcases}  
  We start by showing (C1).
  From (W1) and (W2), we deduce that $\srenv \renvtwo \stthree \indp b  {\head(r)}$.
  We start the following intermediate result:
  \[
    \srenv \renvtwo \stthree{\ud{\head(r)}}.
  \]
  Assume this
  to be false. This means that there is an adversary $\distone$ that 
  distinguishes samples of the distribution ensemble $\sem{head(r)}{}(\stthree)$
  from the samples of the distribution ensemble $\unif\bool$. Out of this distinguisher,
  it is possible to define a new distinguisher $\disttwo$ for the distribution ensembles
  $\sem \rvarone{}(\stthree_\rvarone)$ and $\unif{\stng\nat}$. 
  The adversary $\disttwo$ receives in input a sample $x$ and $1^\nat$, then 
  it computes $b'=\head(x)$ and returns $\distone(b', 1^\nat)$. Observe that
  if the input of $\distone$ follows the distribution $\sem r \nat(\stthree_\rvarone)$,
  then the output of $\disttwo$ is the same of $\distone$ when its input is sampled
  from $\sem{\head(r)}\nat(\stthree)$, because for \Cref{lemma:projexpsem}
  $\sem r \nat(\stthree_\rvarone)$=$\sem r \nat(\stthree)$; similarly, if the input is taken from
  $\unif{\stng\nat}$, then the output of $\disttwo$ is the same of $\distone$
  when the input is sampled from $\unif\bool$. In particular, the 
  advantage of $\disttwo$ is the same of $\distone$,
  but this contradicts (H1). So, it must be the case where
  $\srenv \renvtwo \stthree {\unif{\head(\rvarone)}}$, 
  For axiom (U1) and $\srenv \renvtwo \stthree  \indp b {\head(\rvarone)}$,
  we deduce $\srenv \renvtwo \stthree {\ud b}$. 
  From \Cref{lemma:projexpsem}, we deduce that
  $\sem b {} (\stthree_b)=\sem b {}(\stthree)$.
  From this observation, we conclude $ \srenv {\renvtwo_b} {\stthree_b}  {\ud b}$.
  The proof of (C2) is analogous to that of (C1).
  Now, we show (C3).
  Assume that $\distone$ is a distinguisher for $\stthree_b \comp \stthree_\rvartwo$
  and $\stthree_{b, \rvartwo}$ with non-negligible advantage. 
  Due to compatibility of $\comp$ and $\ind$,~\Cref{lemma:indcompcomp}, $\distone$
  must have non-negligible advantage also in distinguishing 
  $\stfour_b \comp \stfour_\rvartwo$ and $\stthree_{b, \rvartwo}$, where:
  \begin{itemize}
  \item $\stfour_b$ is the efficiently samplable distribution ensemble such that 
    for every $\nat$, $(\stfour_b)_\nat = \{\{b \mapsto 0\}^{\frac 1 2},
    \{b \mapsto 1\}^{\frac 1 2}\}$.
  \item $\stfour_\rvartwo$ is the efficiently samplable distribution ensemble such that 
    for every $\nat$,
    $
      (\stfour_\rvartwo)_\nat=\{\{\rvartwo \mapsto z\}^{2^{-n}},
      \text{ for } z \in \BB^\nat\}
    $. 
  \end{itemize}
  The distinguisher $\distone$ can be used to define a
  new adversary $\disttwo$ that distinguishes 
  the distribution $\sem \rvarone {}(\stthree)$ from $\unif{\stng{n+1}}$.
  The distinguisher $\disttwo$ takes in input 
  its sample and the security parameter $1^n$, then it produces the store:
  \[
    (m_x)_\nat = \{b \mapsto \head(x), \rvartwo \mapsto \tail(x)\}
  \]
  and returns $\distone((m_x)_n, 1^\nat)$.
  If $x$ is sampled according to $\sem \rvarone {}(\stthree)$, we
  observe that for every $\nat \in \NN$, the samples $(m_x)_n$
  follow exactly the distribution $(\stthree_{b, \rvartwo})_\nat$. 
  This holds because for every $\nat$, thanks to (H2) and (H3),
  we have:
  \small
  \[
    \stthree_\nat = \{\{\rvarone \mapsto z, b \mapsto \head(z), \rvartwo \mapsto tail(z)\}^{p_z}, \text{for } z \in  \BB^{\nat+1}\},
  \]
  \normalsize
  where $\{p_z\}_{z \in \BB^{n+1}}$ is an indexed family of non-negative reals 
  such that $\sum_{x\in \BB^{n+1}} p_x =1$. For this reason, we deduce that:
  \[
    \sem \rvarone\nat (\stthree) = \{\{z^{\nat+1}\}^{p_z}, \text{for  } z \in \BB^{n+1}\}
  \]
  So, if the sample $x$ is taken from the distribution $\sem\rvarone \nat(\stthree)$, 
  the distribution of the samples that are given in input
  to $\distone$, is:   $F(\sem \rvarone \nat (\stthree))$ where
  $F(z) = \{ b \mapsto head(z), \rvartwo \mapsto tail(z)\}$. Since $F$ is an injection, we have that: 
  \[
    F(\sem \rvarone \nat (\stthree)) = \{\{b \mapsto head(z), \rvartwo \mapsto tail(z)\}^{p_z}, \text{for } z \in \BB^{\nat+1}\},
  \]
  which is exactly  $(\stthree_{b, \rvartwo})_\nat$.
  For a similar reason, we observe that
  if $x$ is sampled according to $\unif{\stng{\nat+1}}$, the inputs
  of $\distone$ follow the distribution $(\stfour_b \comp \stfour_\rvartwo)_\nat$.
  This shows that $\disttwo$ succeeds in distinguishing $\sem \rvarone \stfour$
  from $\unif{\stng{\nat+1}}$ with non-negligible advantage, but this is
  a contradiction of the assumption $\srenv \renvtwo \stthree  \unif\rvarone$.
  This shows $\stfour_b \comp \stfour_\rvartwo \ind \stthree_{b, \rvartwo}$, which was our goal.
  Now, we show axiom \eqref{eq:axmerge}. The goal is to show:
  \[
    \srenv\renvtwo \stthree  {\ud\rvartwo}
  \]
  Under the assumption:
  \[
    \srenv \renvtwo \stthree {(\tyf{\ud \rvarone}{\renvtwo_\rvarone} \sep \tyf{\ud b}{\renvtwo_b}) \land \tyf {\espl\rvartwo {\concat(\rvarone, b)}}\renvtwo}
  \]
  Which means that there are $\stthree_1, \stthree_2$ such that:
  \begin{proofcases}
    \proofcase[H1] ${\srenv {\renvtwo_{\rvarone}} {\stthree_1} {\ud {\rvarone}}}$
    \proofcase[H2] ${\srenv {\renvtwo_{b}} {\stthree_2} {\ud {b}}}$
    \proofcase[H3] ${{\stthree_1}\comp {\stthree_2} \ind \stthree_{\rvarone, b}}$
  \end{proofcases}
  Let $\stfour_1, \stfour_2$ be the distribution ensembles defined as follows:
  \begin{align*}
    (\stfour_1)_\nat &= \{ \{\rvarone \mapsto z\}^{2^{-p(\nat)}}, \text{for }z \in \BB^\nat\} \\
    (\stfour_2)_\nat &= \{ \{b \mapsto 0\}^{\frac 1 2} , \{b \mapsto 1\}^{\frac 1 2}\}
  \end{align*}
  Observe that $\stfour_1 \ind \stthree_1$ and $\stfour_2 \ind \stthree_2$. We show this 
  result only in the case of $\stfour_1$ and $\stthree_1$, that of $\stfour_2$ and $\stthree_2$
  is analogous.

  Assume that there is a distinguisher $\distone$ for $\stfour_1$ and $\stthree_1$. We can
  use this distinguisher to define a new distinguisher $\disttwo$ for
  $\sem \rvarone {} (\stthree_1)$ and $\unif{\stng {p(\nat)+1}}$. This adversary
  takes in input a sample $x$ and $1^\nat$ and returns
  $\distone(\{\rvarone \mapsto x\}, 1^\nat)$.
  If $x$ is sampled uniformly, then the distribution of $\distone$'s inputs
  is $\{ \{ \rvarone \mapsto z\}^{2^{-\nat}}, \text{for }z \in \BB^\nat\}$, i.e. $\stfour_1$.
  If the input is sampled according to $\sem \rvarone \nat (\stthree_1)$, we observe that:
  \[
    (\stthree_1) = \{\{\rvarone \mapsto z\}^{p_z} \text{ for }z \in \BB^n\}
  \]
  where $\{p_z\}_{z \in \BB^n}$ is a family of non-negative reals such that
  $\sum_{z \in \BB^n} p_r =1$. We observe that:
  \[
    \sem \rvarone \nat (\stthree_1) = \{z^{p_z} \text{ for }z \in \BB^\nat\}
  \]
  because the function $\{k \mapsto z\} \mapsto z$ is a bijection. For the same reason, 
  the inputs of $\distone$ follow exactly the distribution $\stthree_1$. This shows that the 
  adversary $\disttwo$ has the same advantage of $\distone$, but this contradicts the 
  assumption $\srenv {\renvtwo_\rvarone}{\stthree_1}  \ud\rvarone$, so we conclude
  $\stthree_1 \ind \stfour_1$. Using a similar argument on $\stthree_2$ and $\stfour_2$, we
  can apply compatibility to show: $\stfour_1 \comp \stfour_2 \ind \stthree_{\rvarone, b}$ 
  From this result and \Cref{lemma:exprind}, we deduce that:
  \[
    \sem {\concat(\rvarone, b)}{}(\stthree_{\rvarone, b}) \ind \sem{\concat(\rvarone, b)}{}(\stfour_1 \comp \stfour_2)
  \]
  In particular, we observe that
  $\sem{\concat(\rvarone, b)}{}(\stfour_1 \comp \stfour_2) = \unif{\stng{p(\nat)+1}}$ 
  by definition of $\sem \concat {} (\rvarone, b){}$. This means that 
  $\sem {\concat(\rvarone, b)}{}(\stthree_{\rvarone, b}) \ind \unif{\stng{p(\nat)+1}}$, which entails:
  \[
    \srenv {\renvtwo_{\rvarone, b}}{\stthree_{\rvarone, b}}  {\ud{\concat(\rvarone, b)}}.
  \]
  From this result and \Cref{cor:restriction}, we deduce:
  \[
    \srenv \renvtwo \stthree \ud{\concat(\rvarone, b)}.
  \]
  From axioms (W1-2), we deduce $\srenv \renvtwo \stthree {\indp \rvartwo {\concat(\rvarone, b)}}$, and
  with an application of axiom (U1), we conclude $\srenv \renvtwo \stthree \unif \rvartwo$.
\end{proof}

We start by showing
the construction of the trees $\juone_i$ such that:
\[
  \hotj{\tyf{\fone_i^h}{\renvone_h}}{\binjudg {\renvone_h} {\prgone_i} }{\tyf{\fone_{i+1}^h}{\renvone_h}}
\]
where
\begin{align*} 
  \fone_i^h &=\tyf{\ud k}{k:\stng \nat}\sep\;\;\smashoperator[lr]{\Sep_{j=0}^{i-1}}\;\; \tyf{\ud{b_j}}{b_j:\bool}
\end{align*}
So, for $i=h+1$, this corresponds to showing the triple:
\[
  \hotj{\tyf{\fone_0^h}{\renvone_h}}{\binjudg {\renvone_h} {\prgone_h} }{\tyf{\fone_{h+1}^h}{\renvone_h}}.
  \tag{$*$}
\]
\noindent
Observe that $\fone_{h+1}^h$ corresponds to:
\[
\tyf{\ud k}{k:\stng \nat}\sep\;\;\smashoperator[lr]{\Sep_{j=0}^{h}}\;\; \tyf{\ud{b_j}}{b_j:\bool}.
\]
The triple
\[
  \hotj{\ud k }{\binjudg {\renvone_h} {\prgone_0;\ldots;\prgone_h} }{\tyf{\ud k}{k:\stng \nat}\sep{\Sep_{j=0}^{h}}\;\; \tyf{\ud{b_j}}{b_j:\bool}
  }
\]
can thus be obtained from $\juone_{h}$ with an application of the $\RWEAK$ rule and multiple applications of the $\RSEQ$ rule. In particular, with the $\RWEAK$ rule, we observe
\[
  \ud k \models^{\renvone_h} \ud k \sep {\tyf \top \varepsilon} = \fone^h_0
\]

With application of the rule $\SDASS$, we can obtain the triple
\begin{gather*}
  \vdash \vhot{\tyf{\ud k}{k:\stng \nat}\sep{\Sep_{j=0}^{h}}\;\; \tyf{\ud{b_j}}{b_j:\bool}
  }
  {\binjudg{\renvone_{h}} {\ass {\rvartwo_0} k}}
  {\tyf{\ud k \land \espl {\rvartwo_0} k}{k, \rvartwo_0:\stng \nat}\sep{\Sep_{j=0}^{h}}\;\; \tyf{\ud{b_j}}{b_j:\bool}
  }
\end{gather*}
By means of the $\RWEAK$ rule and axioms (W1), (W2) and (U1), we can rewrite the conclusion as follows:
$\tyf{\ud  {\rvartwo_0} }{k, \rvartwo_0:\stng \nat}\sep{\Sep_{j=0}^{h}}\;\; \tyf{\ud{b_j}}{b_j:\bool}$.
Finally, with \Cref{lemma:prekm} and \Cref{lemma:extcompcomp} we can simplify the conclusion one further step, to obtain:
\begin{gather*}
  \vdash \hot{\tyf{\ud k}{k:\stng \nat}\sep{\Sep_{j=0}^{h}}\;\; \tyf{\ud{b_j}}{b_j:\bool}
  }
  {\binjudg{\renvone_{h}} {\ass {\rvartwo_0} k}}
  {\left(\tyf{\ud {\rvartwo_0}}{\rvartwo_0:\stng \nat}\right)\sep\left({\Sep_{j=0}^{h}}\;\; \tyf{\ud{b_j}}{b_j:\bool}\right)
  }
\end{gather*}
Observe that this can be put in a chain of applications of the $\RSEQ$
rule to show:
\begin{equation*}
  \vdash \vhot {\ud k} {\binjudg{\renvone_h}\prgone_0;\ldots;\prgone_h;\ass{\rvartwo_0} k}
  {\left(\tyf{\ud {\rvartwo_0}}{\rvartwo_0:\stng \nat}\right)\sep\left({\Sep_{j=0}^{h}}\;\; \tyf{\ud{b_j}}{b_j:\bool}\right)
  }
\end{equation*}

We continue by defining the formula $\ftwo_i^h$ as follows:

\[
  \ftwo_i^h\defsym \left(\tyf{\ud {\rvartwo_i}}{\rvartwo_i:\stng {\nat+i}}\right)\sep\left({\Sep_{j=i}^{h}}\;\; \tyf{\ud{b_j}}{b_j:\bool}\right)
\]

Then, for every $\prgtwo_i$, for $0\le i\le h$, we build the
$\jutwo_i$ such that

\[
  \jutwo_i \deriv  \hotj {\ftwo_i^h} {\binjudg {\renvone_h} {\prgtwo_i}} {\ftwo_{i+1}^h}  
\]

These trees are descried in \Cref{fig:jutwoi} and can be combined by means of multiple applications of the $\RSEQ$ rule, this produces as conclusion the formula ${\ftwo_{h+1}^h}$, i.e. ${\ud {\rvartwo_{h+1}}{}}$. 

All these trees can
be combined with the $\RSEQ$ rule to prove the triple
\[
  \hot{\tyf{\ud {\rvartwo_0}}{\rvartwo_0:\stng {\nat}}\sep\left({\Sep_{j=i}^{h}}\;\; \tyf{\ud{b_j}}{b_j:\bool}\right)}{\binjudg {\renvone_h} {\seq{\seq{\prgtwo_{0}}\ldots}{\prgtwo_{h}}
    }
  }
  {
    \tyf{\ud{\rvartwo_{h+1}}}{\renvone_h}
  }
\]
\noindent
This triple is obtained with multiple applications of the $\RSEQ$ rule
to the family of trees $\jutwo_0, \ldots, \jutwo_h$.
%
%
The claim is shown with of an application of the
$\RSEQ$ rule to the last triple and $(*)$. 

\begin{landscape}
  
\begin{figure}[t]
  The tree $\juone_i$ is the following one:
  \[
    \infer[\RWEAK]{\hot{{(\tyf{\ud k}{k:\stng \nat} \sep\;\;\smashoperator[lr]{\Sep_{j=0}^{i-1}}\;\; \tyf{\ud{b_j}}{b_j:\bool})}}
      {\binjudg {\renvone_h} {\prgone_i}}{{\tyf{\ud k}{k:\stng \nat}\sep\smashoperator[lr]{\Sep_{j=0}^{i}}\;\; \tyf{\ud{b_j}}{b_j:\bool}}}
    }
    {
      \infer[\RSEQ]{\hot{{(\tyf{\ud k}{k:\stng \nat} \sep\;\;\smashoperator[lr]{\Sep_{j=0}^{i-1}}\;\; \tyf{\ud{b_j}}{b_j:\bool})}}
        {\binjudg {\renvone_h} {\prgone_i}}{{\tyf{\ud{\rvarone_i}\land \espl {b_i} {\head(\rvarone_i)} \land \espl {k} {\tail(\rvarone_i)}}{k:\stng \nat, \rvarone_i:\stng{\nat+1}, b_i:\BB}\sep\smashoperator[lr]{\Sep_{j=0}^{i-1}}\;\; \tyf{\ud{b_j}}{b_j:\bool}}}
      }
      {
        \juone_i^1
        &
        \juone_i^2
        &
        \juone_i^3
      }
      &
      \deduce{\deduce{\text{and \Cref{rem:conjasscomm}}}{\deduce{\text{\Cref{lemma:extcompcomp}}}{\text{\Cref{lemma:prekm}}}}}{\eqref{eq:axsplit}}
    }
  \]

  The trees $\juone_i^1, \juone_i^2, \juone_i^3$ are respectively the following ones:

  \[
    \infer[\RWEAK]{\hot{{(\tyf{\ud k}{k:\stng \nat} \sep\;\;\smashoperator[lr]{\Sep_{j=0}^{i-1}}\;\; \tyf{\ud{b_j}}{b_j:\bool})}}{\binjudg {\renvone_h} {\ass {\rvarone_i}{g(k)}}} {{\tyf{\ud{\rvarone_i}}{k:\stng \nat, \rvarone_i:\stng{\nat+1}}\sep \smashoperator[lr]{\Sep_{j=0}^{i-1}}\;\; \tyf{\ud{b_j}}{b_j:\bool}}}}
    {
      \infer[\SDASS]{\hot{{(\tyf{\ud k}{k:\stng \nat} \sep\;\;\smashoperator[lr]{\Sep_{j=0}^{i-1}}\;\; \tyf{\ud{b_j}}{b_j:\bool})}}{\binjudg {\renvone_h} {\ass {\rvarone_i}{g(k)}}}  {{\tyf{{\ud k} \land \espl {\rvarone_i} {g(k)}}{k:\stng \nat, \rvarone_i:\stng{\nat+1}}\sep \smashoperator[lr]{\Sep_{j=0}^{i-1}}\;\; \tyf{\ud{b_j}}{b_j:\bool}}}}
      {
        \ternjudg {k:\stng \nat} {g(k)} {\stng {\nat+1}}
        &
        \rvarone_i \notin \fv{g(k)}
        &
        {\smashoperator[lr]{\Sep_{j=0}^{i-1}}\;\; \tyf{\ud{b_j}}{b_j:\bool}}\in \fset
        &
        \rvarone_i \notin \dom (\{k:\stng \nat\})
      }
      &
      \text{\eqref{eq:axpotp}, (W1), (W2) and (U1)}
    }
  \]

  \[
    \infer[\RWEAK]{\hot{{\tyf{\ud{\rvarone_i}}{k:\stng \nat, \rvarone_i:\stng{\nat+1}}\sep \smashoperator[lr]{\Sep_{j=0}^{i-1}}\;\; \tyf{\ud{b_j}}{b_j:\bool}}} {\binjudg {\renvone_h} {\ass {b_i}{\head(\rvarone_i)}}} {{\tyf{\ud{\rvarone_i}\land \espl {b_i} {\head(\rvarone_i)}}{\rvarone_i:\stng{\nat+1}, b_i:\BB}\sep\smashoperator[lr]{\Sep_{j=0}^{i-1}}\;\; \tyf{\ud{b_j}}{b_j:\bool}}}}
    {
      \infer[\SDASS]{\hot{{\tyf{\ud{\rvarone_i}}{k:\stng \nat, \rvarone_i:\stng{\nat+1}}\sep \smashoperator[lr]{\Sep_{j=0}^{i-1}}\;\; \tyf{\ud{b_j}}{b_j:\bool}}} {\binjudg {\renvone_h} {\ass {b_i}{\head(\rvarone_i)}}} {{\tyf{\ud{\rvarone_i}\land \espl {b_i} {\head(\rvarone_i)}}{k:\stng \nat, \rvarone_i:\stng{\nat+1}, b_i:\BB}\sep\smashoperator[lr]{\Sep_{j=0}^{i-1}}\;\; \tyf{\ud{b_j}}{b_j:\bool}}}}
      {
        \ternjudg {k:\stng \nat, \rvarone_i:\stng{\nat+1}} {\head(\rvarone_i)} \BB
        &
        b_i \notin \fv{\head(\rvarone_i)}
        &
        \smashoperator[lr]{\Sep_{j=0}^{i-1}}\;\; \tyf{\ud{b_j}}{b_j:\bool}\in \fset
        &
        b_i \notin \dom (\{k:\stng \nat,  \rvarone_i:\stng{\nat+1}\})
      }
    }
  \]

  \[
    \infer[\SDASS]{\hot{{\tyf{\ud{\rvarone_i}}{k:\stng \nat, \rvarone_i:\stng{\nat+1}}\sep \smashoperator[lr]{\Sep_{j=0}^{i-1}}\;\; \tyf{\ud{b_j}}{b_j:\bool}}}
      {\binjudg {\renvone_h} {\ass {k}{\tail(\rvarone_i)}}}{{\tyf{\ud{\rvarone_i}\land \espl {b_i} {\head(\rvarone_i)} \land \espl {k} {\tail(\rvarone_i)}}{k:\stng \nat, \rvarone_i:\stng{\nat+1}, b_i:\BB}\sep\smashoperator[lr]{\Sep_{j=0}^{i-1}}\;\; \tyf{\ud{b_j}}{b_j:\bool}}}
    }
    {
      \ternjudg {\rvarone_i:\stng{\nat+1}, b_i:\BB} {\tail(\rvarone_i)} {\stng \nat}
      &
      k \notin \fv{\tail(\rvarone_i)}
      &
      \smashoperator[lr]{\Sep_{j=0}^{i-1}}\;\; \tyf{\ud{b_j}}{b_j:\bool} \in \fset
      &
      k \notin \dom (\{\rvarone_i:\stng{\nat+1}, b_i:\BB\})
    }
  \]

  \label{fig:jutwo_i}  \caption{Construction of the family of trees $\juone_i$ for $i \in \NN$.}
  \label{fig:juonei}
\end{figure}

\begin{figure}[t]
  \resizebox{25cm}{!}{\(
    \infer[\RWEAK]{\hot{\tyf{\ud {\rvartwo_i}}{\rvartwo_i:\stng {\nat+i}}\sep\left({\Sep_{j=i}^{h}}\;\; \tyf{\ud{b_j}}{b_j:\bool}\right)} {\binjudg {\renvone_h} {\prgtwo_i}} {\tyf{\ud {\rvartwo_{i+1}}}{\rvartwo_{i+1}:\stng {\nat+i+1}}\sep\left({\Sep_{j=i+1}^{h}}\;\; \tyf{\ud{b_j}}{b_j:\bool}\right)}}
    {
      \infer[\RWEAK]{\hot{\tyf{\ud {\rvartwo_i}\sep \ud {b_i}}{\rvartwo_i:\stng {\nat+i}, b_i:\BB}\sep\left({\Sep_{j={i+1}}^{h}}\;\; \tyf{\ud{b_j}}{b_j:\bool}\right)} {\binjudg {\renvone_h} {\prgtwo_i}} {\tyf{\ud {\rvartwo_{i+1}}}{\rvartwo_{i+1}:\stng {\nat+i+1}}\sep\left({\Sep_{j=i+1}^{h}}\;\; \tyf{\ud{b_j}}{b_j:\bool}\right)}}
      {        
        \infer[\SDASS]{\hot{\tyf{\ud {\rvartwo_i}\sep \ud {b_i}}{\rvartwo_i:\stng {\nat+i}, b_i:\BB}\sep\left({\Sep_{j={i+1}}^{h}}\;\; \tyf{\ud{b_j}}{b_j:\bool}\right)} {\binjudg {\renvone_h} {\prgtwo_i}} {\tyf{\ud {\rvartwo_i}\sep \ud {b_i}\land \espl {\rvartwo_{i+1} }{\concat ({\rvartwo_i}, {b_i})}}{\rvartwo_{i+1}:\stng {\nat+i+1},\rvartwo_i:\stng {\nat+i}, b_i:\BB}\sep\left({\Sep_{j=i+1}^{h}}\;\; \tyf{\ud{b_j}}{b_j:\bool}\right)}}
        {
          \ternjudg {\rvartwo_{i}:\stng{\nat+i}, b_i:\BB} {\concat({\rvartwo_i}, {b_i})} {\stng {\nat+i+1}}
          &
          \rvartwo_{\nat+i+1} \notin \fv{{\concat({\rvartwo_i}, {b_i})}}
          &
          {\Sep_{j={i+1}}^{h}}\;\; \tyf{\ud{b_j}}{b_j:\bool} \in \fset
          &
          \rvartwo_{i+1} \notin \dom (\{\rvartwo_i:\stng{\nat+i}, b_i:\BB\})
        }
        &
        \deduce{\text{ and \Cref{lemma:prekm}.}}{\deduce{\text{\Cref{lemma:extcompcomp}}}{\text{\eqref{eq:axmerge}}}}
      }
      &
      \text{\Cref{rem:conjasscomm}}
    }
  \)}
  \caption{Construction of the family of trees $\jutwo_i$ for $i \in \NN$. }
  \label{fig:jutwoi}
\end{figure}
\end{landscape}


\section{Characterizing Cryptographic Security Through Computational Independence}
\label{sec:compindep}
The study of computational independence and its cryptographic applications 
originates with \cite{Fay14}\footnote{The approach of \cite{Fay14} was 
motivated by Yao's characterization of encryption security using a 
computational notion of mutual information \cite{Yao}}, which we take as our 
starting point. We note however, that the definition given there does not 
directly coincide with the notion that we use to interpret separating 
conjunction in \CPSL. In this section, we will give a definition that is more 
closely related to the approach of this paper, and show that with respect to 
efficiently samplable distributions it is indeed equivalent to that given in 
\cite{Fay14}. Then we use this definition to give a characterization of 
encryption security in terms of computational independence. In this way,
we show that the forms 
of logical and structural reasoning provided by \CPSL are consistent 
with the standard approach of provable security in computational cryptography.

\subsection{Defining and Characterizing Computational Independence}
\label{sec:compindchar}

\noindent
We begin by recalling some basic definitions for this setting.
\begin{definition} A \emph{random variable} $\rvone$ over a finite set $X$ denotes a random process taking values in $X$ according to some distribution $\mathcal{D}_\rvone$. We write $\Pr[\rvone=x]$ to denote $\mathcal{D}_\rvone(x)$. Given a family $\{X_n\}_{n\in \omega}$ of finite sets, a \emph{random variable ensemble} is a family $\{\rvone_n\}_{n\in\omega}$, where $\rvone_n$ is a random variable over $X_n$. Let $\rvone,\rvtwo$ be jointly distributed random variables over $X,Y$. $\rvone$ and $\rvtwo$ are \emph{(statistically) independent} if their joint distribution function satisfies: $\Pr[\rvone=x \wedge \rvtwo=y]=\Pr[\rvone=x]\cdot\Pr[\rvone=y]$. The joint distribution function $\mathcal{D}_{\rvone,\rvtwo}$ of $\rvone,\rvtwo$ induces \emph{marginal distribution functions:} $\mathcal{D}_\rvone$ is defined by $\mathcal{D}_\rvone(x)=\sum_{y \in Y}\mathcal{D}_{\rvone,\rvtwo}(x,y)$ and $\mathcal{D}_\rvtwo$ is defined similarly. We write $\rvone \tensprod \rvtwo$ to denote the random variable obtained by sampling pairs of values independently from the random variables with these respective marginal distribution functions.
We note that these definitions may be generalized in the natural way from pairs 
of random variables to arbitrary tuples, and also lift to random variable 
ensembles. A random variable is \emph{efficiently samplable} if its 
distribution function is efficiently samplable. Random variable ensembles 
$\{\rvone_n\}_{n\in\omega},\{\rvtwo_n\}_{n\in\omega}$ where $\rvone_n$ has 
distribution $\distrone_n$ and $\rvtwo_n$ has distribution $\distrtwo_n$ are 
\emph{computationally indistinguishable} if $\{\distrone_n\}_{n\in\omega} \ind 
\{\distrtwo_n\}_{n\in\omega}$. In this case we write 
$\{\rvone_n\}_{n\in\omega}\ind\{\rvtwo_n\}_{n\in\omega}$.
\hfill\qed
\end{definition}

To simplify notation, we will write $\{\rvone_n\}$ as shorthand for 
$\{\rvone_n\}_{n\in\omega}$, and write $\{\rvone_n\}, \{\rvtwo_n\},...$ to denote random variable ensembles. We note that the notation $\tensprod$ is 
consistent with that introduced for distributions. In particular, it is not 
hard to verify that $\mathcal{D}_{\rvone \tensprod \rvtwo}=\mathcal{D}_\rvone 
\tensprod \mathcal{D}_\rvtwo$. What we are calling a random variable (ensemble) 
is often called a distribution (ensemble). Indeed, in the discrete setting, 
the notions of distribution, random variable, and probability mass function are 
often used synonymously. Here we are being more careful.

It is now time to rephrase one of the key concepts of this paper in terms of 
random variables:
\begin{definition}
 $\{\rvone_n\}$, $\{\rvtwo_n\}$ are \emph{computationally independent} if there exist statistically independent ensembles  $\{\rvone'_n\}$, $\{\rvtwo'_n\}$ such that $(\{\rvone_n\},\{\rvtwo_n\})\ind(\{\rvone'_n\},\{\rvtwo'_n\})$.
 \hfill\qed
\end{definition}

This definition, from \cite{Fay14}, was originally given with respect to both uniform and 
nonuniform adversaries. We restrict attention to uniform adversaries. Moreover, 
unless otherwise noted we will assume that all random variable ensembles are 
efficiently samplable. This includes the witnessing ensembles in the definition 
of computational independence. 

The alternate characterization mentioned above is given by the following 
result, which states that computational independence can be characterized, 
again through indistinguishability, but referring to tensor products, instead.
\begin{theorem}
\label{thm:compindepeq}
$\{\rvone_n\}$, $\{\rvtwo_n\}$ are computationally independent if and only if $(\{\rvone_n\},\{\rvtwo_n\})\ind \{\rvone_n\}\tensprod \{\rvtwo_n\}$.
\hfill\qed
\end{theorem}

First of all, one can observe that efficient 
samplability is preserved by taking projections:
\begin{proposition}
\label{marg}
Suppose that jointly distributed $(\{\rvone_n\},\{\rvtwo_n\})$ are efficiently samplable.
Then $\{\rvone_n\}$ and $\{\rvtwo_n\}$  are efficiently samplable.
\hfill\qed
\end{proposition}

Then, one can easily prove that indistinguishability between the 
corresponding 
components 
marginals of a product random variable follows from the indistinguishability 
between the two products themselves:
\begin{lemma}
\label{eq1}
For any (not necessarily efficiently samplable)
 $\{\rvone_n\}$, $\{\rvtwo_n\}$, $\{\rvone'_n\}$ and $\{\rvtwo'_n\}$,
 if $(\{\rvone_n\},\{\rvtwo_n\})\ind(\{\rvone'_n\},\{\rvtwo'_n\})$ , then $\{\rvone_n\} \ind\{\rvone'_n\}$  and  $\{\rvtwo_n\}\ind\{\rvtwo'_n\}$.
    \hfill\qed
\end{lemma}
\begin{proof}
Suppose $\distone$ is a probabilistic poly-time (PPT) distinguisher that distinguishes $\{\rvone_n\}$ and $\{\rvone'_n\}$ with advantage $\epsilon(n)$ and define the distinguisher $\distone'$ by $\distone'(x,y)=\distone(x)$. Then  $\distone'$  is a PPT distinguisher that distinguishes  $(\{\rvone_n\},\{\rvtwo_n\})$ and $(\{\rvone'_n\},\{\rvtwo'_n\})$  with advantage $\epsilon(n)$.
\end{proof}

\begin{lemma}
\label{eq2}
Suppose $\{\rvone_n\}$, $\{\rvone'_n\}$ and $\{\rvtwo_n\}$, are  ensembles such that  $\{\rvone_n\}\ind\{\rvone'_n\}$. Then $\{\rvone_n\}\tensprod\{\rvtwo_n\}\ind \{\rvone'_n\}\tensprod\{\rvtwo_n\}$.
\end{lemma}
\begin{proof}
Suppose  $\{\rvone_n\}\tensprod\{\rvtwo_n\}\not\ind \{\rvone'_n\}\tensprod\{\rvtwo_n\}$, as witnessed by the PPT distinguisher $\distone$ with advantage $\epsilon(n)$ Let $\mathcal{S}$ be a PPT sampler for  $\{\rvtwo_n\}$. Define $\distone'$ by
$\distone'(x)=\distone(x,\mathcal{S}(1^{|x|}))$. Then  $\distone'$ is PPT, and it distinguishes $\{\rvone_n\}$ from $\{\rvone'_n\}$ with advantage $\epsilon(n)$. 
\end{proof}

Using \Cref{eq2}, we can show that two pairs of indistinguishable random 
variables ensembles can be combined through tensor
product preserving indistinguishability:
\begin{lemma}
\label{eq3}
If $\{\rvone_n\}\ind\{\rvone'_n\}$ and 
$\{\rvtwo_n\}\ind\{\rvtwo'_n\}$, then $\{\rvone_n\}\tensprod\{\rvtwo_n\}\ind 
\{\rvone'_n\}\tensprod\{\rvtwo'_n\}$.
    \hfill\qed
\end{lemma}
\begin{proof}
By Lemma~\ref{eq2} we have $\{\rvone_n\}\tensprod\{\rvtwo_n\}\ind \{\rvone'_n\}\tensprod\{\rvtwo_n\}$ and, by a similar argument, $\{\rvone'_n\}\tensprod\{\rvtwo_n\}\ind \{\rvone'_n\}\tensprod\{\rvtwo'_n\}$. The desired result follows by the transitivity of $\ind$.
\end{proof}

\Cref{thm:compindepeq} follows naturally from these intermediate
results.

\begin{proof}
Suppose  $\{\rvone_n\}$, $\{\rvtwo_n\}$ are computationally independent, as witnessed by statistically independent distributions $\{\rvone'_n\}$, $\{\rvtwo'_n\}$.
By Lemma~\ref{eq1}, $\{\rvone_n\} \approx \{\rvone'_n\}$ and $\{\rvtwo_n\} \approx \{\rvtwo'_n\}$ and then by Lemma~\ref{eq3},  $\{\rvone_n\}\tensprod \{\rvtwo_n\}\approx \{\rvone'_n\}\tensprod\{\rvtwo'_n\}$. By the independence of $\{\rvone'_n\}$ and $\{\rvtwo'_n\}$, we have that $ (\{\rvone'_n\},\{\rvtwo'_n\})$ and $\{\rvone'_n\}\tensprod\{\rvtwo'_n\}$ are identically distributed, so \emph{a fortiori} $ (\{\rvone'_n\},\{\rvtwo'_n\})\approx\{\rvone'_n\}\tensprod\{\rvtwo'_n\}$. This gives\\[1.5ex]
\begin{align*}
  (\{\rvone_n\},\{\rvtwo_n\})\approx(\{\rvone'_n\},\{\rvtwo'_n\})\approx\{\rvone'_n\}\tensprod\{\rvtwo'_n\}\approx\{\rvone_n\}\tensprod \{\rvtwo_n\}
\end{align*}
For the converse, note that by Proposition~\ref{marg} the marginal distributions $\{\rvone_n\}$ and $\{\rvtwo_n\}$ are efficiently samplable, so just take independently sampled copies $\{\rvone'_n\}$,$\{\rvtwo'_n\}$ of $\{\rvone_n\}$ and $\{\rvtwo_n\}$, respectively.
\end{proof}

In the following, we take Theorem~\ref{thm:compindepeq} as the characterization 
of computational independence, and introduce the following notation:
we write $\{\rvone_n\}\indep_c\{\rvtwo_n\}$ to denote 
$(\{\rvone_n\},\{\rvtwo_n\})\ind\{\rvone_n\}\tensprod\{\rvtwo_n\}$.

\subsection{Computational Independence in a Cryptographic Setting}
\label{sec:compindcrypto}

We now give two examples of the use of computational independence in a cryptographic setting. The first proves a new fundamental property of pseudorandom random variable ensembles. The second is a characterization of encryption security against eavesdropping adversaries. Such a characterization was provided for nonuniform adversaries in \cite{Fay14}. Here we extend the characterization to the uniform setting.

Write $\U{\{\rvone_n\}}$ to denote that $\{\rvone_n\}$ is pseudorandom (i.e. 
indistinguishable from the uniform random variable ensemble.) The following result
tells us that a joint random variable ensemble being pseudorandom is equivalent 
to their components being computationally independent \emph{and} pseudorandom.
\begin{lemma}
  \label{lemma:unifcompind}
Suppose  $\{\rvone_n\}$, $\{\rvtwo_n\}$ are samplable ensembles. Then $\U{(\{\rvone_n\},\{\rvtwo_n\})}$  if and only if  $\{\rvone_n\} \indep_c \{\rvtwo_n\}$ and $\U{\{\rvone_n\}}$, $\U{\{\rvtwo_n\}}$.
\hfill\qed
\end{lemma}
\begin{proof}
Suppose $\U{\{\rvone_n\}}$, $\U{\{\rvtwo_n\}}$ and $\{\rvone_n\} \indep_c \{\rvtwo_n\}$. So we have  $\{\rvone_n\}\approx \{\mathbf{U}_n\}$  and $\{\rvtwo_n\}\approx \{\mathbf{V}_n\}$ where $\mathbf{U}_n$, $\mathbf{V}_n$ are uniform. It then follows by Lemma~\ref{eq3} that $\U{\{\rvone_n\} \otimes \{\rvtwo_n\}}$. From $\{\rvone_n\} \indep_c \{\rvtwo_n\}$ we may conclude that 
$\U{(\{\rvone_n\},\{\rvtwo_n\})}$.
Conversely, if $\U{(\{\rvone_n\},\{\rvtwo_n\})}$ then $(\{\rvone_n\},\{\rvtwo_n\})\approx \{\mathbf{W}_n\}$ where $\mathbf{W}_n$ is uniform.
But then $(\{\rvone_n\},\{\rvtwo_n\})\approx (\{\mathbf{U}_n\},\{\mathbf{V}_n\})$ where $\mathbf{U}_n$, $\mathbf{V}_n$ are uniform and independent.
  So by Theorem~\ref{thm:compindepeq},  $\{\rvone_n\} \indep_c \{\rvtwo_n\}$, and by Lemma~\ref{marg}, $\U{\{\rvone_n\}}$ and $\U{\{\rvtwo_n\}}$.
\end{proof}
The Lemma above demonstrates that the results of \Cref{subsec:stretching} closely reflect the computational setting, 
in particular the way the $\U{\cdot}$ predicate and separating conjunction 
interact in \CPSL.
Note that for the ``if'' direction, if the assumptions are witnessed by 
$\epsilon_i(n)$, $i=1,2,3$ then the conclusion is witnessed by 
$\epsilon_1(n)+\epsilon_2(n)+\epsilon_3(n)$. However, in the ``only if" 
direction, if $\U{(\{\rvone_n\},\{\rvtwo_n\})}$ is witnessed by $\epsilon(n)$ 
then so are each of the conclusions. So the equivalence is not tight -- 
``merging" increases the distance but ``splitting" does not decrease it.
 
 It is time to examine how computational independence and encryption security 
 play with each other. We start by giving a definition of what an encryption 
 scheme actually \emph{is}:
\begin{definition}
A \emph{private-key encryption scheme} is a triple $\Pi=(G,E,D)$ where $G,E$ are PPT functions and $D$ is a $\poly$-time function such that for all $x \in \{0,1\}^n$,
if $k \leftarrow G(1^n)$, then $D(k,E(k,x))=x$.
\hfill\qed
\end{definition}

The following is the standard definition of encryption security against 
(uniform) efficient adversaries:
\begin{definition}[\hspace{-0.05em}\cite{G93} Definition 5]
\label{indist}
An encryption scheme $\Pi=(G,E,D)$ has \emph{indistinguishable encryptions} if for every efficiently samplable ensemble tuple $\{\eT_n\}=\{(\eX_n,\eY_n,\eZ_n)\}$ with $|\eX_n|=|\eY_n|$, and every PPT adversary  $\distone$ there is a negligible function $\nu$ such that for every  $n$,
\begin{equation*}
    \big|\Pr[\distone(\eZ_n,E(G(1^n),\eX_n))=1]- \Pr[\distone(\eZ_n,E(G(1^n),\eY_n))=1]\big|\le\nu(n).
  \qedeqn
\end{equation*}
\end{definition}
In the non-uniform setting, Definition~\ref{indist} reduces to the more 
familiar notion of indistinguishable encryptions (see \cite{G93}, Remark~12).
The following definition is instead based on independence, and can be seen as a 
rephrasing of Shannon's perfect security when statistical independence is 
replaced by its computational counterpart:
\begin{definition}
We say that $\Pi=(G,E,D)$ has \emph{independent ciphertexts} if for any efficiently samplable ensemble $\{\eX_n\}$, if $\eC_n=E(G(1^n),\eX_n)$, then $\{\eX_n\} \indep_c \{\eC_n\}$.
\qed
\end{definition}

Fay (see \cite{Fay14}, Theorem XX) proves that in the non-uniform setting an 
encryption scheme has indistinguishable messages if and only if it has independent 
ciphertexts. Extending the above equivalence in the nonuniform is left as an 
open problem in \cite{Fay14}. The rest of this section is devoted to proving 
that it holds
\begin{theorem}[Uniform Version of \cite{Fay14}, Theorem XX]
\label{unif}
With respect to uniform adversaries, an encryption scheme has indistinguishable  encryptions if and only if it has independent ciphertexts.
\hfill\qed
\end{theorem}

To prove the ``if'' direction, we will use an equivalent characterization of encryption security, namely \emph{semantic security}. In the uniform setting, this may be defined as follows:
\begin{definition}[\hspace{-0.05em}\cite{G93}, Definition 4] 
  An encryption scheme $\Pi=(G,E,D)$ is \emph{semantically secure} if there is a negligible function $\nu$ such that for every PPT $\distone$, there exists a PPT $\distone'$ such that for every efficiently samplable ensemble $\{\eX_n\}$, every poly-time function $h:\{0,1\}^*\rightarrow\{0,1\}^*$, every function $f:\{0,1\}^*\rightarrow\{0,1\}^*$ and every $n$
  \begin{equation*}
      \Pr[\distone(E(G(1^n),\eX_n),1^{|\eX_n|},h(\eX_n)) = f(\eX_n)]\le \Pr[\distone'(1^{|\eX_n|},h(\eX_n)) = f(\eX_n)] + \nu(n)
    \qedeqn
  \end{equation*}
\end{definition}
This definition is not identical to that of \cite{G93}, but is equivalent (in the uniform setting). See \cite{Goldreich} for details.

\begin{theorem}[\hspace{-0.05em}\cite{G93} Theorem 1]\label{thm:equi}
In the uniform setting, an encryption scheme is semantically secure if and only if it has
indistinguishable encryptions.
    \hfill\qed
\end{theorem}
Theorem~\ref{unif} follows from \Cref{thm:equi} and a couple of other auxiliary results
stating that indistinguishable encryptions imply independent ciphertexts, and that the latter imply
semantic security.

Theorem~\ref{unif} follows from the next two lemmas. 
\begin{lemma}
  \label{lemma:indepenctosec}
  With respect to uniform adversaries, if an encryption scheme has indistinguishable encryptions then it has independent ciphertexts.
\end{lemma}
\begin{proof}
Suppose $\Pi=(G,E,D)$ has indistinguishable encryptions. Define the ensemble $\{\eT_n\}=\{(\eX_n,\eY_n,\eX_n)\}$ where $\eY_n$ is independent of and distributed identically to $\eX_n$. It then follows immediately from the specialization of Definition~\ref{indist} to this distribution that $\{\eX_n\} \indep_c \{\eC_n\}$.
\end{proof}

\begin{lemma}
  \label{lemma:indchiphertosec}
With respect to uniform adversaries, if an encryption scheme has independent ciphertexts then it is semantically secure.
\end{lemma}
\begin{proof}
Suppose $\Pi=(G,E,D)$ has independent ciphertexts. Suppose $\{\eX_n\}$ is an efficiently samplable random variable ensemble. For $n >0$, let $\widehat{\eX}_n$ denote a distribution that is independent of $\eX_n$ but identically distributed. By assumption, for any PPT adversary $\distone$ and any $n$,\\[1.5ex]
\[
  \left|\Pr[\distone(\eX_n,E(G(1^n),\eX_n))=1]-\Pr[\distone(\eX_n,E(G(1^n),\widehat{\eX}_n))=1]\right|\le \nu(n)
\]
For any $\distone$ and poly-time $h$ and $f$, define $\disttwo(x,y)=1$ if $\distone(y,1^{|x|},h(x))=f(x)$ and $\disttwo(x,y)=0$ otherwise. Then $\disttwo$ is PPT, and applying the assumption to $\disttwo$ yields that for any PPT $\distone$, poly-time $h$ and $f$ and any $n$
\begin{equation*}
  \Big|\Pr[\distone(E(G(1^n),1^{|\eX_n|},h(\eX_n)))=f(\eX_n)]-\\
    \Pr[\distone(E(G(1^n),\widehat{\eX}_n),1^{|\eX_n|},h(\eX_n))=f(\eX_n)]\Big|\le \nu(n)
\end{equation*}
Now define $\distone'$ such that $\distone'(t,z)$ calls the sampler for $\{\eX_n\}$ on the $n$ such that $|\eX_n|=|t|$ using new uniform randomness to obtain $x$ and then returns $\distone(E(G(1^n),x),t,z)$. (Note that the samplability of $\{\eX_n\}$ implies that $\distone'$ can efficiently compute $n$ from $t$, as needed.) Then $\Pr[\distone'(1^{|\eX_n|},h(\eX_n))=f(\eX_n)]=\Pr[\distone(E(G(1^n),\widehat{\eX}_n),1^{|\eX_n|},h(\eX_n))=f(\eX_n)]$
(where $\widehat{\eX}_n$ is independent of $\eX_n$ but identically distributed) and so
\begin{equation*}
\Big|\Pr[\distone(E(G(1^n),1^{|\eX_n|},h(\eX_n)))=f(\eX_n)]-\\
  \Pr[\distone'(1^{|\eX_n|},h(\eX_n))=f(\eX_n)]\Big|\le \nu(n)
\end{equation*}
which means
\begin{equation*}
  \Pr[\distone(E(G(1^n),1^{|\eX_n|},h(\eX_n)))=f(\eX_n)]\le\\
  \Pr[\distone'(1^{|\eX_n|},h(\eX_n))=f(\eX_n)]+\nu(n)
\end{equation*}
so that $\Pi$ is semantically secure.
\end{proof}


\section{On the Limitations on \CPSL}
\label{sec:limitations}
 The logic we have introduced in this paper and the class of programs that can 
be treated through it are limited, at least apparently: on the one hand we are 
not able to deal with programs with loops, on the other 
hand both forms of implication are not available. This section is designed to
discuss and justify the reasons for these limitations.

The absence of iterative constructs in the language of programs 
is not directly due to
complexity limitations. For example, 
for-loops would not affect the crucial 
result that programs denote endofunctions on the space of efficiently samplable 
distributions. What would be difficult to manage, however, is the evolution of 
the \emph{advantage} of the 
adversary along a loop: while this advantage
is guaranteed not to grow too much along sequential compositions,
the same is not true for iterative constructs. And, it would thus
be necessary to explicitly keep track of how the adversary's advantage
evolves during the computation and not just of the fact that this advantage remains negligible. 
Developing a logic powerful enough to make all this possible is outside the 
scope of this work. On the other hand, some interesting proofs can anyway be 
given, even without loops, as shown in \Cref{sec:examples} above.

\emph{Kripke Monotinicity} is a key property of \BI and \PSL that
fails in \CPSL because of complexity constraints.
This property states that by looking at any formula $\phi$, it is 
possible to identify a non-trivial overapproximation to the set of variables
influencing the semantics of $\phi$, namely the free variables of that formula. 
This observation is fundamental for showing the following property: if 
$\fv{\fone}\subseteq \dom(\stone)$ and $\stone \ext \sttwo$, then $\sr \stone 
\fone \Leftrightarrow \sr \sttwo \fone$. 
This property, in turn, is crucial for the proof of soundness for the frame 
rule. However, the proof of this property cannot be adapted to \CPSL when 
$\fone$ is an implication: again, it is not known 
whether the class of efficiently samplable distributions is closed 
under general couplings. This is why we restricted ourselves to the conjunctive 
fragment of \PSL.


\revision{
Finally, the use of typed formulas is mainly aimed 
at making the set of variables on which a formula is interpreted
an explicit part of the formula itself. We believe that,
for our aim, typed formulas are convenient:
they make binary propositions reflexive,
they render the semantics of the separating conjunction less ambiguous, and
they allow a natural formulation of the $\RFRAME$ rule,
where, in particular, the way in which the state is divided
by the separating conjunction is induced by the formula and is
validated by the type system.
However, we also believe that using unlabeled formulas
would not affect the foundations of our work.
}


\section{Related Work}\label{sec:related}
Yao \cite{Yao} implicitly introduced the notion of computational independence in a cryptographic setting via a characterization of encryption security using a computational version of mutual information, analogous to Shannon's characterization of perfect secrecy \cite{Sha49}. The definition of \cite{Yao} was quite complex, relying on a characterization of mutual entropy via compression, and due to subsequent and more direct characterizations given in \cite{GoldwasserMicali84}, this earlier characterization was largely forgotten, although \cite{MRS88} addressed the equivalence of the various definitions, including Yao's. Subsequently, the only work addressing computational independence and its application in cryptographic settings is \cite{Fay14}.

An early contribution to computationally sound systems for cryptographic security is \cite{AR07}. This work did not provide a general account of computational indistinguishability but did allow the expression of equivalences between certain expressions involving symbolic encryption. One drawback of the system is that soundness requires a syntactic restriction on expressions with key cycles. Subsequent work addressed this issue by modifying the computational model with a stronger form of encryption (KDM security) \cite{BRS02,ABHS09}, adding explicit detection of key cycles to the symbolic semantics \cite{Lau02} or modifying the symbolic model to provide a co-inductive semantics \cite{Mic10}. Except for \cite{Mic19}, which includes a symbolic category for pseudorandom keys, none of these works explicitly model pseudorandomness. In a different direction, works which explicitly axiomatize computational indistinguishability include \cite{IK06} and \cite{BDKL10}.

\revision{Our approach provides benefits similar to those of earlier computationally sound formal systems. In particular, there is no need to explicitly model a computationally bounded adversary nor to prove security via a computational reduction. This is in contrast with approaches motivated by Halevi's influential proposal for a programme of formalization of the game-hopping paradigm \cite{Halevi}. While the scope of constructions we are currently able to consider is more constrained than in these approaches, we complement approaches based on the model of \cite{AR07} by providing a framework with more general forms of construction.}

The program logic of~\cite{PSL} was the first to use the separating conjunction for capturing probabilistic independence. In this work, formulas are expressed in the logic of bunched implication (\BI,~\cite{OHearnPym99}) and are in interpreted as resource monoids~\cite{Pym04}. In the same spirit of~\cite{PSL}, the separating conjunction was used in~\cite{Bao} to model negative dependence.

\section{Conclusion}

The aim of this work is to show how Yao's computational independence can be exploited to simplify certain cryptographic proofs by hiding the role of the adversary. An essential tool in all this is separation logic, to which this paper attributes semantics in a new way, following the work of Barthe et al., but at the same time injecting computational indistinguishability into the playground.

Much remains to be done, especially with respect to the underlying class of programs. As already mentioned, dealing with the presence of loops does not seem trivial and would certainly require substantial changes to the logic and/or in its semantics, in order to keep the adversary's advantage under control. On the other hand, this would potentially allow extending the applicability of the presented technique to more complex proofs, normally conducted via the so-called hybrid arguments. The authors, for this reason, are actively dealing with this problem, which however remains challenging.

\bibliography{biblio}



\end{document}